\pdfoutput=1
\documentclass[a4paper,onecolumn]{article}
\usepackage[top=2.5cm,bottom=2.5cm,left=2.5cm,right=2.5cm]{geometry}  
\usepackage{authblk}
\usepackage{bbm}
\usepackage{braket}
\usepackage{graphicx}
\usepackage{dsfont}
\usepackage[T1]{fontenc}
\usepackage[utf8]{inputenc}
\usepackage{float}
\usepackage{braket}
\usepackage{indentfirst}
\usepackage{enumerate}
\usepackage{url}
\usepackage[dvipsnames,x11names]{xcolor} 	
\usepackage[stable]{footmisc}
\usepackage{fancyhdr}
\usepackage[all]{xy}
\usepackage{physics}
\usepackage{graphicx}
\usepackage{tikz}
\usepackage{framed}
\usepackage{ytableau}
\usepackage{youngtab}
\usepackage{mathtools}
\usepackage{enumerate}
\usepackage[title]{appendix}
\usepackage{textcomp}
\usepackage{newcent}
\usepackage[sc,noBBpl]{mathpazo}
\usepackage{amsmath,amsfonts,amssymb,amsthm}
\usepackage[mathscr]{eucal}
\usepackage{tcolorbox}
\usepackage[colorlinks=true,citecolor=PineGreen,linkcolor=RoyalBlue,urlcolor=Violet]{hyperref}
\usepackage[labelfont=bf]{caption} 
\usepackage{subcaption}
\usepackage[export]{adjustbox}
\usepackage[
backend=biber,
style=nature,
sorting=none
]{biblatex}
\AtEveryBibitem{\clearfield{month}}
\AtEveryBibitem{\clearfield{day}}
\AtEveryBibitem{\clearfield{issue}}
\theoremstyle{plain}
\newtheorem{theorem}{Theorem}
\newtheorem{lemma}[theorem]{Lemma}
\newtheorem{proposition}[theorem]{Proposition}
\newtheorem{corollary}[theorem]{Corollary}

\newtheorem{definition}[theorem]{Definition}
\newtheorem{fact}[theorem]{Fact}

\newtheoremstyle{note}{\topsep}{\topsep}{\slshape}{}{\scshape}{}{ }{}
\theoremstyle{note}

\newtheorem{remark}[theorem]{Remark}
\newtheorem{example}[theorem]{Example}

\newcommand{\id}{\mathds{1}}

\newcommand{\<}{\langle}
\renewcommand{\>}{\rangle}

\newcommand\be{\begin{equation}}
\newcommand\ee{\end{equation}}
\newcommand\bea{\begin{array}}
	\newcommand\eea{\end{array}}
\newcommand\ben{\begin{eqnarray}}
\newcommand\een{\end{eqnarray}}

\newcommand{\SU}{\mathcal{SU}}
\newcommand{\s}{\mathcal{S}}

\newcommand\bei{\begin{itemize}}
	\newcommand\eei{\end{itemize}}
\newcommand\bee{\begin{enumerate}}
	\newcommand\eee{\end{enumerate}}

          \newcommand\hlight[1]{\tikz[overlay, remember picture,baseline=-\the\dimexpr\fontdimen22\textfont2\relax]\node[rectangle,fill=white!50,rounded corners,fill opacity = 0.2,draw,thick,text opacity =1] {$#1$};}

\newcommand{\G}{\mathcal{G}}
\newcommand{\h}{\mathcal{H}}

\begin{document}
\title{Group-Adapted Irreducible Matrix Units for the Walled Brauer Algebra}
\author[1,3]{Michał Studziński}
\author[1]{Tomasz Młynik}
\author[2]{Marek Mozrzymas}
\author[1,3]{Michał Horodecki}
\author[4,5,6]{Dmitry Grinko}
\affil[1]{\small{\textit{Institute of Theoretical Physics and Astrophysics, Faculty of Mathematics, Physics, and Informatics, University of Gda\'nsk, Wita Stwosza 57, 80-308 Gda\'nsk, Poland}}}
\affil[2]{\small{\textit{Institute for Theoretical Physics, University of Wrocław, plac Maxa Borna 9, 50-204 Wrocław, Poland}}}
\affil[3]{\small{\textit{International Centre for Theory of Quantum Technologies, University of Gdańsk, Jana Bażyńskiego 1A, 80-309 Gdańsk, Poland}}}
\affil[4]{\small{\textit{QuSoft, Amsterdam, The Netherlands}}}
\affil[5]{\small{\textit{Institute for Logic, Language and Computation, University of Amsterdam, Amsterdam, The Netherlands}}}
\affil[6]{\small{\textit{Korteweg-de Vries Institute for Mathematics, University of Amsterdam, The Netherlands}}}
\date{}
\maketitle
\begin{abstract} 
  This paper investigates the representation theory of the algebra of partially transposed permutation operators, $\mathcal{A}^d_{p,p}$, which provides a matrix representation for the abstract walled Brauer algebra. This algebra has recently gained significant attention due to its relevance in quantum information theory, particularly in the efficient quantum circuit implementation of the mixed Schur-Weyl transform.

   In contrast to previous Gelfand-Tsetlin type approaches, our main technical contribution is the explicit construction of irreducible matrix units in the second-highest ideal that are group-adapted to the action of $\mathbb{C}[\s_p]\times \mathbb{C}[\s_p]$ subalgebra, where $\s_p$ is the symmetric group. This approach suggests a recursive method for constructing irreducible matrix units in the remaining ideals of the algebra. The framework is general and applies to systems with arbitrary numbers of components and local dimensions. 

   In addition, we present a complementary construction method based on tensor networks of Clebsch-Gordan coefficients of the unitary group.  This approach enables the construction of all group-adapted irreducible matrix units, but requires knowledge of certain Littlewood-Richardson coefficients. This method can be successfully applied for a reasonably small number of particles with the support of dedicated software.
   
   The obtained results are applied to a special class of operators motivated by the mathematical formalism appearing in all variants of the port-based teleportation protocols through the mixed Schur-Weyl duality. We demonstrate that the given irreducible matrix units are, in fact, eigenoperators for the considered class.
   
\end{abstract}
\tableofcontents
\section{Introduction}
\subsection{Motivation and background}
Symmetry serves as a fundamental principle in both physics and mathematics, playing a crucial role in problem-solving and the formulation of physical laws. 
Two primary types of symmetry—unitary group symmetries and permutational symmetry—are intricately connected. Specifically, the processes of permuting identical particles and applying identical unitary transformations to each particle commute with one another.
This interplay is formalized in the famous Schur–Weyl duality \cite{fulton_harris,Goodman}, a pivotal concept that has found widespread applications across quantum information science, quantum computing theory, many-body physics, and matrix quantum mechanics.

In quantum information, which is the focus of our work, the Schur–Weyl duality serves as a particularly valuable tool when dealing with multiple identical copies of a quantum state or when performing parallel unitary operations across multiple systems. Technically, this means that the operation of permuting the particles commutes with a local basis change on each system. In the simplest case, one has the relation $[U\otimes U, SWAP]=0$, where $SWAP=\sum_{i,j}|ij\>\<ji|$ is the operation of permuting two systems, and $U$ is a unitary operation. This of course can be generalised for more systems. This commutation relation implies that the irreducible representations for both, the symmetric group and diagonal action of the unitary group can be obtained simultaneously. Then instead of working with a particular problem on the whole space, one can deal with it on every irreducible block separately, getting potentially simpler and more elegant description. 
The Schur-Weyl duality has found a deep impact on quantum computing science, due to its efficient quantum circuit implementation \cite{Krovi2019efficienthigh,bacon2006efficient,bacon2007quantum,Kirby_2018,nguyen2023mixedschurtransformefficient}. This led, for example, to recent important results regarding boosting quantum computing and quantum machine learning for systems with the symmetries \cite{PRXQuantum.4.020327}.
Applications of the Schur–Weyl duality, however, reach far more than quantum computing and has an impact on fundamental concepts in quantum information. We mention here results including qubit quantum cloning \cite{_wikli_ski_2012}, the theory of quantum gates \cite{bacon2006efficient, Chiribella_2016, PhysRevLett.114.120504}, quantum error correction codes \cite{Junge_2005, PhysRevResearch.4.023107}, entanglement distillation \cite{Czechlewski_2012}, entropy estimation problem~\cite{acharya2020estimating}, higher-order quantum computing \cite{Quintino2022deterministic,Ebler_2023}, quantum state tomography \cite{KEYL_2006,Haah_2017,10.1145/3055399.3055454,o2016efficient}, quantum computing~\cite{jordan2009permutationalquantumcomputing,Gross2021}, state estimation \cite{Marvian_Iman_Spekkens_Robert_W}, symmetry reduction in SDP problems \cite{huber2024refutingspectralcompatibilityquantum}, generating operator inequalities and identities in several matrix variables \cite{Huber_2021}, and many others.

In this paper, we focus on studying a variant of the Schur-Weyl duality, called the mixed Schur-Weyl duality. This duality differs from the previous one by involving into game concepts of the partial transposition $t$ and complex conjugation $\overline{U}$ of an arbitrary unitary operation $U\in \mathcal{U}(d)$. When one applies complex conjugation over the last $q$ systems to diagonal action of $U^{\otimes (p+q)}$ on the representation space $(\mathbb{C}^d)^{\otimes (p+q)}$, getting $U^{\otimes p}\otimes \overline{U}^{\otimes q}$, can show it mutually commutes with elements from the algebra of partially transposed permutation operators $\mathcal{A}_{p,q}^{d}$ with respect to the last $q$ systems~\cite{KOIKE198957,BENKART1994529, Halverson1996CharactersOT,PhysRevA.63.042111,Moz1}. This means we have the relation $[U^{\otimes p}\otimes \overline{U}^{\otimes q}, X]=0$, for every $X\in \mathcal{A}_{p,q}^{d}$. It can be shown and will be discussed later in this manuscript that the algebra $\mathcal{A}_{p,q}^{d}$ stands for matrix a representation of diagrammatic algebra called the walled Brauer algebra~\cite{VGTuraev_1990,KOIKE198957,BENKART1994529,BEN96,bulgakova:tel-02554375,Cox1}. It is easy to see that the mixed Schur-Weyl duality reduces to the standard Schur-Weyl duality when one imposes $p=0$ or $q=0$.

One of the key feature of the mixed Schur-Weyl duality, underlying its importance for quantum information science is its relation to entanglement. Namely, in the simplest case of a bipartite system $AB$, partially transposed  operator $SWAP^{t_B}=(\sum_{i,j}|ij\>\<ji|)^{t_B}$, so element of $\mathcal{A}_{1,1}^{d}$,  is proportional to the projector on the maximally entangled state $|\psi^+\>\<\psi^+|=(1/d)\sum_{i,j} |ii\>\<jj|$, which is $U\otimes \overline{U}$ invariant. Now, allowing for more particles in our system and a larger number of partial transpositions we easily generalize this observation. These connections, however, appear frequently in various problems in quantum information and have attracted significant attention recently. This is partly due to the development of a mathematical apparatus that enables the efficient analysis of the mixed Schur-Weyl duality and the algebra $\mathcal{A}_{p,q}^{d}$ on the space $(\mathbb{C}^{d})^{\otimes (p+q)}$~\cite{MozJPA,StudzinskiIEEE22,Mozrzymas2024,grinko2023gelfandtsetlinbasispartiallytransposed}. One of the most recent important applications are structure descriptions for variants of quantum teleportation, called port-based teleportation~\cite{studzinski2017port, mozrzymas2018optimal, leditzky2020optimality, christandl2021asymptotic,Wang2016}, multi-port-based teleportation~\cite{StudzinskiIEEE22,mozrzymas2021optimal}, together with their efficient implementation as quantum circuits~\cite{grinko2024efficientquantumcircuitsportbased,fei2023efficientquantumalgorithmportbased,PRXQuantum.5.030354}. In fact, the latter relies on the fact that the mixed Schur transform can be efficiently simulable on quantum circuits~\cite{nguyen2023mixedschurtransformefficient,grinko2024efficientquantumcircuitsportbased,fei2023efficientquantumalgorithmportbased}. Except for the teleportation, this theory plays a role in entanglement detection~\cite{PhysRevA.63.042111,Bardet2020,balanzo2024positive}, equivariant quantum circuits~\cite{marvian2021qudit}, theory of universal cloning machines~\cite{PhysRevA.89.052322,Nechita2023,Nechita2021}, theory of higher-order quantum operations~\cite{Quintino2022deterministic}, integrable antiferromagnetic systems~\cite{Candu2011}, high-energy physics~\cite{2008PhRvD..78l6003K,1126-6708-2007-11-078,Kimura2010,Kimura:2016bzo}, symmetry reduction is SDP problems~\cite{grinko2023linearprog}, an efficient algorithm for unitary mixed Schur sampling~\cite{marcinska24}. 

Due to various important applications of the mixed Schur-Weyl duality and representation theory of the algebra $\mathcal{A}_{p,q}^{d}$ further mathematical studies on them are at the center of mathematical developments for quantum information.  In this paper, we deliver novel mathematical techniques for the mixed Schur-Weyl duality that can be potentially used for various quantum information-motivated tasks.

\subsection{Summary of the main results}
In addition to the central result outlined in the abstract, we have obtained several supplementary findings that hold independent significance for the research community. For the reader's convenience, we provide a summary of these results below.
\begin{itemize}
    \item We derived expressions for sandwiching elements from the algebra $\mathbb{C}[\s_p]\times \mathbb{C}[\s_p]$ by the operators generating the highest and the second highest ideal of the algebra $\mathcal{A}^{d}_{p,p}$ of partially transposed permutation operators. In particular, we show how such an operation of sandwiching connects with the partial trace calculated in the natural representation of the $\mathcal{A}^{d}_{p,p}$. Additionally, we proved a few new trace rules from compositions of the elements from $\mathcal{A}^{d}_{p,p}$ and $\mathbb{C}[\s_p]\times \mathbb{C}[\s_p]$. The main results concerning the above are contained in Fact~\ref{F:concatenation} and Theorem~\ref{thm:wba_element} in Section~\ref{sec:facts_of_alg_of_part_op}, and used by us later in this manuscript.
    \item  We construct a family of irreducible matrix units within the second-highest ideal \(\mathcal{M}^{(p-1)}\) that are group-adapted to the action of $\mathbb{C}[\s_p]\times \mathbb{C}[\s_p]$ subalgebra, and thoroughly establish their properties - these are the main result of our work contained in Theorem~\ref{thm:lower_ideal_operator_G} and Theorem~\ref{thm:irrepsMp1}. This substantially extends our knowledge, since previous constructions~\cite{grinko2023gelfandtsetlinbasispartiallytransposed} did not take into account symmetry appearing on both sides of the wall. The results are obtained through a two-step approach. We start from Proposition~\ref{prop:Linspan} from Section~\ref{sec:family_of_spanning_operators_in_ideals}, where we introduce two classes of operators that form the linear span of the ideals \(\mathcal{M}^{(p)}\) and \(\mathcal{M}^{(p-1)}\), and prove Lemma~\ref{lemma:composition_of_fs} and~\ref{lemma:ad+b=m_mu}, establishing their composition and trace rules. These findings are subsequently applied in Section~\ref{sec:ideal_m-1} to construct and analyze the irreducible matrix units within \(\mathcal{M}^{(p-1)}\). Additionally, in Lemma~\ref{lemma:operator_H_generates_ideal_mp-1} we show how to express the partially transposed permutation operator $V^{(p-1)}$, generating the ideal \(\mathcal{M}^{(p-1)}\), in terms of the derived irreducible matrix basis. To the best of our knowledge, all results concerning the ideal \(\mathcal{M}^{(p-1)}\) are novel and represent a significant extension of prior work~\cite{StudzinskiIEEE22}. 

    \item In Section~\ref{sec:tn_rep_mat_units}, we also present an complementary construction of matrix units of $\mathcal{A}^{d}_{p,q}$ based on tensor networks of Clebsch--Gordan coefficients of the unitary group, see Figure~\ref{fig:Matrix_units_new}. This approach is ``dual'' to the results outlined above in the sense that this alternative approach is based on the representation theory of the unitary group, while the first one is based on symmetric group. Using this method we can obtain numerically the matrix elements of the partially transposed permutation generator $V(\sigma_n)'$ in all irreps of all ideals, see Figure~\ref{fig:A333_contraction_all_irreps}. A list of matrix generators for the individual irreducible representations, illustrated on the example of the algebra $\mathcal{A}_{3,3}^3$, is provided in the Appendix~\ref{app:tensor_network_example_gens}.
    
    \item Using the constructed irreducible matrix units, in Theorem~\ref{thm:spec_of_rho_Gp} from Section \ref{sec:spectra_of_wba_elements},  we find matrix elements of the averaged (twirled) partially transposed permutation operators $V^{(p)},V^{(p-1)}$, generating the ideals \(\mathcal{M}^{(p)}\) and \(\mathcal{M}^{(p-1)}\),  with respect to the cross action of $\s_p \times \s_p$:
    \begin{align}\label{eqn:rho_operator0}
    \rho(p) &:= \frac{1}{(p!)^2}\sum_{\sigma_1,\sigma_2 \in \s_p} \Big(V_{\sigma_1}\otimes V_{\sigma_2}\Big) V^{(p)} \Big(V_{\sigma_{1}^{-1}}\otimes V_{\sigma_{2}^{-1}}\Big),\\
    \rho(p-1) &:= \frac{1}{(p!)^2}\sum_{\sigma_1,\sigma_2 \in \s_p} \Big(V_{\sigma_{1}}\otimes V_{\sigma_{2}}\Big) V^{(p-1)} \Big(V_{\sigma_{1}^{-1}}\otimes V_{\sigma_{2}^{-1}}\Big),
\end{align}
where $V_{\sigma_{1}},V_{\sigma_{2}}$ are operators that permute $p$ systems according to permutation $\sigma_1, \sigma_2\in \s_p$ respectively. These operators are the natural matrix representations of the elements from the walled Brauer algebra $\mathcal{B}^d_{p,p}$ and are natural generalizations of the port-based teleportation operators.  In the same section, in Lemma~\ref{fact:twirl_of_general_operators}, we proved a general statement about the relation between the twirl operation and the derived irreducible basis from Section~\ref{sec:ideal_m-1}. This proves the usefulness of the presented tools for practical calculations for problems with symmetries suggested by the walled Brauer algebra. 
    \item Having the construction of the irreducible matrix units in maximal ideal $\mathcal{M}^{(p)}$ and $\mathcal{M}^{(p-1)}$ together with proved twirl rules, we provide an explicit example how our methods work in practice. In Section \ref{sec:example} for the fixed number of particles $p=3$ and local dimension $d=3$ we derived the matrix elements of the operators $\rho(3)$ and $\rho(2)$ analytically. We showed that these operators are block-diagonal in the constructed irreducible basis, we evaluated corresponding eigenvalues with their mulitpilicities. The results are summarised in the Figure~\ref{fig:structure_diagram}. Additionally, we approach this problem from the perspective of the results presented in Section~\ref{sec:tn_rep_mat_units}, where we obtain the matrix elements of $\rho(1)$ using a numerical procedure. These results are collected in Appendix~\ref{app:tensor_network_example_rhos}. Here, by “numerical” we do not mean a brute-force computation, but a structured, representation-theoretic approach that leads to explicit matrix elements.
\end{itemize}

\section{Elements of group representation theory}\label{sec:elements_of_rep_theory}
This section provides the mathematical foundations relevant to our further studies and makes our paper self-consistent.  We begin by exploring Young diagrams and their essential properties, which form a cornerstone in the study of combinatorial representation theory. We then briefly examine the representation theory of the symmetric group and the Schur-Weyl duality, focusing on how these concepts describe the relationship between irreducible representations of the symmetric group and the unitary group. Lastly, we introduce the walled Brauer algebra and its relationship to the concept of partial transposition and further to the algebra of partially transposed permutation operators.

\subsection{Partitions and Young frames}
For a given  natural number $p$ we define its partition $\mu$ denoted by $\mu\vdash p$ as a $p-$tuple of positive numbers $\mu = (\mu_1,\mu_2,\ldots,\mu_r)$, such that
\begin{align}\label{eq:partition}
    \mu_1\geq\mu_2\geq\ldots\geq\mu_r\geq 0, \qquad \sum_{i=1}^{r}\mu_i = p.
\end{align}
Every partition can be visualized as a Young diagram which can be associated with the array formed by $p$ boxes with left-justified rows. With every diagram, we can associate a set of coordinates of its boxes:
\begin{align}
[\mu]=\{(i,j): 1\leq i\leq r, 1\leq j\leq \mu_i\}.
\end{align}
In this paper by $\mu,\nu$ and $\alpha,\alpha'$ we associate Young frames with $p$ and $p-1$ boxes respectively. The length of the first column in a given Young frame $\mu$ (equivalently number of summed elements in~\eqref{eq:partition}) is denoted as a height $\operatorname{ht}(\mu)$, see Figure~\ref{fig:y_diagram}. 

\begin{figure}[h!]
\centering\includegraphics[width=1\textwidth]{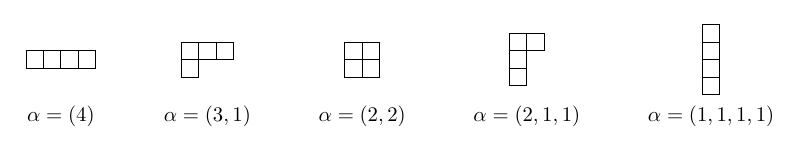}
    \caption{Graphic presents five possible Young frames for $p=4$, which also corresponds to all possible abstract irreducible representations of $\s_4$. Considering representation space $(\mathbb{C}^d)^{\otimes 4}$, there appear only irreps for which the height of corresponding Young frames is no larger than $d$. For  example, if $d=2$, irreps $(2,1,1), (1,1,1,1)$ do not appear since $\operatorname{ht}((2,1,1)) = 3$, and $\operatorname{ht}((1,1,1,1))=4$ respectively.}
    \label{fig:y_diagram}
\end{figure}

By writing $\mu = \alpha +\Box$ we denote a Young diagram $\mu\vdash p$ obtained from a diagram $\alpha \vdash (p-1)$ by adding a one box. Inversely, by $\alpha = \mu - \Box$ we denote a Young diagram $\alpha \vdash (p-1)$ by subtracting a single box from a Young diagram $\mu \vdash p$. This procedure is illustrated in Figure~\ref{fig:y_diag_rel}. 

\begin{figure}[htbp]
\centering\includegraphics[width=0.7\textwidth]{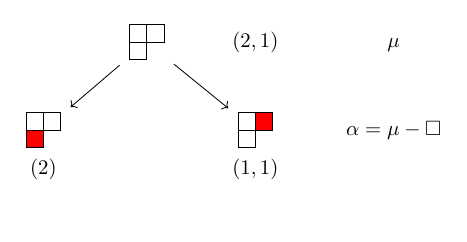}
    \caption{Graphic presents possible Young diagrams $\alpha\vdash 2$, which can be obtained from a diagram $\mu=(2,1)$ by removing a single box, depicted here in red. In this particular case, by writing $\alpha =\mu-\Box$ only possible $\alpha$ are $(2), (1,1)$. In the same manner, we define adding a box into a Young diagram.}
    \label{fig:y_diag_rel}
\end{figure}

With a given Young diagram $\mu\vdash p$ we can associate a standard Young tableaux (SYT) which is obtained by filling in
the boxes of the Young diagram with numbers $1,2,\ldots,p$. Numbers in the boxes must strictly increase from left to right in every row and from top to bottom in every column. The number of all standard Young tableaux for the fixed shape we denoted by $d_\mu$ and can be evaluated by using certain combinatorial expressions like the hook-length formula \cite{FultonSchur, ceccherini2010representation},
\begin{align}\label{eqn:equation_on_d}
    d_\mu:=|\operatorname{SYT}(\mu)| = \frac{(\mu_1+\ldots + \mu_r)!}{\prod_{(i,j)\in\mu} \operatorname{hook}(i,j)}.
\end{align}
In the above, the quantity $\operatorname{hook}(i,j)$ is the number of boxes in the same row to the right of $(i,j)\in [\lambda] +$ number of boxes in the same column below $(i,j)+1$. 

With a given Young diagram $\mu\vdash p$ we can also associate a semi-standard Young tableaux (SSYT). These objects are obtained by filling in the boxes of the Young diagram with natural numbers $1,2,\ldots,d$ for some $d$. In every row numbers must be arranged in non-decreasing order from the left to the right and strictly increasing order in every column from the top to the bottom. The number of all semi-standard Young frames for a given tableau is denoted by $m_\mu$ and it can also be evaluated by combinatorial rules like the hook-content formula \cite{FultonSchur, ceccherini2010representation},
\begin{align}\label{eqn:equation_on_m}
    m_\mu:=|\operatorname{SSYT}(\mu,d)| = \prod_{1\leq i<j\leq d}\frac{\mu_i - \mu_j+j-i}{j-i},
\end{align}
It is easy to see that when $d < p $ not all semi-standard Young frames exist. Namely, we have to exclude all Young frames whose height is greater than $d$. The concepts described above are closely connected with the representation theory of the symmetric group $\s_p$ and the unitary group $\SU(d)$ connected to each other by the Schur-Weyl duality. We discuss this in the next section.

\subsection{The Schur-Weyl duality}
\label{sec:SW}

We define the \textit{permutational representation} $V: \s_p \rightarrow \operatorname{Hom}\left((\mathbb{C}^{d})^{\otimes p}\right)$ of the symmetric group $\s_p$ on the Hilbert space $\mathcal{H} = (\mathbb{C}^d)^{\otimes p}$. The elements of $V(\s_p)$ act by permuting vectors in $(\mathbb{C}^d)^{\otimes p}$ according to a permutation $\sigma \in \s_p$, as follows: 
\begin{align} \label{eq:actionS_k} 
V_\sigma (\ket{v_1}\otimes \ket{v_2} \otimes \dots \otimes \ket{v_p}) := \ket{v_{\sigma^{-1}(1)}}\otimes \ket{v_{\sigma^{-1}(2)}} \otimes \dots \otimes \ket{v_{\sigma^{-1}(p)}}. 
\end{align} This representation $V(\s_p)$ naturally extends to the representation of the group algebra $\mathbb{C}[\s_p]:=\operatorname{span}_{\mathbb{C}}\{V_\sigma:\sigma \in \s_p\}$. Similarly, we define a diagonal action $U^{\otimes p}: \SU(d) \rightarrow \operatorname{Hom}\left((\mathbb{C}^d)^{\otimes p}\right)$ for elements $U \in \SU(d)$. The action of $U^{\otimes p}$ on $(\mathbb{C}^d)^{\otimes p}$ is given by: 
\begin{align} \label{eq:actionUk}
U^{\otimes p}(\ket{v_1}\otimes \ket{v_2} \otimes \dots \otimes \ket{v_p}) := U\ket{v_1}\otimes U\ket{v_2}\otimes \dots \otimes U\ket{v_p}. 
\end{align}

It is straightforward to verify that the actions defined in~\eqref{eq:actionS_k} and \eqref{eq:actionUk} commute. Consequently, there exists a basis in which both the tensor product space $(\mathbb{C}^d)^{\otimes p}$, and the operators $V_\sigma$ and $U^{\otimes p}$, decompose as: \begin{align} \label{wedd} (\mathbb{C}^d)^{\otimes p} &= \bigoplus_{\substack{\mu \vdash p \\ \operatorname{ht}(\mu) \leq d}} \mathcal{U}_{\mu} \otimes \mathcal{S}_{\mu},\\ V_\sigma &= \bigoplus_{\substack{\mu \vdash p \\ \operatorname{ht}(\mu) \leq d}} \id_{\mathcal{U}_\mu} \otimes \varphi^{\mu}(\sigma), \\ U^{\otimes p} &= \bigoplus_{\substack{\mu \vdash p \\ \operatorname{ht}(\mu) \leq d}} U_{\mu} \otimes \id_{\mathcal{S}_\mu}. 
\end{align} The direct sum is taken over all Young diagrams of $p$ boxes with a height no larger than the local dimension $d$, i.e. $\operatorname{ht}(\mu)\leq d$. The symbols $\varphi^{\mu}(\sigma)$ and $U_{\mu}$ denote the irreducible representations (irreps) of $V_\sigma$ and $U^{\otimes p}$, respectively.

From the decomposition~\eqref{wedd}, we deduce that for a given irrep $\mu$ of $\s_n$, the space $\mathcal{U}_{\mu}$ serves as the multiplicity space, with dimension $m_{\mu}$ (the multiplicity of irrep $\mu$ is the number of semi-standard Young tableaux for integer $d$). On the other hand, the space $\mathcal{S}_{\mu}$ is the representation space, with dimension $d_{\mu}$ (the dimension of irrep $\mu$ is the number of standard Young tableaux). This implies that the permutations act non-trivially only on the space $\mathcal{S}_{\mu}$.

By Schur’s Lemma, any operator commuting with $U^{\otimes p}$ can be expressed as a linear combination of the irreducible matrix units $E^{\mu}_{ij}$, given by: 
\begin{align} 
E^{\mu}_{ij} = \id_{\mathcal{U}_\mu} \otimes \ketbra{\mu, i_\mu}{\mu, j_\mu}_{\mathcal{S}_\mu}, 
\label{eq:def_E}
\end{align} where $i,j \in {1, \dots, m_{\mu}}$, and ${\ket{\mu, i_\mu}}$ forms an orthonormal basis of $\mathcal{S}_{\mu}$. The irreducible matrix units~\eqref{eq:def_E} satisfy the following composition rules: 
\begin{align} \label{eqn:composition_of_e_operators}
E_{ij}^\mu E_{kl}^\nu = \delta^{\mu \nu} \delta_{jk} E_{il}^{\mu}, \qquad \tr(E_{ij}^\mu) = m_{\mu} \delta_{ij}. 
\end{align} 
It is evident that the diagonal operators of the form $E^{\mu}_{ij}$ are projectors of rank $m_\mu$. Moreover, using diagonal operators from~\eqref{eq:def_E}, we can define Young projectors on the irreducible components labeled by $\mu$:
\begin{align}
\label{eq:Young}
P_\mu=\sum_{i=1}^{d_\mu} E^\mu_{ii},\qquad P_\mu P_\nu=\delta^{\mu \nu}P_\mu,\qquad \tr (P_\mu)=m_\mu d_\mu.
\end{align}
Furthermore, we define the natural representation of the irreducible matrix units given in~\eqref{eq:def_E} for the irrep $\mu \vdash p$ of $\s_p$ as: 
\begin{align} \label{eqn:basis_Eij}
E_{ij}^{\mu} = \frac{d_\mu}{p!} \sum_{\sigma \in \s_p} \varphi_{ji}^\mu (\sigma^{-1}) V_\sigma, 
\end{align} 
where $\varphi_{ji}^\mu(\sigma^{-1})$ is the matrix element of the irrep $\sigma^{-1} \in \s_p$, and $i,j = 1, \dots, d_\mu$, with $d_\mu$ being the dimension of the irrep $\mu$.  
Using equations~\eqref{eqn:basis_Eij} one can write an arbitrary permutation operator $V_\sigma$ as
\begin{align}
\label{eq:VasE}
\forall \ \sigma\in \s_p \qquad V_\sigma=\sum_{\alpha}\sum_{ij}\varphi^{\mu}_{ij}(\sigma)E^{\mu}_{ij}.
\end{align}
The above equation, together with~\eqref{eqn:composition_of_e_operators} allows us to write explicitly the left and the right action of an arbitrary operator $V_\sigma$ on irreducible matrix units~\eqref{eqn:basis_Eij}:
\begin{align}
\label{eq:actionVonE}
V_\sigma E_{ij}^\mu=\sum_k\varphi^{\mu}_{ki}(\sigma)E_{kj}^{\mu},\qquad E_{ij}^\mu V_\sigma=\sum_k \varphi^{\mu}_{jk}(\sigma)E_{ik}^\mu.
\end{align}
The natural representation of $E^{\mu}_{ij}$ from~\eqref{eqn:basis_Eij} together with equations~\eqref{eq:VasE}, \eqref{eq:actionVonE}  will be extensively used later.

So far, the vectors in~\eqref{eq:def_E} have been arbitrary, as long as they span their respective subspaces. Specifically, for the basis in \eqref{eq:def_E}, we can use the Young-Yamanouchi basis \cite{young1931quantitative, yamanouchi1937construction} of $\mathcal{S}_{\mu}$. In this construction, each element of the Young-Yamanouchi basis ${\ket{\mu, i_\mu}}$ is associated with a standard tableau $s^{\mu}_{i_\mu}$.

The Young-Yamanouchi basis is a subgroup-adapted basis, meaning that the action of the subgroup $\s_{p-1} \subset \s_p$ on the Young-Yamanouchi basis ${\ket{\mu,i_\mu}}$ for $\mu \vdash p$ is unitarily equivalent to the action of $\s_{p-1}$ on the Young-Yamanouchi basis ${\ket{\alpha, i_\alpha}}$ for $\alpha \vdash p-1$, as follows: 
\begin{align} 
\bra{\mu, i_\mu} \varphi^{\mu}(\sigma) \ket{\mu, j_\mu} = \delta^{\alpha \alpha'} \bra{\alpha, i_\alpha} \varphi^{\alpha}(\sigma) \ket{\alpha', j_{\alpha'}} \quad \forall \sigma \in \s_{p-1}, \label{eq:subgroup_adapted}
\end{align} 
where the standard tableaux $s^{\alpha}_{i_\alpha}$ and $s^{\alpha'}_{j_{\alpha'}}$, corresponding to the basis vectors $\ket{\alpha, i_\alpha}$ and $\ket{\alpha', j_{\alpha'}}$, are obtained by removing a box labeled \fbox{$n$} from $s^{\mu}_{i_\mu}$ and $s^{\mu}_{j_\mu}$, respectively.
Notice that the indices $i_\mu$ and $i_\alpha$ are different indices connected to each other by removing a proper box in the standard Young tableau. This is an isomorphic mapping. Using this relation, for fixed $\mu$ we can write the label $i_\mu$ alternatively through a triple $i_\mu = (\mu,\alpha,i_\alpha)$, where the index $i_\alpha$ denotes the position in irrep $\alpha$, $\alpha$ tells in which irrep $\alpha$ we are, and index $\mu$ tells from which irrep $\mu$ we take the partially reduced irreducible representation (PRIR) $\alpha$ (for more information, see \cite{StudzinskiIEEE22}).  Using this notation we can write operators~\eqref{eqn:basis_Eij} in the terms of basis adapted to the subgroup $\s_{p-1}$:
\begin{align}\label{eqn:e_prir_notation}
    E^{\mu}_{ij} \equiv E_{i_\mu j_\mu}^\mu \equiv E_{i_\alpha j_\beta}^{\alpha\beta}(\mu),
\end{align}
where we identify the lower indices with a following triplets $i\equiv i_\mu = (\mu,\alpha,i_\alpha)$ and $j\equiv j_\mu = (\mu,\beta,i_\beta)$. In fact, the last notation in~\eqref{eqn:e_prir_notation} can be further simplified to
\begin{align}
\label{eqn:e_prir_notation2} 
E_{i_\alpha j_\beta}^{\alpha\beta}(\mu)\equiv E^{\mu}_{(\alpha,i_\alpha) (\beta,j_\beta)} \equiv E^{\mu}_{i_\alpha j_\beta},
\end{align}
and it will be used later in parts of the manuscript.
Additionally, through paper, we are gonna use right and left half-PRIR notation, where we leave some indices in terms of group $\s_p$ and rest in basis adapted subgroup $\s_{p-1}$, for example, consider operator $\mathbb{C}[\s_p]\times\mathbb{C}[\s_p]$,
\begin{align}\label{eqn:e_halfprir_notation}
     E_{i_\mu j_\mu}^{\mu}\otimes E^{\nu}_{i'_\nu j'_\nu} \equiv E^{\mu}_{i_\mu j_\beta}\otimes E^{\nu}_{i'_{\nu} j'_{\beta'}} \qquad E_{i_\mu j_\mu}^{\mu}\otimes E^{\nu}_{i'_\nu j'_\nu} \equiv E^{ \ \ \mu}_{i_\alpha j_\mu}\otimes E^{ \ \  \ \nu}_{i'_{\alpha'} j'_{\nu}} \ ,
\end{align}
where $j_\mu = (\mu,\beta,j_\beta)$ and $j'_{\nu} = (\nu,\beta',j'_{\beta'})$. 

At the end of this section, we collect two useful facts regarding the irreducible matrix units~\eqref{eq:def_E} taken from other papers. In some of them, we slightly changed the notation to make it more suitable for this note.
\begin{lemma}[Lemma 2 in~\cite{Ebler_2023}]\label{lemma:biger_e_operator}
    Operator $E^{\alpha}_{ij}\otimes \id$, where $E^{\alpha}_{ij}$ are irreducible matrix units of the algebra $\mathbb{C}[\s_{p-1}]$, can be written in terms of irreducible matrix units of the algebra $\mathbb{C}[\s_p]$ as
    \begin{align}
        E^{\alpha}_{i_\alpha j_\alpha}\otimes \id = \sum_{\mu\ni \alpha}E^{\mu}_{i_\alpha j_\alpha}.
    \end{align}
\end{lemma}

\begin{lemma}[Lemma 7 in~\cite{StudzinskiIEEE22}]
\label{L3a}
	Let   $E^{\mu}_{ij}\equiv E_{i_{\alpha}j_{\alpha'}}^{\mu}$ be the irreducible matrix units of the algebra $\mathbb{C}[\s_p]$ written in group adapted basis to $\mathbb{C}[\s_{p-1}]$. Then the partial trace from it over the last system equals to
	\be	\tr_{p}E_{i_{\alpha}j_{\alpha'}}^{\mu}=\frac{m_{\mu}}{m_{\alpha}}E^{\alpha}_{i_{\alpha}j_{\alpha}}\delta^{\alpha \alpha'}.
	\ee
\end{lemma}

The Young-Yamanouchi basis is built from a single box \fbox{$1$} representing a single system. Then by adding another box \fbox{$2$} in a proper way, which we associate with standard tableaux $s^{\alpha}_{i_{\alpha}}$, for $\alpha\vdash 2$. Notice that due to the construction of the Young tableau, we obtain frame $\alpha = (2)$ or $\alpha = (1,1)$ which represents different vectors $|\alpha,i_\alpha\>$ associated with different paths on the Bratteli diagram, see Figure~\ref{fig:y_paths} for the illustration. Adding other boxes we can make a basis with $p$ systems and by adding them in a specific way, we obtain different vectors for $\mu\vdash p$. For instance, in the case of 
\begin{align}
    \ytableausetup
{boxsize=1em}
\ytableausetup
{aligntableaux=top}
\alpha = \ydiagram{2,1}, \quad s^{\alpha}_{i_a} = \ytableaushort{12,3},\quad \alpha+\Box = \Biggl\{ \mu = \ydiagram{3,1}, \quad \nu=\ydiagram{2,2},\quad \lambda = \ydiagram{2,1,1} \Biggr\},
\end{align}
the standard tableaux obtained by adding a box \fbox{$4$} to $s^{\alpha}_{i_a}$ are given by
\begin{align}
     \ytableausetup
{boxsize=1em}
\ytableausetup
{aligntableaux=top}
s^{\mu}_{a^{\alpha}_{\mu}} = \ytableaushort{124,3},\quad s^{\nu}_{a^{\alpha}_{\nu}} = \ytableaushort{12,34}, \quad s^{\lambda}_{a^{\alpha}_{\lambda}} = \ytableaushort{12,3,4}.
\end{align}

\begin{figure}[ht]
\centering\includegraphics[width=0.9\textwidth]{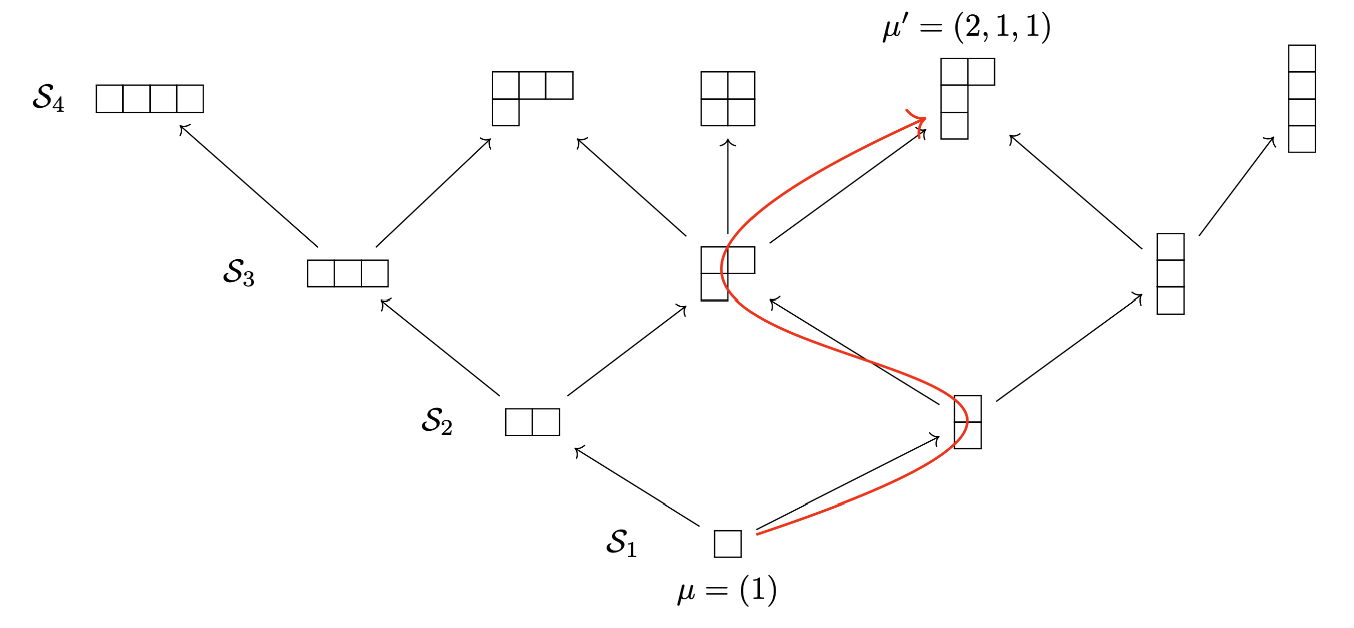}
    \caption{The Bratteli diagram  with four consecutive layers labeled by permutation group from $\s_1$ to $\s_4$. By the red arrow, we present a possible path from irrep $\mu=(1)$ of $\s_1$ to irrep $\mu'=(2,1,1)$ of $\s_4$.}
    \label{fig:y_paths}
\end{figure}

\subsection{Partial transposition and the algebra of partially transposed permutation operators \texorpdfstring{$\mathcal{A}_{p-k,k}^{d}$}{Lg}}

The operation known as partial transposition, denoted as \((\cdot)^{t_k}\) for the \(k\)-th system, is defined as the linear extension of the standard matrix transposition \((\cdot)^t\) for a specified basis, such that \(\langle i|X^t|j\rangle := \bra{j}X \ket{i}\). In a bipartite setting, the partial transposition \(t_1\) concerning the first system is expressed as:
\begin{align}
t_1:\; |i\rangle\langle j| \otimes |k\rangle\langle l| \mapsto |j\rangle\langle i| \otimes |k\rangle\langle l|.
\end{align}
Similarly, we can define the partial transposition \(t_2\) for the second subsystem. The relationship between the permutation operator \(SWAP\equiv V_{(12)}\) and the Bell state \(\ket{\psi^+}_{12}=\frac{1}{\sqrt{d}} \sum_{i=1}^d \ket{ii}_{12}\) is given by:
\begin{align}
\label{eq:BellV'}
V_{(12)} = d \dyad{\psi^+}^{t_1}_{12}.
\end{align}
This definition of partial transposition can also be extended to a multi-partite context. Additionally, we utilize a "ping-pong" trick for the maximally entangled state:
\begin{align}
\label{eq:ping-pong}
\forall X\in M(d,\mathbb{C}), \quad X\otimes \id \ket{\psi^+}_{12} = \id \otimes X^t \ket{\psi^+}_{12}.
\end{align}
With the definition of the group algebra \(\mathbb{C}[\s_p]\) and the concept of partial transposition, we can define the algebra of partially transposed operators concerning the last \(k\) subsystems:
\begin{align}
\label{eqn:a_p^k}
\mathcal{A}_{p-k,k}^{d} := \operatorname{span}_{\mathbb{C}}\{ V^{t_{k}\circ t_{k+1}\circ \cdots \circ t_p}_{\pi} : \pi \in \s_p \},
\end{align}
where \(t_{k}\circ t_{k+1}\circ \cdots \circ t_p\) represents the composition of the partial transpositions for systems \(k, k+1, \ldots, p\). The elements of \(\mathcal{A}_{p-k,k}^{d}\) commute with the mixed actions of the unitary group of the form \(U^{\otimes (p-k)} \otimes \overline{U}^{\otimes k}\), with \(\overline{U}\) denoting complex conjugation, and \(U \in \mathcal{U}(d)\). 

The algebra $\mathcal{A}_{p-k,k}^{d}$ admits the following chain inclusion of two-sided ideals:
\begin{align}
\label{eqn:inclusions_of_m}
\mathcal{M}^{(k)} \subset \mathcal{M}^{(k-1)} \subset \cdots \subset \mathcal{M}^{(1)} \subset \mathcal{M}^{(0)} \equiv \mathcal{A}^{d}_{p-k,k},
\end{align}
where for $0\leq l\leq k$ we have
\begin{align}
\label{eqn:ideal_m}
\mathcal{M}^{(l)} := \{ V_\sigma \otimes V_{\sigma'} V^{(l)} V_{\pi}^\dagger \otimes V_{\pi'}^\dagger \ | \ \sigma,\pi\in \s_{p-k},\quad \sigma',\pi' \in \s_k \},
\end{align}
with
\begin{align}
V^{(l)}&:=\id_{1,p}\otimes\ldots\otimes \id_{p-k-l-1,p-k+1+l}\otimes V_{(p-k,p-k+1)}^{t_{p-k+1}}\otimes\ldots  \otimes V_{(p-k-l+1,p-k+l)}^{t_{p-k+l}}\label{eq:Vl}.
\end{align}

The algebra \(\mathcal{A}^{d}_{p-k,k}\) serves as a matrix representation of a diagram algebra known as the walled Brauer algebra \(\mathcal{B}^\delta_{p-k,k}\), where \(\delta \in \mathbb{C}\), and was firstly introduced in~\cite{Bra37}, and studied later in the context of quantum information tasks for $k=1$~\cite{Yongzhang2021permutation,Studziński_2013_commutant_structure,Moz1,MozJPA}, and $k\geq 1$~\cite{grinko2023gelfandtsetlinbasispartiallytransposed}. This algebra is a restricted version of the full Brauer algebra~\cite{Bra37} studied extensively in the literature~\cite{vladimirgturaev2021,KOIKE198957,BENKART1994529,BEN96,Halverson1996CharactersOT,bulgakova:tel-02554375}.
The abstract algebra \(\mathcal{B}^\delta_{p-k,k}\) consists of formal combinations of diagrams, with each diagram featuring two rows containing \(p-k + k\) nodes separated by a vertical wall. Nodes can be paired under the following rules:
\begin{itemize}
\item If both nodes are in the same row, they must be on opposite sides of the wall.
\item If the nodes are in different rows, they must be on the same side of the wall.
\end{itemize}
Figure \ref{fig:wba_fig} provides a visual representation of the composition of such diagrams.

\begin{figure}[htbp]
  \includegraphics[width=1\linewidth]{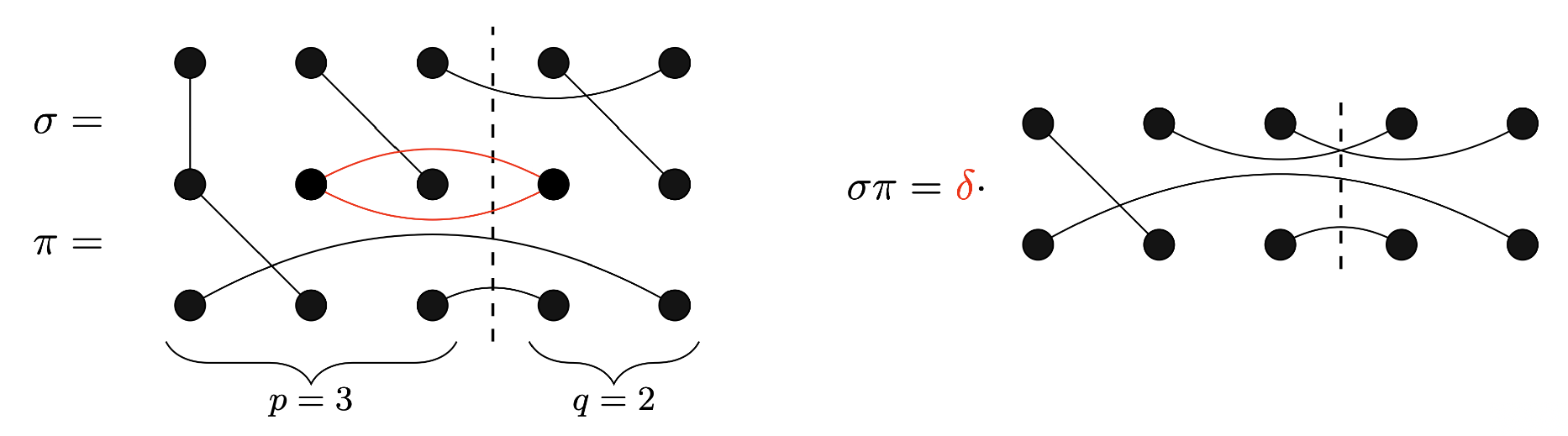}
\caption{A graphical depiction of the composition of two diagrams \(\sigma,\pi \in \mathcal{B}_{3,2}^\delta\). Identifying a closed loop (in red) results in multiplying the diagram by a scalar \(\delta \in \mathbb{C}\), showing that the composition \(\sigma\pi\) remains within \(\mathcal{B}_{3,2}^\delta\).}
\label{fig:wba_fig}
\end{figure}

\section{Technical facts regarding partial traces in the algebras 
\texorpdfstring{$\mathcal{A}_{p,p}^{d}$ and $\mathbb{C}[\mathcal{S}_{p}]\times \mathbb{C}[\mathcal{S}_{p}]$}{Lg}}\label{sec:facts_of_alg_of_part_op}
In this section, we prove technical facts describing properties of the algebra of partially transposed permutation operators and their relations to irreducible matrix units of the group algebra $\mathbb{C}[\s_p]$ introduced in equation~\eqref{eqn:basis_Eij}.
We focus on a particular setup where we consider in total $2p$ particles separated by a 'wall' where particles are divided evenly, meaning we have split $p:p$ (see Fig \ref{fig:visualisation_if_vprim}). The systems are labeled by $1,2,\ldots,p-1,p,p+1,\ldots,2p-1,2p$ respectively. Through the paper from now on we provide the following convention for building irreducible matrix units $E_{i_\mu j_\mu}^{\mu}$, where $\mu\vdash p$, acting on both sides of the wall. When considering operator $E_{i_\mu j_\mu}^{\mu}$ on the left-hand side of the wall, we assume the Young-Yamanouchi basis described in Section~\ref{sec:SW} is constructed from the system $p^{\text{th}}$ up to $1^{\text{st}}$ system. 
On the right-hand side of the wall, we assume the basis is constructed by starting from the system $(p+1)^{\text{th}}$ system and finishing on $(2p)^{\text{th}}$ system. 
Additionally, to simplify the notation we distinguish two types of operators from the algebra $\mathcal{A}_{p,p}^{d}$. These objects naturally appear in further calculations. Namely, we define the following:
\begin{align}
 V^{(p)}&:= V_{(1,2p)}^{t_{2p}}\otimes V_{(2,2p-1)}^{t_{2p-1}}\otimes \ldots \otimes V_{(p-k+1,p+k)}^{t_{p+k}}\otimes \ldots \otimes V_{(p,p+1)}^{t_{p+1}},\label{eq:Vp}\\
V^{(p-1)}&:=\id_{1,2p}\otimes V_{(2,2p-1)}^{t_{2p-1}}\otimes V_{(3,2p-2)}^{t_{2p-2}}\otimes \ldots \otimes V_{(p-k+1,p+k)}^{t_{p+k}}\otimes \ldots \otimes V_{(p,p+1)}^{t_{p+1}},\label{eq:Vp-1}
\end{align}
where $\id_{1,2p}$ is the identity operator acting on systems 1 and $2p$.
Notice that the operators $V^{(p)}, V^{(p-1)}$ can be written as a tensor product of the maximally entangled states on respective systems, as it is shown in~\eqref{eq:BellV'} for the bipartite case.

\begin{figure}[htbp]
\includegraphics[width=1\textwidth]{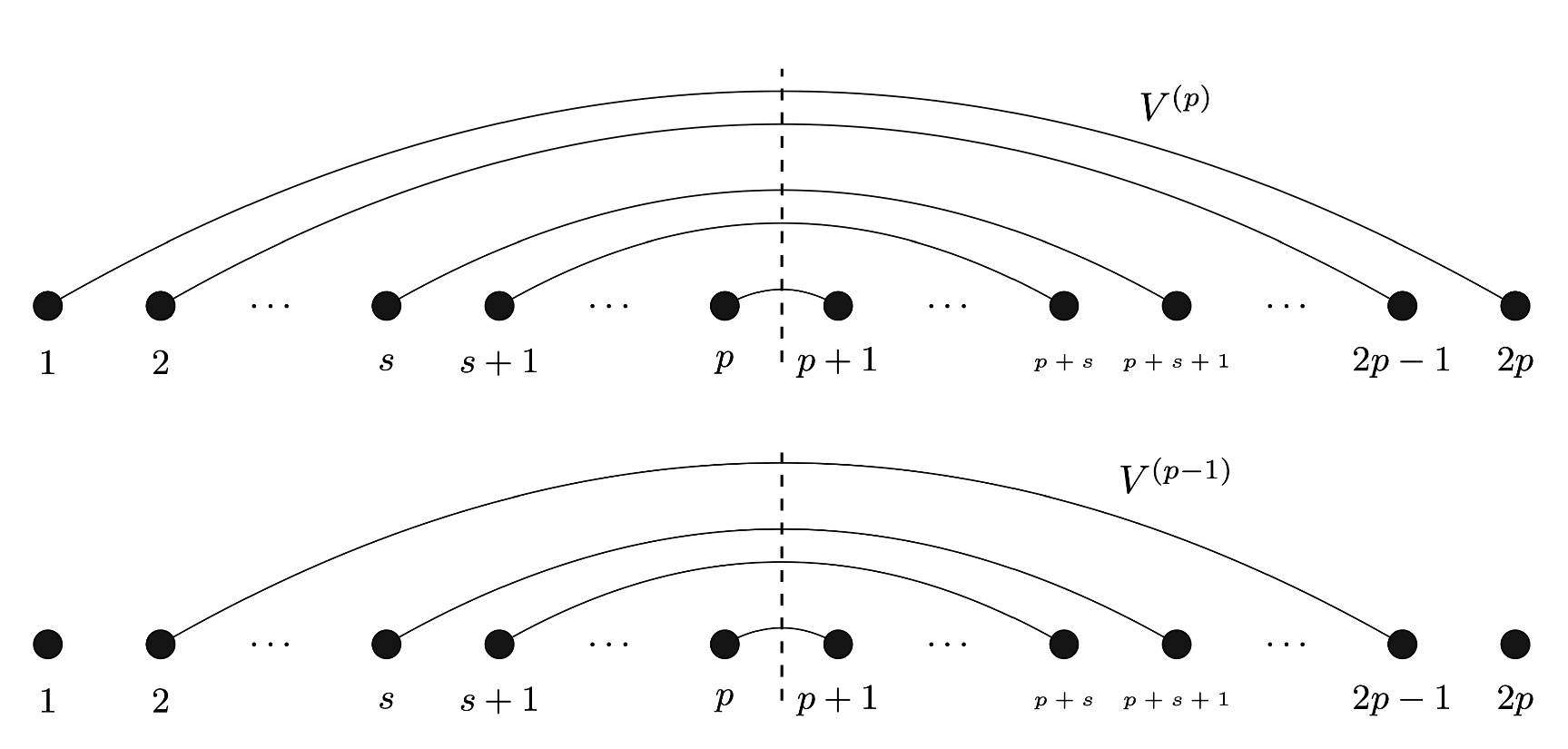}
\caption[]{Top: Graphic illustration on the diagram level of the operator $V^{(p)}$ from equation~\eqref{eq:Vp}. Bottom:  Graphic illustration on the diagram level of the operator $V^{(p-1)}$ from equation~\eqref{eq:Vp-1}. On the abstract level, these objects are special cases of elements from the walled Brauer algebra $\mathcal{B}_{p,p}^\delta$ with $\delta=d$. On the representation space $(\mathbb{C}^d)^{\otimes 2p}$ operators $V^{(p)},V^{(p-1)}$ are elements of the algebra of the partially transposed permutation operators $\mathcal{A}_{p,p}^{d}$.}
\label{fig:visualisation_if_vprim}
\end{figure}

We can generalize the 'ping-pong' trick \eqref{eq:ping-pong} to a more general scenario with a greater number of systems using the operator $V^{(p)}$ \eqref{eq:Vp}. Such generalization highlights the important effect of shuffling the systems that depend on the construction of operator $V^{(p)}$, as presented in the following fact.

\begin{fact}\label{fact:max_ent_ordering}
        Let $A_i$, $B_i\in M(d,\mathbb{C})$ for every $i=1,\ldots,2p$. Define the following product operators $A:=A_1\otimes A_2\otimes \ldots \otimes A_p$, $B:=B_{p+1}\otimes B_{p+2}\otimes \ldots \otimes B_{2p}$. Then for the operator $V^{(p)}$ from \eqref{eq:Vp} the following equality holds
        \begin{align}\label{eqn:pingpongsubsytems}
            (A\otimes B)V^{(p)} =\Big(A_1B_{2p}^T\otimes A_2B_{2p-1}^T\otimes\ldots\otimes A_{p}B_{p+1}^T \otimes \id_{p+1} \otimes\ldots\otimes \id_{2p}\Big)V^{(p)},
        \end{align}
        \end{fact}
    \begin{proof}
    The proof goes straightforwardly by recursively using the 'ping-pong' trick from equation~\eqref{eq:ping-pong}:
        \begin{align}
            &(A\otimes B)V^{(p)} = \Big(A_1\otimes A_2\otimes \ldots \otimes A_p \otimes B_{p+1}\otimes B_{p+2}\otimes \ldots \otimes B_{2p}\Big) V_{(1,2p)}^{t_{2p}} \otimes V_{(2,2p-1)}^{t_{2p-1}}\otimes\ldots\otimes V_{(p,p+1)}^{t_{p+1}}\\
            &= \Big(A_1B_{2p}^T \otimes A_2\otimes \ldots\otimes A_{p} \otimes B_{p+1}\otimes B_{p+2} \otimes \ldots\otimes \id_{2p}\Big)V_{(1,2p)}^{t_{2p}} \otimes V_{(2,2p-1)}^{t_{2p-1}}\otimes\ldots\otimes V_{(p,p+1)}^{t_{p+1}}\\
            &= \Big(A_1B_{2p}^T\otimes A_2B_{2p-1}^T\otimes\ldots\otimes A_{p}B_{p+1}^T  \otimes \id_{p+1} \otimes\ldots\otimes \id_{2p}\Big)V_{(1,2p)}^{t_{2p}} \otimes V_{(2,2p-1)}^{t_{2p-1}}\otimes\ldots\otimes V_{(p,p+1)}^{t_{p+1}}.
        \end{align}
        \end{proof}
        
\begin{remark}
\label{R:Eij}
In particular, taking in Fact~\ref{fact:max_ent_ordering} instead of arbitrary matrices $A,B\in M(d^p,\mathbb{C})$ irreducible matrix units from~\eqref{eqn:basis_Eij} respecting ordering of the Young-Yamanouchi construction on the both sides of the wall, we get
   \begin{align}
        (E^\mu_{i_\mu j_\mu}\otimes E^{\nu}_{k_\nu l_\nu})V^{(p)} =(E^\mu_{i_\mu j_\mu} E^{\nu}_{l_\nu k_\nu} \otimes \id)V^{(p)}=\delta^{\mu \nu}\delta_{j_\mu l_\nu}(E^\mu_{i_\mu l_\mu}\otimes \id)V^{(p)}.
    \end{align}
\end{remark}

In the following, we prove facts regarding sandwiching an arbitrary operator $X\in \mathcal{A}_{p,p}^{d}$ by the elements $V^{(p)}, V^{(p-1)}$ from~\eqref{eq:Vp} and~\eqref{eq:Vp-1} respectively. We start by reminding the Fact 1 from~\cite{StudzinskiIEEE22}:
\begin{fact}[Fact 1 in \cite{StudzinskiIEEE22}]\label{fact:sandwitch_of_v}
    For an arbitrary operator $X\in M(d^p,\mathbb{C})$ acting on first $p$ systems and the operator $V^{(p)}$ from~\eqref{eq:Vp}, we have the following equality
    \begin{align}
        V^{(p)}\Big(X\otimes \id_{p+1,\ldots,2p}\Big)V^{(p)} = \tr(X)V^{(p)},
    \end{align}
    where $\id_{p+1,\ldots,2p}:=\id_{p+1}\otimes \cdots \otimes \id_{2p}$ is the identity operator acting on $p$ last subsystems.
\end{fact}
Now, we present and prove more general property involving element $V^{(p-1)}$ from~\eqref{eq:Vp-1} and arbitrary elements from the algebra $\mathcal{A}_{p,p}^{d}$. Namely, we have:
\begin{fact}
\label{F:concatenation}
For an arbitrary operator $X\in \mathcal{A}_{p,p}^{d}$ and operator $V^{(p-1)}$ given through equation~\eqref{eq:Vp-1} the following equality holds:
\begin{equation}
\label{eq:F1}
V^{(p-1)}XV^{(p-1)}=\widetilde{X}\otimes V^{(p-1)},
\end{equation}
where $\widetilde{X}\in \mathcal{A}^{d}_{1,1}$, with its explicit form
\begin{align}
  \widetilde{X}
  &=\tr_{2,\ldots,p,p+1,\ldots,2p-1}\left(XV^{(p-1)}\right).
\end{align}
\end{fact}

\begin{proof}
We prove the statement recursively. Let us take any $1\leq k\leq p$. Using relation~\eqref{eq:BellV'} we can write $V^{t_{p+k}}_{(p-k+1,p+k)}=dP^+_{p-k+1,p+k}$, where $P^+_{p-k+1,p+k}=|\psi^+\>\<\psi^+|$ is rank one projector on the maximally entangled state between systems $p-k+1$ and $p+k$. Set $k=1$, then we can write the following
\begin{align}
d^2\left(P^+_{p,p+1}\otimes \id_{\overline{p,p+1}}\right)X\left(P^+_{p,p+1}\otimes \id_{\overline{p,p+1}}\right)&=d^2\left(|\psi^+\>\<\psi^+|\otimes \id_{\overline{p,p+1}}\right)X\left(|\psi^+\>\<\psi^+|\otimes \id_{\overline{p,p+1}}\right)\label{eq:sand1}\\
&=\<\psi^+|X|\psi^+\>\otimes P^+_{p,p+1}\\
&=d^2\tr_{p,p+1}\left(XP^+_{p,p+1}\right)\otimes P^+_{p,p+1}\\
&=\tr_{p,p+1}\left(XV^{t_{p+1}}_{(p,p+1)}\right)\otimes V^{t_{p+1}}_{(p,p+1)},
\end{align}
where $\id_{\overline{p,p+1}}$ denotes identity operator acting on all systems but $p$ and $p+1$. Obviously, we have $XV^{t_{p+1}}_{(p,p+1)}\in \mathcal{A}^{d}_{p,p}$, and $\tr_{p,p+1}\left(XV^{t_{p+1}}_{(p,p+1)}\right)\in \mathcal{A}^{d}_{p-1,p-1}$. In the next step, when sandwiching the result from~\eqref{eq:sand1} with $P^+_{p-1,p+2}\otimes \id_{\overline{p-1,p+2}}$, we get
\begin{align}
\tr_{p-1,p+2}\left(\tr_{p,p+1}\left(XV^{t_{p+1}}_{(p,p+1)}\right)V^{t_{p+2}}_{(p-1,p+2)}\right)\otimes V^{t_{p+1}}_{(p,p+1)} \otimes V^{t_{p+2}}_{(p-1,p+2)}.
\end{align}
This time the concatenation of partial traces produces an element from $\mathcal{A}^{d}_{p-2,p-2}$. For an arbitrary level of concatenation $1\leq k\leq p$, we have then
\begin{align}
&\tr_{p-k+1,p+k}\left(\cdots\tr_{p-1,p+2}\left(\tr_{p,p+1}\left(XV^{t_{p+1}}_{(p,p+1)}\right)V^{t_{p+2}}_{(p-1,p+2)}\right)\cdots V^{t_{p+k}}_{p-k+1,p+k}\right)\otimes \\
&\otimes V^{t_{p+k}}_{p-k+1,p+k}\otimes \cdots \otimes V^{t_{p+1}}_{(p,p+1)} \otimes V^{t_{p+2}}_{(p-1,p+2)},
\end{align}
where the result of the traces concatenation belongs to $\mathcal{A}^{d}_{p-k,p-k}$. Continuing this procedure up to $k=p-1$, we indeed obtain expression~\eqref{eq:F1} with $\widetilde{X}\in \mathcal{A}^{d}_{1,2}$. This finishes the proof.
\end{proof}

Notice that in Fact \ref{fact:max_ent_ordering} we can plug in irreducible matrix units given by \eqref{eqn:basis_Eij} in a place of matrices $A$ and $B$. Then as a result we obtain  This together with the Facts \ref{fact:sandwitch_of_v} and~\ref{F:concatenation} provides the following lemma.

\begin{lemma}\label{lem:trace_values}
Let us take the operators $V^{(p)}, V^{(p-1)}$ given through equations~\eqref{eq:Vp}-\eqref{eq:Vp-1}, and irreducible matrix units $E^{\mu}_{i_\mu j_\mu},E^{\nu}_{k_\nu l_\nu}\in \mathbb{C}[\s_p]$ from~\eqref{eqn:basis_Eij}. Then writing irreducible matrix units $E^{\mu}_{i_\mu j_\mu}, E^{\nu}_{k_\nu l_\nu}$ in basis adapted to subalgebra $\mathbb{C}[\s_{p-1}]$ as $ E^{\mu}_{i_\mu j_\mu} \equiv E_{i_\alpha j_{\alpha'}}^{\mu}$ and $ E^{\nu}_{k_\nu l_\nu} \equiv E^{\nu}_{k_{\beta}l_{\beta'}}$ we get the following: 
    \begin{enumerate}[1)]
    \item $\tr(E^\mu_{i_\mu j_\mu}\otimes E^{\nu}_{k_\nu l_\nu}V^{(p)})=m_\mu \delta^{\mu\nu}\delta_{k_\nu i_\mu}\delta_{l_\nu j_\mu}$,
        \item $\tr(E^\mu_{i_\mu j_\mu }\otimes E^{\nu}_{k_\nu l_\nu }V^{(p-1)})\equiv \tr(E^\mu_{i_\alpha j_{\alpha'} }\otimes E^{\nu}_{k_\beta l_{\beta'} }V^{(p-1)})=\frac{m_\mu m_\nu}{m_\alpha} \delta^{\alpha\alpha'}\delta^{\beta\beta'}\delta^{\alpha\beta}\delta_{j_{\alpha'} l_{\beta'}}\delta_{i_\alpha k_\beta} $.
    \end{enumerate}
    \begin{proof}
        We split the calculations into two parts concerning terms with $V^{(p)}$ and $V^{(p-1)}$.
        \begin{itemize}
            \item  We use ordering to construct the Young-Yamanouchi basis on both sides of the wall described at the beginning of this section. Now, to prove the first part of the statement, we have the following chain of equalities:
            \begin{align}
                \tr(E^\mu_{i_\mu j_\mu}\otimes E^{\nu}_{k_\nu l_\nu}V^{(p)}) &= \tr( ( E^{\mu}_{i_\mu j_\mu} E^{\nu}_{l_\nu k_\nu}\otimes \id ) V^{(p)} )\\
                &= \tr(E^{\mu}_{i_\mu j_\mu} E^{\nu}_{l_\nu k_\nu}  \tr_{p+1\ldots 2p}(V^{(p)}))\\
                &=\tr(E^{\mu}_{i_\mu j_\mu} E^{\nu}_{l_\nu k_\nu} )\\
                &=\tr(E^{\mu}_{i_\mu k_\mu })\delta^{\mu \nu}\delta_{j_\nu l_\mu}\\
            &= m_\mu  \delta^{\mu \nu}\delta_{i_\nu k_\mu}\delta_{j_\nu l_\mu}.
            \end{align}
            In the first line, we use the 'ping-pong' trick from equation~\eqref{eq:ping-pong}. In the second line, we use the fact that the composition $ E^{\mu}_{i_\mu j_\mu} E^{\nu}_{l_\nu k_\nu}$ acts trivially on the last $p$ systems. Thanks to this we can evaluate the partial trace from $V^{(p)}$ over the last $p$ systems which is equal to the identity operator $\id_{1\ldots p}$ acting on the first $p$ systems. In the third line, we apply the orthogonality relations for irreducible operator basis given in~\eqref{eqn:basis_Eij}. Finally, in the fourth line, we exploit the trace rule from~\eqref{eqn:basis_Eij}.
            \item Proving the second part of the statement is more involving. It requires decomposition of the indices $\mu, \nu,i,j,k,l$ labeling irreducible matrix units of $\mathbb{C}[\s_p]$ in the way adapted to the subalgebra $\mathbb{C}[\s_{p-1}]$. To do so we exploit convention from equation~\eqref{eqn:e_prir_notation} and write:
            \begin{align}
            &E_{i_\mu j_\mu}^{\mu}\equiv E_{i_\alpha j_{\alpha'}}^{\mu}\label{eqn:E_op_build_nor},\\
            &E_{k_\nu l_\nu}^{\nu}\equiv E_{k_\beta l_{\beta'}}^{\nu}.\label{eqn:E_op_build_cnor}
            \end{align}
            Here we again use ordering in constructing the Young-Yamanouchi basis on both sides of the wall described at the beginning of this section. First, we express irreducible matrix units in the PRIR basis given in~\eqref{eqn:e_prir_notation}.
            Then, we start calculations by using Lemma~\ref{L3a} and taking a partial trace over the last system on the left and right-hand side on both sides of the wall:
            \begin{align}
               \tr(E^\mu_{i_\mu j_\mu}\otimes E^{\nu}_{k_\nu l_\nu}V^{(p-1)}) &= \tr( \tr_1 \Big(E_{i_{\alpha} j_{\alpha'}}^{\mu}\Big) \otimes \tr_{2p}\Big(E^{\nu}_{k_{\beta} l_{\beta'}}\Big) V^{(p-1)})\\
            &= \frac{m_\mu}{m_\alpha}\frac{m_\nu}{m_\beta}\tr( \Big(E^{\alpha}_{i_\alpha j_\alpha} \otimes E^\beta_{k_\beta l_\beta}\Big) V^{(p-1)})\delta^{\alpha\alpha'}\delta^{\beta\beta'}\\
            &= \frac{m_\mu}{m_\alpha}\frac{m_\nu}{m_\beta}\tr( \Big(E^{\alpha}_{i_\alpha j_\alpha} (E^\beta_{ k_\beta l_\beta})^T \otimes \id\Big) V^{(p-1)})\delta^{\alpha\alpha'}\delta^{\beta\beta'}\\
            &= \frac{m_\mu}{m_\alpha}\frac{m_\nu}{m_\beta}\tr( E^{\alpha}_{i_\alpha j_\alpha} E^\beta_{l_\beta k_\beta} \tr_{p+1,\ldots,2p-1}(V^{(p-1)}))\delta^{\alpha\alpha'}\delta^{\beta\beta'}\\
            &= \frac{m_\mu m_\nu}{m_\alpha} \delta^{\alpha\alpha'}\delta^{\beta\beta'}\delta^{\alpha\beta}\delta_{j_{\alpha'} l_{\beta'}}\delta_{i_\alpha k_\beta}.
            \end{align}
            In the above, we use orthogonality relation, together with the trace property \eqref{eqn:composition_of_e_operators} as well as Lemma \ref{L3a} and the 'ping-pong' trick~\eqref{eq:ping-pong}. The operators $ E^{\alpha}_{i_\alpha j_\alpha}$ and $E^\beta_{l_\beta k_\beta}$ after the 'ping-pong' trick are built in the same ordering, hence we can compose them.
        \end{itemize}
    \end{proof}
\end{lemma}

Having proven all the above facts and lemmas we are in a position to prove the main result of this section. In the following, we prove a proposition and a theorem giving explicit formulas for the operator $\widetilde{X}$ from Fact~\ref{F:concatenation} when sandwiching irreducible matrix units $E_{i_\mu j_\mu}^{\mu }\otimes E_{k_\nu l_\nu}^{\nu } \in \mathbb{C}[\mathcal{S}_{p}]\times \mathbb{C}[\mathcal{S}_{p}]$ by the operators $V^{(p)},V^{(p-1)}$ from~\eqref{eq:Vp},~\eqref{eq:Vp-1}.

\begin{proposition}\label{prop:vpeevp}
Let $E_{i_\mu j_\mu}^{\mu }\otimes E_{k_\nu l_\nu }^{\nu } \in \mathbb{C}[\mathcal{S}_{p}]\times \mathbb{C}[\mathcal{S}_{p}]$ be tensor product of the irreducible matrix units defined through \eqref{eqn:basis_Eij}, and $V^{(p)}\in \mathcal{A}^{d}_{p,p}$ from~\eqref{eq:Vp}, then 
\begin{equation}
V^{(p)}[E_{i_\mu j_\mu}^{\mu }\otimes E_{k_\nu l_\nu}^{\nu }]V^{(p)}= m_\mu V^{(p)}\delta_{j_\mu l_\nu}\delta_{i_\mu k_\nu}\delta^{\mu\nu},
\end{equation}
where the number $m_\mu$ denotes the multiplicity of the irrep $\mu\vdash p$ in the Schur-Weyl duality.
\end{proposition}

\begin{proof}
Consider following sandwiching by $V^{(p)}$ of tensor product of irreducible matrix units \eqref{eqn:basis_Eij}, then by straightforward calculations we get
\begin{align}
    V^{(p)}[E_{i_\mu j_\mu}^{\mu }\otimes E_{k_\nu l_\nu}^{\nu }]V^{(p)}&= V^{(p)}[E_{i_\mu j_\mu}^{\mu } (E_{k_\nu l_\nu}^{\nu })^T\otimes \id]V^{(p)}\\
    &= V^{(p)}[E_{i_\mu j_\mu}^{\mu } E_{l_\nu k_\nu}^{\nu }\otimes \id]V^{(p)}\\
    &= V^{(p)}[E_{i_\mu k_\mu}^{\mu } \otimes \id]V^{(p)} \delta_{j_\mu l_\nu}\delta^{\mu\nu}\\
    &= \tr(E_{i_\mu k_\mu}^{\mu })V^{(p)}\delta_{j_\mu l_\nu}\delta^{\mu\nu}\\
    &= m_\mu V^{(p)}\delta_{j_\mu l_\nu}\delta_{i_\mu k_\nu}\delta^{\mu\nu}.
\end{align}
In the first line, we use the 'ping-pong' trick from~\eqref{eq:ping-pong}, in the second line the composition rule again from~\eqref{eq:def_E}, and in the third line the trace rule from~\eqref{eq:def_E}.
\end{proof}

\begin{theorem}\label{thm:wba_element}
    Let $V^{(p-1)}$ be operator given by equation \eqref{eq:Vp-1}. For the irreducible matrix units $E_{i_\mu j_\mu}^\mu, E_{k_\nu l_\nu}^{\nu} \in \mathbb{C}[\s_p]$ written in basis adapted to subalgebra $\mathbb{C}[\s_{p-1}]$ as $ E^{\mu}_{i_{\mu} j_{\mu}} \equiv E_{i_\alpha j_{\alpha'}}^{\mu}$ and $ E^{\nu}_{kl} \equiv E^{\nu}_{k_{\beta}l_{\beta'}}$ respectively the following equality holds
    \begin{align}\label{eqn:operator_wba}
        V^{(p-1)} E_{i_\mu j_\mu}^\mu \otimes E_{k_\nu l_\nu}^{\nu} V^{(p-1)} &\equiv V^{(p-1)} E_{i_\alpha j_{\alpha'}}^{\mu} \otimes E^{\nu}_{k_{\beta}l_{\beta'}} V^{(p-1)}\\
        &= a^{\mu \nu}(\alpha,\beta,\alpha',\beta',i_\alpha,j_{\alpha'},k_\beta,l_{\beta'}) V^{(p)} + b^{\mu \nu}(\alpha,\beta,\alpha',\beta',i_\alpha,j_{\alpha'},k_\beta,l_{\beta'})V^{(p-1)},
    \end{align}
    where the coefficients $a^{\mu \nu}(\alpha,\beta,\alpha',\beta',i_\alpha,j_{\alpha'},k_\beta,l_{\beta'}),b^{\mu \nu}(\alpha,\beta,\alpha',\beta',i_\alpha,j_{\alpha'},k_\beta,l_{\beta'})$ have a form
    \begin{align}
        a^{\mu \nu}(\alpha,\beta,\alpha',\beta',i_\alpha,j_{\alpha'},k_\beta,l_{\beta'}) &=\frac{1}{d(d^2-1)}\Bigg(d m_\mu \delta^{\mu\nu} \delta_{k_\nu i_\mu}\delta_{l_\nu j_\mu} - \frac{m_\mu m_\nu}{m_\alpha} \delta^{\alpha\alpha'}\delta^{\beta\beta'}\delta^{\alpha\beta} \delta_{j_{\alpha'}l_{\beta'}}\delta_{i_\alpha k_\beta} \Bigg),\label{eqn:coefficient_a}\\
        b^{\mu \nu}(\alpha,\beta,\alpha',\beta',i_\alpha,j_{\alpha'},k_\beta,l_{\beta'})&=\frac{1}{d(d^2-1)}\Bigg(d\frac{m_\mu m_\nu}{m_\alpha} \delta^{\alpha\alpha'}\delta^{\beta\beta'}\delta^{\alpha\beta}\delta_{j_{\alpha'}l_{\beta'}}\delta_{i_\alpha k_\beta} - m_\mu \delta^{\mu\nu}\delta_{k_\nu i_\mu}\delta_{l_\nu j_\mu} \Bigg).\label{eqn:coefficient_b}
    \end{align}
    \end{theorem}
    \begin{proof}
    First let us notice that the operator $E_{i_\mu j_\mu}^\mu \otimes E_{k_\nu l_\nu}^{\nu}\in \mathcal{A}_{p,p}^{d}$. Applying Fact~\ref{F:concatenation} to the operator $V^{(p-1)} E_{i_\mu j_\mu}^\mu \otimes E_{k_\nu l_\nu}^{\nu} V^{(p-1)}$ we conclude it belongs to $\mathcal{A}_{1,1}^{d}=\operatorname{span}_{\mathbb{C}}\left\{V^{(1)},\id_{1,2p}\right\}$, where $V^{(1)}=V^{t_{2p}}_{1,2p}$ and $\id_{1,2p}$ is identity operator acting on systems 1 and $2p$. Hence the element $V^{(p-1)} E^\mu_{i_\mu j_\mu}\otimes E^{\nu}_{k_\nu l_\nu}V^{(p-1)}$ can be written as
        \begin{align}\label{eqn:wba_element}
            V^{(p-1)}E_{i_\mu j_\mu}^\mu \otimes E_{k_\nu l_\nu}^{\nu}V^{(p-1)} = (a V^{(1)} + b\id_{1,2p})\otimes V^{(p-1)} = aV^{(p)}+b V^{(p-1)},
        \end{align}
        where $a,b\in \mathbb{C}$, and for compactness we rid of indices appearing in $a,b$. To evaluate the coefficients $a$ and $b$ we take the trace of both sides of \eqref{eqn:wba_element}
        \begin{align}\label{eqn:evaluationaandb1}
            d^{p-1}\tr( E_{i_\mu j_\mu}^\mu \otimes E_{k_\nu l_\nu}^{\nu} V^{(p-1)}) = a d^p + b d^{p+1},
        \end{align}
         On the other hand, we can multiply expression~\eqref{eqn:wba_element} by $V^{(1)}$, getting
         \begin{align}
         \label{eq:second}
           V^{(p-1)}E_{i_\mu j_\mu}^\mu \otimes E_{k_\nu l_\nu}^{\nu}V^{(p)}  = adV^{(p)}+b V^{(p)}=(ad+b)V^{(p)}, 
         \end{align}
         since $V^{(p-1)}V^{(1)}=V^{(p)}$, and $V^{(p)}V^{(1)}=dV^{(p)}$. Now taking trace of both sides from~\eqref{eq:second}, we obtain
        \begin{align}\label{eqn:evaluationaandb2}
            d^{p-1}\tr(E_{i_\mu j_\mu}^\mu \otimes E_{k_\nu l_\nu}^{\nu} V^{(p)}) = a d^{p+1}+b d^p.
        \end{align}
        We have two linear equations \eqref{eqn:evaluationaandb1} and \eqref{eqn:evaluationaandb2} with two coefficients $a$ and $b$ to solve. Denoting by 
        \begin{align}\label{eqn:x_coefficient}
            x&=\tr( E_{i_\mu j_\mu}^\mu \otimes E_{k_\nu l_\nu}^{\nu} V^{(p-1)}),\\
            y&=\tr(E_{i_\mu j_\mu}^\mu \otimes E_{k_\nu l_\nu}^{\nu} V^{(p)}),\label{eqn:y_coefficient}
        \end{align}
        these equations read
        \[
\setlength\arraycolsep{1pt}
\left\{
\begin{array}{rcrcrc@{\qquad}l}
x &  =   &   a d +bd^2,\\
 y   &   =   &   a d^2 + bd,
\end{array}
\right.
\]
with the following solution
\begin{align}
    a=\frac{dy-x}{d(d^2-1)},\quad b=\frac{dx-y}{d(d^2-1)}.
\end{align}
Implementing the coefficients $x$ and $y$ given by \eqref{eqn:x_coefficient} and \eqref{eqn:y_coefficient} respectively, we get the following
\begin{align} 
a&= \frac{1}{d(d^2-1)}\Bigl[d\tr( E_{i_\mu j_\mu}^\mu \otimes E_{k_\nu l_\nu}^{\nu}V^{(p)}) - \tr(E_{i_\mu j_\mu}^\mu \otimes E_{k_\nu l_\nu}^{\nu}V^{(p-1)}) \Bigr],\\
b&= \frac{1}{d(d^2-1)}\Bigl[d\tr(E_{i_\mu j_\mu}^\mu \otimes E_{k_\nu l_\nu}^{\nu}V^{(p-1)})-\tr(E_{i_\mu j_\mu}^\mu \otimes E_{k_\nu l_\nu}^{\nu}V^{(p)})\Bigr].
\end{align}
Finally using Lemma \ref{lem:trace_values} we evaluate explicit values of the above coefficients:
\begin{align}
a&=\frac{1}{d(d^2-1)}\Bigg(d m_\mu \delta^{\mu\nu} \delta_{k_\nu i_\mu}\delta_{l_\nu j_\mu} - \frac{m_\mu m_\nu}{m_\alpha} \delta^{\alpha\alpha'}\delta^{\beta\beta'}\delta^{\alpha\beta}\delta_{j_{\alpha'}l_{\beta'}}\delta_{i_\alpha k_\beta} \Bigg) ,\\
b&=\frac{1}{d(d^2-1)}\Bigg(d\frac{m_\mu m_\nu}{m_\alpha} \delta^{\alpha\alpha'}\delta^{\beta\beta'}\delta^{\alpha\beta}\delta_{j_{\alpha'}l_{\beta'}}\delta_{i_\alpha k_\beta} - m_\mu \delta^{\mu\nu}\delta_{k_\nu i_\mu}\delta_{l_\nu j_\mu} \Bigg) .
\end{align}
    \end{proof}
In the following, to simplify the notations we will write just $a,b$ for the number from~\eqref{eqn:operator_wba} instead of $a^{\mu \nu}(\alpha,\beta,\alpha',\beta',i_\alpha,j_{\alpha'},k_\beta,l_{\beta'})$, $b^{\mu \nu}(\alpha,\beta,\alpha',\beta',i_\alpha,j_{\alpha'},k_\beta,l_{\beta'})$, whenever it is clear from the context.
The number $a,b$ given through expressions~\eqref{eqn:coefficient_a},~\eqref{eqn:coefficient_b} in Theorem~\ref{thm:wba_element} satisfies a useful property that will be used later as a technical result in proofs of other results. Namely, we have the following:

\begin{lemma}\label{lemma:ad+b=m_mu}
    Let us take the following numbers
    \begin{align}\label{eqn:number_a_with_deltas}
        a\equiv a^{\mu \nu}(\alpha,\beta,\alpha',\beta',i_\alpha,j_{\alpha'},k_\beta,l_{\beta'}) &=\frac{1}{d(d^2-1)}\Bigg(d m_\mu \delta^{\mu\nu} \delta_{k_\nu i_\mu}\delta_{l_\nu j_\mu} - \frac{m_\mu m_\nu}{m_\alpha} \delta^{\alpha\alpha'}\delta^{\beta\beta'}\delta^{\alpha\beta} \delta_{j_{\alpha'}l_{\beta'}}\delta_{i_\alpha k_\beta} \Bigg)\\
        b\equiv b^{\mu \nu}(\alpha,\beta,\alpha',\beta',i_\alpha,j_{\alpha'},k_\beta,l_{\beta'})&=\frac{1}{d(d^2-1)}\Bigg(d\frac{m_\mu m_\nu}{m_\alpha} \delta^{\alpha\alpha'}\delta^{\beta\beta'}\delta^{\alpha\beta}\delta_{j_{\alpha'}l_{\beta'}}\delta_{i_\alpha k_\beta} - m_\mu \delta^{\mu\nu}\delta_{k_\nu i_\mu}\delta_{l_\nu j_\mu} \Bigg) \label{eqn:number_b_with_deltas}
    \end{align}
    where $\mu,\nu\vdash p$ with indices $i_\mu,j_\mu, k_\nu,l_\nu$ written in the PRIR notation \eqref{eqn:e_prir_notation} with respect to subgroup $\mathbb{C}[\s_{p-1}]$, i.e. $i_\mu=(\mu, \alpha, i_\alpha)$
    Then the following equality holds true
    \begin{align}
        ad+b = \frac{m_\mu}{d}\delta^{\mu\nu},
    \end{align}
\end{lemma}
\begin{proof}
    By straightforward calculations, and using the half-PRIR notation \eqref{eqn:e_halfprir_notation}, we get the following
    \begin{align}
        ad+b&= \frac{1}{d^2-1}\Bigg(d m_\mu \delta^{\mu\nu} \delta_{k_\nu i_\mu}\delta_{l_\nu j_\mu} - \frac{m_\mu m_\nu}{m_\alpha} \delta^{\alpha\alpha'}\delta^{\beta\beta'}\delta^{\alpha\beta} \delta_{j_{\alpha'}l_{\beta'}}\delta_{i_\alpha k_\beta} \Bigg) \\
        &+ \frac{1}{d(d^2-1)}\Bigg(d\frac{m_\mu m_\nu}{m_\alpha} \delta^{\alpha\alpha'}\delta^{\beta\beta'}\delta^{\alpha\beta}\delta_{j_{\alpha'}l_{\beta'}}\delta_{i_\alpha k_\beta} - m_\mu \delta^{\mu\nu}\delta_{k_\nu i_\mu}\delta_{l_\nu j_\mu} \Bigg)\\
        &= \Bigg(\frac{d}{d^2-1} m_\mu  \delta^{\mu \nu}\delta_{k_\nu i_\mu}\delta_{l_\nu j_\mu} - \frac{1}{d^2-1}\frac{m_\mu m_\nu}{m_\alpha}\delta^{\alpha \alpha'}\delta^{\beta \beta'}\delta^{\alpha \beta}\delta_{j_{\alpha'}l_{\beta'}}\delta_{i_\alpha k_\beta}\\
        &+\frac{1}{d^2-1}\frac{m_\mu m_\nu}{m_\alpha} \delta^{\alpha \alpha'}\delta^{\beta \beta'}\delta^{\alpha \beta}\delta_{j_{\alpha'}l_{\beta'}}\delta_{i_\alpha k_\beta} - \frac{1}{d(d^2-1)} m_\mu  \delta^{\mu \nu}\delta_{k_\nu i_\mu}\delta_{l_\nu j_\mu} \Bigg)\\
        &= \Bigg(\frac{d}{d^2-1}  -   \frac{1}{d(d^2-1)}  \Bigg)m_\mu  \delta^{\mu \nu}  \delta_{k_\nu i_\mu}\delta_{l_\nu j_\mu}\\
        &= \frac{1}{d} m_\mu  \delta^{\mu \nu} \delta_{k_\nu i_\mu}\delta_{l_\nu j_\mu}.\label{eqn:last_term_of_ad+b}
    \end{align}
    Because in the numbers $a$ and $b$ defined in \eqref{eqn:number_a_with_deltas}, \eqref{eqn:number_b_with_deltas} respectively each has term with deltas $\delta_{k_\nu i_\mu}\delta_{l_\nu j_\mu}$ and observe that in the last equation of above calculations \eqref{eqn:last_term_of_ad+b} we have these deltas too, hence we conclude that 
    \begin{align}
        ad+b = \frac{m_\mu}{d}\delta^{\mu \nu},
    \end{align}
    which finishes the proof.
\end{proof}

\section{Family of spanning operators in the ideals algebra \texorpdfstring{$\mathcal{M}^{(p)},\mathcal{M}^{(p-1)} \subset \mathcal{A}^{d}_{p,p}$}{Lg}}\label{sec:family_of_spanning_operators_in_ideals}
In this section we introduce a family of operators defined on the ideals $\mathcal{M}^{(p)},\mathcal{M}^{(p-1)}$ used later in the paper to construct irreducible matrix units in $\mathcal{M}^{(p)}$ and $\mathcal{M}^{(p-1)}$ respectively. Here, we focus mainly on the linear span properties for the considered objects. From the definition of the algebra \(\mathcal{A}_{p,p}^{d}\) as a matrix representation of the diagram-walled Brauer algebra \(\mathcal{B}^\delta_{p-k,k}\), it is clear that the ideals $\mathcal{M}^{(p)}$ and $\mathcal{M}^{(p-1)}$, defined through~\eqref{eqn:ideal_m}, can be written as the following linear span:
\begin{align}
&\mathcal{M}^{(p)}=\operatorname{span}_{\mathbb{C}}\left\{\Big(E^{\mu}_{i_\mu j_\mu} \otimes E^{\mu'}_{i'_{\mu'} j'_{\mu'}}\Big) V^{(p)} \Big(E_{k_\nu l_\nu}^\nu \otimes E_{k'_{\nu'} l'_{\nu'}}^{\nu'}\Big)\right\},\label{eq:linspan1}\\
&\mathcal{M}^{(p-1)}=\operatorname{span}_{\mathbb{C}}\left\{\Big(E^{\mu}_{i_\mu j_\mu} \otimes E^{\mu'}_{i'_{\mu'} j'_{\mu'}}\Big) V^{(p)} \Big(E_{k_\nu l_\nu}^\nu \otimes E_{k'_{\nu'} l_{\nu'}}^{\nu'}\Big), \Big(E_{i_\mu \ j_\alpha}^{\mu} \otimes E_{k_\nu \ l_\beta}^{\nu}\Big) V^{(p-1)} \Big(E_{j'_{\alpha'} \ i'_{\mu'}}^{\ \ \ \ \mu'} \otimes E_{l'_{\beta'} \ j'_{\nu'}}^{\ \ \ \ \nu'}\Big)\right\},\label{eq:linspan2}
\end{align}
where the operators $V^{(p)}, V^{(p-1)}$ are given through~\eqref{eq:Vp} and~\eqref{eq:Vp-1} respectively. Additionally, the operators sandwiching $V^{(p-1)}$ in \eqref{eq:linspan2} are written in half-PRIR notation introduced in~\eqref{eqn:e_halfprir_notation}. Indeed, the above spanning property holds true since every permutation operator $V_{\sigma_1}, V_{\sigma_2}$ from~\eqref{eqn:ideal_m} can be written in terms of irreducible matrix units\eqref{eqn:basis_Eij} for the algebra $\mathbb{C}[\s_p]$ as it is presented in~\eqref{eq:VasE}.
However, due to sandwiching the operators $V^{(p)}, V^{(p-1)}$ some of the elements in~\eqref{eq:linspan1},~\eqref{eq:linspan2} are redundant. Indeed, we have:

\begin{proposition}
\label{prop:Linspan}
Let $E^{\mu}_{i_\mu j_\mu}\in \mathbb{C}[\mathcal{S}_p]$ be irreducible matrix units for an irrep labelled by $\mu \vdash p$. The ideals $\mathcal{M}^{(p)}$ and $\mathcal{M}^{(p-1)}$ can be written as the following linear spans:
\begin{align}
  &\mathcal{M}^{(p)}=\operatorname{span}_{\mathbb{C}}\left\{ \Big(E^{\mu}_{i_\mu j_\mu} \otimes \id\Big) V^{(p)} \Big(E_{j'_\nu i'_\nu}^\nu \otimes \id\Big) \right\},\label{eq:Linspan1}\\
  &\mathcal{M}^{(p-1)}=\operatorname{span}_{\mathbb{C}}\left\{\Big(E^{\mu}_{i_\mu j_\mu} \otimes \id\Big) V^{(p)} \Big(E_{j'_\nu i'_\nu}^\nu \otimes \id\Big) ,\Big(E_{i_\mu \ 1_\alpha}^{\mu} \otimes E_{j_\nu \ 1_\alpha}^{\nu}\Big) V^{(p-1)} \Big(E_{1_{\alpha'} \ i'_{\mu'}}^{\ \ \ \ \mu'} \otimes E_{1_{\alpha'} \ j'_{\nu'}}^{\ \ \ \ \nu'}\Big)\right\},\label{eq:Linspan2}
\end{align}
where the operators $V^{(p)}, V^{(p-1)}$ are given through~\eqref{eq:Vp} and~\eqref{eq:Vp-1} respectively.
\end{proposition}

\begin{proof}
 We start from proving expression~\eqref{eq:Linspan1}. Here the proof is based on a straightforward application the 'ping-pong' trick from~\eqref{eq:BellV'} adapted to irreducible matrix units for the group algebra $\mathbb{C}[\s_p]$ as it is in Remark~\ref{R:Eij}. Starting from~\eqref{eq:linspan1}, we have:
\begin{align}
\Big(E^{\mu}_{i_\mu j_\mu} \otimes E^{\mu'}_{i'_{\mu'} j'_{\mu'}}\Big) V^{(p)} \Big(E_{k_\nu l_\nu}^\nu \otimes E_{k'_{\nu'} l'_{\nu'}}^{\nu'}\Big)&= \Big(E^{\mu}_{i_\mu j_\mu}E^{\mu'}_{j'_{\mu'} i'_{\mu'} } \otimes \id \Big) V^{(p)} \Big(E_{k_\nu l_\nu}^\nu E_{l'_{\nu'} k'_{\nu'} }^{\nu'} \otimes \id \Big)\\
&=\delta^{\mu \mu'}\delta^{\nu \nu'} \delta_{j_\mu j'_{\mu'}}\delta_{l_\nu l'_{\nu'}}\Big(E^{\mu}_{i_\mu i'_\mu} \otimes \id\Big) V^{(p)} \Big(E_{k_\nu k'_\nu}^\nu \otimes \id\Big),\label{eq:80}
\end{align}
where in~\eqref{eq:80} we applied composition rule~\eqref{eqn:composition_of_e_operators}.

Proving redundancy of elements in~\eqref{eq:linspan2} is more involving. Our proof relies on the property that if for an arbitrary operator, $X$ one has $\tr(XX^\dagger)=0$, then $X=0$. Taking $X$ to be the second element from~\eqref{eq:linspan2}, we have:
\begin{align}
&\tr\left(\Big(E_{i_\mu \ j_\alpha}^{\mu} \otimes E_{k_\nu \ l_\beta}^{\nu}\Big) V^{(p-1)} \Big(E_{j'_{\alpha'} \ i'_{\mu'}}^{\ \ \ \ \mu'} \otimes E_{l'_{\beta'} \ j'_{\nu'}}^{\ \ \ \ \nu'}\Big)\Big(E_{i'_{\mu'} \ j'_{\alpha'}}^{\mu'} \otimes E_{j'_{\nu'} \ l'_{\beta'}}^{\nu'}\Big)V^{(p-1)}\Big(E_{j_\alpha i_\mu}^{\ \ \mu} \otimes E_{ l_\beta k_\nu}^{\ \ \nu}\Big)  \right)\\
&=\tr\left(\Big(E_{i_\mu \ j_\alpha}^{\mu} \otimes E_{k_\nu \ l_\beta}^{\nu}\Big) V^{(p-1)}\Big(E_{j'_{\alpha'}j'_{\alpha'}}^{\mu'}\otimes E_{l'_{\beta'}l'_{\beta'}}^{\nu'}\Big)
 V^{(p-1)}\Big(E_{j_\alpha i_\mu}^{\ \ \mu} \otimes E_{ l_\beta k_\nu}^{\ \ \nu}\Big)\right),\\
&=\tr\left(\Big(E_{j_\alpha j_\alpha}^{\mu} \otimes E_{l_\beta l_\beta}^{\nu}\Big) V^{(p-1)}\Big(E_{j'_{\alpha'}j'_{\alpha'}}^{\mu'}\otimes E_{l'_{\beta'}l'_{\beta'}}^{\nu'}\Big)
 V^{(p-1)}\right)\\
 &=\frac{1}{d^{2(p-1)}}\tr\left(V^{(p-1)}\Big(E_{j_\alpha j_\alpha}^{\mu} \otimes E_{l_\beta l_\beta}^{\nu}\Big)V^{(p-1)} V^{(p-1)}\Big(E_{j'_{\alpha'}j'_{\alpha'}}^{\mu'}\otimes E_{l'_{\beta'}l'_{\beta'}}^{\nu'}\Big)
 V^{(p-1)}\right)
\end{align}
In the first line, we used composition rule~\eqref{eqn:composition_of_e_operators}. In the second line, we wrote the product of composition in the PRIR notation~\eqref{eqn:e_prir_notation}. In the third by trace cyclicity, we again applied~\eqref{eqn:composition_of_e_operators}, and then wrote the result in the PRIR notation. Finally, in the last line, we used property $V^{(p-1)}V^{(p-1)}=d^{(p-1)}V^{(p-1)}$. For operators sandwiched by $V^{(p-1)}$ we can apply results of Theorem~\ref{thm:wba_element}. From this theorem, it follows that coefficients $a,b$ in decomposition~\eqref{eqn:operator_wba}, and given through~\eqref{eqn:coefficient_a},~\eqref{eqn:coefficient_b}, are nonzero only when 
\begin{align}
    \left( \alpha=\beta,\qquad \wedge \qquad j_\alpha=l_\beta\right), \qquad \wedge \qquad  \left(\alpha'=\beta', \qquad \wedge \qquad j'_{\alpha'}=l'_{\beta'}\right).
\end{align}
This, however, implies that the non-zero $X$ must be of the form:
\begin{align}
\label{eq:86}
   \Big(E_{i_\mu \ j_\alpha}^{\mu} \otimes E_{k_\nu \ j_\alpha}^{\nu}\Big) V^{(p-1)} \Big(E_{j'_{\alpha'} \ i'_{\mu'}}^{\ \ \ \ \mu'} \otimes E_{j'_{\alpha'} \ j'_{\nu'}}^{\ \ \ \ \nu'}\Big). 
\end{align}
In the next step, we show that the above expression is invariant with respect to changing the indices $j_\alpha, j'_{\alpha'}$. First, let us denote by $d_\alpha, d_{\alpha'}$ dimensions of the respective irreps, we know that $1\leq j_\alpha \leq d_{\alpha}$, and $1\leq j'_{\alpha'}\leq d_{\alpha'}$. We can insert in~\eqref{eq:86} indices $1_\alpha,1_{\alpha'}$ instead of $j_\alpha, j'_{\alpha'}$. Indices $1_\alpha,1_{\alpha'}$ label the first irreducible basis vectors in the irrep construction - for example, the first vectors following from the Young-Yamanouchi construction. In other words, vectors within the irrep $\alpha$ are labeled by numbers $1_\alpha, 2_\alpha,\ldots, d_\alpha$. From the construction, irreducible basis vectors within the irrep $\alpha$ are mutually orthonormal, and span a coordinate system in $d_\alpha-$dimensional space. This coordinate system can be rotated in such a way as to align the axis labeled $j_\alpha$ with the axis labeled $1_\alpha$. This can be done by an orthonormal matrix $O_\alpha$, i.e. matrix satisfying $O_\alpha O_\alpha^T=O_\alpha^T O_\alpha=\id_\alpha$, where $\id_\alpha$ is just standard identity matrix of dimension $d_\alpha$. The transformation is achieved by writing $O_\alpha |\alpha, j_\alpha\>=|\alpha, 1_\alpha\>$. In fact, the matrix $O_\alpha$ is a permutation matrix, composed of 0's and 1's only. Recalling that $E_{i_\mu \ j_\alpha}^{\mu}=|\mu,i_\mu\>\<\mu,\alpha, j_\alpha|$ and assumption that $O_\alpha$ acts only in the sector labeled by $\alpha$, we have:
\begin{align}
\Big(E_{i_\mu \ j_\alpha}^{\mu} \otimes E_{k_\nu \ j_\alpha}^{\nu}\Big) V^{(p-1)} \Big(E_{j'_{\alpha'} \ i'_{\mu'}}^{\ \ \ \ \mu'} \otimes E_{j'_{\alpha'} \ j'_{\nu'}}^{\ \ \ \ \nu'}\Big)&=\Big(E_{i_\mu \ j_\alpha}^{\mu} \otimes E_{k_\nu \ j_\alpha}^{\nu} O_\alpha O_\alpha^T\Big) V^{(p-1)} \Big(O_{\alpha'}^T O_{\alpha'} E_{j'_{\alpha'} \ i'_{\mu'}}^{\ \ \ \ \mu'} \otimes E_{j'_{\alpha'} \ j'_{\nu'}}^{\ \ \ \ \nu'}\Big)\\
&=\Big(E_{i_\mu \ j_\alpha}^{\mu}O_\alpha \otimes E_{k_\nu \ j_\alpha}^{\nu} O_\alpha \Big) V^{(p-1)} \Big( O_{\alpha'} E_{j'_{\alpha'} \ i'_{\mu'}}^{\ \ \ \ \mu'} \otimes O_{\alpha'}E_{j'_{\alpha'} \ j'_{\nu'}}^{\ \ \ \ \nu'}\Big)\\
&=\Big(E_{i_\mu \ 1_\alpha}^{\mu} \otimes E_{k_\nu \ 1_\alpha}^{\nu}  \Big) V^{(p-1)} \Big(  E_{1_{\alpha'} \ i'_{\mu'}}^{\ \ \ \ \mu'} \otimes E_{1_{\alpha'} \ j'_{\nu'}}^{\ \ \ \ \nu'}\Big).
\end{align}
The last expression is exactly the second element from~\eqref{eq:Linspan2}, so the proof is finished.
\end{proof}
Having the above proposition, we distinguish a subset of operators from~\eqref{eq:linspan1},~\eqref{eq:linspan2} with non-redundant elements.
\begin{definition}
\label{Def:opF}
Let $E^{\mu}_{i_\mu j_\mu}\in \mathbb{C}[\mathcal{S}_p]$ be irreducible matrix units for an irrep labelled by $\mu \vdash p$.  We define the following two families of  non-redundant operators
\begin{align}\label{eqn:F_operator}
    F_{i_\mu j_\mu \ i'_{\nu} j'_{\nu}}^{\mu \ \ \ \nu}(p) :
    &= \Big(E^{\mu}_{i_\mu j_\mu} \otimes \id\Big) V^{(p)} \Big(E_{j'_\nu i'_\nu}^\nu \otimes \id\Big),
\end{align}
\begin{align}\label{eqn:F_poperator}
    F_{i_\mu j_\nu \ i'_{\mu'} j'_{\nu'}}^{\mu\nu \ \ \mu'\nu'}(p-1,\alpha,\alpha') := \Big(E_{i_\mu \ 1_\alpha}^{\mu} \otimes E_{j_\nu \ 1_\alpha}^{\nu}\Big) V^{(p-1)} \Big(E_{1_{\alpha'} \ i'_{\mu'}}^{\ \ \ \ \mu'} \otimes E_{1_{\alpha'} \ j'_{\nu'}}^{\ \ \ \ \nu'}\Big),
\end{align}
where the labels $1_\alpha,1_{\alpha'}$ are arbitrary but consistent with irreps labeled by $\nu\vdash p, \mu'\vdash p$ respectively, and the identity operator $\id$ in~\eqref{eqn:F_operator} acts on the last $p$ systems. Additionally, the operators given in \eqref{eqn:F_poperator} are written in half-PRIR notation \eqref{eqn:e_halfprir_notation}. 
\end{definition}

\begin{remark}
\label{Remark12}
When the lower indices $i_\alpha,j_{\alpha'},k_\beta, l_{\beta'}$ in expression~\eqref{eqn:operator_wba} of Theorem~\ref{thm:wba_element} are fixed, then the explicit form of the coefficients in~\eqref{eqn:coefficient_a},~\eqref{eqn:coefficient_b} is much simpler. In particular, in the following text, we consider a special case where the operators $E_{1_{\alpha'} \ 1_\beta}^{\mu'} \otimes E_{1_{\alpha'} \ 1_\beta}^{\nu'}$ are sandwiched by the operator $V^{(p-1)}$. Namely, we have:
\begin{align}
                &V^{(p-1)} \Big(E_{1_{\alpha'} \ 1_\beta}^{\mu'} \otimes E_{1_{\alpha'} \ 1_\beta}^{\nu'}\Big)   V^{(p-1)}\\ &= a^{\mu' \nu'}(\alpha',\beta,\alpha',\beta,1_{\alpha'},1_\beta,1_{\alpha'},1_{\beta})\cdot V^{(p)} + b^{\mu' \nu'}(\alpha',\beta,\alpha',\beta,1_{\alpha'},1_\beta,1_{\alpha'},1_{\beta})\cdot V^{(p-1)}\\
                &=a^{\mu' \nu'}(\alpha',\beta)\cdot V^{(p)} + b^{\mu' \nu'}(\alpha',\beta)\cdot V^{(p-1)}.
            \end{align}
            Due to Theorem \ref{thm:wba_element} coefficients in the above decomposition have a form
            \begin{align}
                a^{\mu' \nu'}(\alpha',\beta,\alpha',\beta,1_{\alpha'},1_\beta,1_{\alpha'},1_{\beta})\equiv a^{\mu' \nu'}(\alpha',\beta)&= \frac{1}{d(d^2-1)}\Bigg(dm_{\mu'}\delta^{\mu' \nu'}-\frac{m_{\mu'}m_{\nu'}}{m_{\alpha'}}\delta^{\alpha' \beta}\delta_{1_{\alpha'} 1_\beta} \Bigg)\\
                &=\frac{1}{d(d^2-1)}\Bigg(dm_{\mu'}\delta^{\mu' \nu'}-\frac{m_{\mu'}m_{\nu'}}{m_{\alpha'}}\delta^{\alpha' \beta} \Bigg),\label{eq:wsp_a_general}
            \end{align}
            \begin{align}
                b^{\mu' \nu'}(\alpha',\beta,\alpha',\beta,1_{\alpha'},1_\beta,1_{\alpha'},1_{\beta})\equiv b^{\mu' \nu'}(\alpha',\beta)  &= \frac{1}{d(d^2-1)}\Bigg(d\frac{m_{\mu'}m_{\nu'}}{m_{\alpha'}}\delta^{\alpha' \beta}\delta_{1_{\alpha'} 1_\beta} - m_{\mu'}\delta^{\mu' \nu'} \Bigg)\\
                &=\frac{1}{d(d^2-1)}\Bigg(d\frac{m_{\mu'}m_{\nu'}}{m_{\alpha'}}\delta^{\alpha' \beta}- m_{\mu'}\delta^{\mu' \nu'} \Bigg).\label{eq:wsp_b_general}
            \end{align}
            In the two above equations, we get rid of $\delta_{1_{\alpha'} 1_\beta}$. This is consistent, since when $\alpha'=\beta$, we have $\delta^{\alpha'\alpha'}\delta_{1_{\alpha'}1_{\alpha'}}=1$, and for $\alpha'\neq \beta$, one has $\delta^{\alpha'\beta}\delta_{1_{\alpha'}1_{\beta}}=0$. In fact, we since we fix indices $i_\alpha, i_{\alpha'},k_\beta,k_{\beta'}$ in~\eqref{eqn:operator_wba} to $1_{\alpha'}, 1_\beta$ we will write later just $a^{\mu'\nu'}(\alpha',\beta),b^{\mu'\nu'}(\alpha',\beta)$.
\end{remark}

 For the operators given through \eqref{eqn:F_operator}, \eqref{eqn:F_poperator} we can formulate and prove the following composition composition rules: 
\begin{lemma}\label{lemma:composition_of_fs}
    The operators $F_{i_\mu j_\mu \ i'_{\nu} j'_{\nu}}^{\mu \ \ \ \ \nu}(p)$, $F_{i_\mu j_\nu \ i'_{\mu'} j'_{\nu'}}^{\mu\nu \ \ \mu'\nu'}(p-1,\alpha,\alpha')$ given through Definition~\ref{Def:opF}  the following relations hold:
    \begin{enumerate}[1)]
        \item \begin{align} 
        \label{eq:rel1}
        F_{i_\mu j_\mu \ i'_{\nu} j'_{\nu}}^{\mu \ \ \ \ \nu}(p) \cdot F_{k_{\Tilde{\mu}} l_{\Tilde{\mu}} \ k"_{\nu"} l"_{\nu"}}^{\Tilde{\mu}\ \ \ \ \nu"}(p) = m_\nu F^{\mu \ \ \ \nu"}_{i_\mu j_\mu \ k"_{\nu"}l"_{\nu"}}(p)\delta^{\nu\Tilde{\mu}}\delta_{i'_\nu k_{\Tilde{\mu}}}\delta_{j'_\nu l_\nu}
        \end{align}
        where $m_\nu$ is the multiplicity of irrep $\nu \vdash p$ in the Schur-Weyl duality.
        
        \item \begin{align}
        \label{eq:rel2}
            F_{i_\mu j_\mu \ i'_{\nu} j'_{\nu}}^{\mu \ \ \ \ \nu}(p) \cdot F_{k_{\Tilde{\mu}} l_{\Tilde{\nu}} \ k"_{\mu"} l"_{\nu"}}^{\Tilde{\mu}\Tilde{\nu} \ \ \mu"\nu"}(p-1,\beta,\beta') = \frac{m_\nu}{d}F^{\mu \ \ \ \mu"}_{i_\mu j_\mu \ k"_{\mu"} l"_{\mu"}}(p)\delta^{\nu\Tilde{\mu}}\delta^{\nu\Tilde{\nu}}\delta^{\mu"\nu"}\delta_{i'_\nu k_{\Tilde{\mu}}}\delta_{j'_\nu l_{\Tilde{\nu}}}
        \end{align}
         where $m_\nu$ is the multiplicity of irrep $\nu \vdash p$ in the Schur-Weyl duality.
        \item \begin{align}
        \label{eq:rel3}
         F_{i_\mu j_\nu \ i'_{\mu'} j'_{\nu'}}^{\mu\nu \ \ \mu'\nu'}(p-1,\alpha,\alpha') \cdot F_{k_{\Tilde{\mu}} l_{\Tilde{\mu}} \ k"_{\nu"} l"_{\nu"}}^{\Tilde{\mu}\ \ \ \ \nu"}(p) = \frac{m_{\mu'}}{d} F^{\mu \ \ \nu"}_{i_\mu j_\mu \ k"_{\nu"} l"_{\nu"}}(p)\delta^{\mu'\Tilde{\mu}} \delta^{\nu'\Tilde{\mu}} \delta^{\mu\nu}\delta_{i'_{\mu'} k_{\Tilde{\mu}}}\delta_{j'_{\nu'} l_{\Tilde{\nu}}}   
        \end{align}
         where $m_{\mu'}$ is the multiplicity of irrep $\mu' \vdash p$ in the Schur-Weyl duality.
        \item \begin{align}
        \label{eq:rel4}
           &F_{i_\mu j_\nu \ i'_{\nu'} j'_{\nu'}}^{\mu\nu \ \  \ \mu'\nu'}(p-1,\alpha,\alpha') \cdot F_{k_{\Tilde{\mu}} l_{\Tilde{\nu}} \ k"_{\mu"} l"_{\nu"}}^{\Tilde{\mu}\Tilde{\nu}  \ \ \mu" \nu"}(p-1,\beta,\beta')\\
        &= \Bigg(a^{\mu'\nu'}(\alpha',\beta) F^{\mu \ \ \ \mu"}_{i_\mu j_\mu \ k"_{\mu"} l"_{\mu"}}(p) \delta^{\mu\nu}\delta^{\mu"\nu"} + b^{\mu'\nu'}(\alpha',\beta) F^{\mu\nu \ \ \mu" \nu"}_{i_\mu j_\nu \ k"_{\mu"} l"_{\nu"}}(p-1,\alpha,\beta')\Bigg)\delta^{\mu'\Tilde{\mu}}\delta^{\nu'\Tilde{\nu}}\delta_{i'_{\mu'}k_{\Tilde{\mu}}}\delta_{j'_{\nu'}l_{\Tilde{\nu}}}
        \end{align}
         where the coefficients $a^{\mu'\nu'}(\alpha',\beta),b^{\mu'\nu'}(\alpha',\beta)\in \mathbb{R}$ are given by the Theorem \ref{thm:wba_element}.
    \end{enumerate}
   
    \begin{proof}
        The proof is based on straightforward calculations and will be split into four parts, each for each composition relation.
        \begin{enumerate}[1)]
            \item In the first part we calculate the following
            \begin{align}
                F_{i_\mu j_\mu \ i'_{\nu} j'_{\nu}}^{\mu \ \ \ \ \nu}(p) \cdot F_{k_{\Tilde{\mu}} l_{\Tilde{\mu}} \ k"_{\nu"} l"_{\nu"}}^{\Tilde{\mu}\ \ \ \ \nu"}(p) &=\Big(E^{\mu}_{i_\mu j_\mu} \otimes \id\Big) V^{(p)} \Big(E_{j'_\nu i'_\nu}^\nu \otimes \id\Big) \cdot \Big(E^{\Tilde{\mu}}_{k_{\Tilde{\mu}} l_{\Tilde{\mu}}} \otimes \id\Big) V^{(p)} \Big(E_{k"_{\nu"} l"_{\nu"}}^{\nu"} \otimes \id\Big)\\
                &= \Big(E^{\mu}_{i_\mu j_\mu} \otimes \id\Big) V^{(p)}\Big(E^{\nu}_{j'_\nu l_{\nu}} \otimes \id \Big) V^{(p)} \Big(E_{k"_{\nu"} l"_{\nu"}}^{\nu"} \otimes \id\Big) \delta^{\nu\Tilde{\mu}}\delta_{i'_\nu k_{\Tilde{\mu}}}
            \end{align}
            In the first line, we use the composition rule from~\eqref{eq:def_E}.
            By applying the Fact \ref{fact:sandwitch_of_v} to term between the operators $V^{(p)}$ we get 
            \begin{align}
                F_{i_\mu j_\mu \ i'_{\nu} j'_{\nu}}^{\mu \ \ \ \ \nu}(p) \cdot F_{k_{\Tilde{\mu}} l_{\Tilde{\mu}} \ k"_{\nu"} l"_{\nu"}}^{\Tilde{\mu}\ \ \ \ \nu"}(p) &= m_\nu\Big(E^{\mu}_{i_\mu j_\mu} \otimes \id\Big) V^{(p)}\Big(E_{k"_{\nu"} l"_{\nu"}}^{\nu"} \otimes \id\Big) \delta^{\nu\Tilde{\mu}}\delta_{i'_\nu k_{\Tilde{\mu}}}\delta_{j'_\nu l_\nu}\\
                &= m_\nu F^{\mu \ \ \ \nu"}_{i_\mu j_\mu \ k"_{\nu"}l"_{\nu"}}\delta^{\nu\Tilde{\mu}}\delta_{i'_\nu k_{\Tilde{\mu}}}\delta_{j'_\nu l_\nu}
            \end{align}
            \item The second part we calculate 
            \begin{align}
                &F_{i_\mu j_\mu \ i'_{\nu} j'_{\nu}}^{\mu \ \ \ \ \nu}(p) \cdot F_{k_{\Tilde{\mu}} l_{\Tilde{\nu}} \ k"_{\mu"} l"_{\nu"}}^{\Tilde{\mu}\Tilde{\nu} \ \ \mu"\nu"}(p-1,\beta,\beta') \\
                &=\Big(E^{\mu}_{i_\mu j_\mu} \otimes \id\Big) V^{(p)} \Big(E_{j'_\nu i'_\nu}^\nu \otimes \id\Big)\cdot  \Big(E_{k_{\Tilde{\mu}} \ 1_\beta}^{\Tilde{\mu}} \otimes E_{l_{\Tilde{\nu}} \ 1_\beta}^{\Tilde{\nu}}\Big) V^{(p-1)} \Big(E_{1_{\beta'} \ k"_{\mu"}}^{\ \ \ \ \mu"} \otimes E_{1_{\beta'} \ l"_{\nu"}}^{\ \ \ \ \nu"}\Big)\\
                &=\Big(E^{\mu}_{i_\mu j_\mu} \otimes \id\Big) V^{(p)} \Big(E_{1_{\alpha'} \ i'_{\nu}}^{\ \ \ \ \nu} \otimes E_{1_{\alpha'} \ j'_{\nu}}^{\ \ \ \ \nu}\Big)  \Big(E_{k_{\Tilde{\mu}} \ 1_\beta}^{\Tilde{\mu}} \otimes E_{l_{\Tilde{\nu}} \ 1_\beta}^{\Tilde{\nu}}\Big) V^{(p-1)} \Big(E_{1_{\beta'} \ k"_{\mu"}}^{\ \ \ \ \mu"} \otimes E_{1_{\beta'} \ l"_{\nu"}}^{\ \ \ \ \nu"}\Big)\\
                &= \Big(E^{\mu}_{i_\mu j_\mu} \otimes \id\Big) V^{(p)} \Big( E^{\nu}_{1_{\alpha'}1_\beta}\otimes E^{\nu}_{1_{\alpha'} 1_\beta} \Big)V^{(p-1)} \Big(E_{1_{\beta'} \ k"_{\mu"}}^{\ \ \ \ \mu"} \otimes E_{1_{\beta'} \ l"_{\nu"}}^{\ \ \ \ \nu"}\Big) \delta^{\nu\Tilde{\mu}}\delta^{\nu\Tilde{\nu}}\delta_{i'_\nu k_{\Tilde{\mu}}}\delta_{j'_\nu l_{\Tilde{\nu}}}
            \end{align}
            In the second line we apply Remark~\ref{R:Eij} with the composition rule~\eqref{eq:def_E} by writing $V^{(p)} \Big(E_{j'_\nu i'_\nu}^\nu \otimes \id\Big)=V^{(p)}\Big(E_{1_{\alpha'} \ i'_{\nu}}^{\ \ \ \ \nu} \otimes E_{1_{\alpha'} \ j'_{\nu}}^{\ \ \ \ \nu}\Big)$. In the third line we again apply the composition rule from~\eqref{eq:def_E}. In the fourth line we use  Fact \ref{fact:sandwitch_of_v} for operators sandwiched by $V^{(p)}$ and obtain:
            \begin{align}
                F_{i_\mu j_\mu \ i'_{\nu} j'_{\nu}}^{\mu \ \ \ \ \nu}(p) \cdot F_{k_{\Tilde{\mu}} l_{\Tilde{\nu}} \ k"_{\mu"} l"_{\nu"}}^{\Tilde{\mu}\Tilde{\nu} \ \ \mu"\nu"}(p-1,\beta,\beta') &=\frac{m_\nu}{d} \Big(E^{\mu}_{i_\mu j_\mu} \otimes \id\Big) V^{(p)}\Big(E_{1_{\beta'} \ k"_{\mu"}}^{\ \ \ \ \mu"} \otimes E_{1_{\beta'} \ l"_{\nu"}}^{\ \ \ \ \nu"}\Big) \delta^{\nu\Tilde{\mu}}\delta^{\nu\Tilde{\nu}}\delta_{i_\nu k_{\Tilde{\mu}}}\delta_{j'_\nu l_{\Tilde{\nu}}}\\
                &= \frac{m_\nu}{d} \Big(E^{\mu}_{i_\mu j_\mu} \otimes \id\Big) V^{(p)} \Big( E_{l"_{\mu"} \ k"_{\mu"}}^{\mu"} \otimes \id \Big)\delta^{\nu\Tilde{\mu}}\delta^{\nu\Tilde{\nu}}\delta^{\mu"\nu"}\delta_{i_\nu k_{\Tilde{\mu}}}\delta_{j'_\nu l_{\Tilde{\nu}}}\\
                &=\frac{m_\nu}{d}F^{\mu \ \ \ \mu"}_{i_\mu j_\mu \ k"_{\mu"} l"_{\mu"}}\delta^{\nu\Tilde{\mu}}\delta^{\nu\Tilde{\nu}}\delta^{\mu"\nu"}\delta_{i_\nu k_{\Tilde{\mu}}}\delta_{j'_\nu l_{\Tilde{\nu}}}
            \end{align}
            Finally, in the first line we apply the composition relation~\eqref{eq:def_E}.
            \item The proof of the third property goes analogously  to the proof of the second property.
            \item To prove the last property we write 
            \begin{align}
                &F_{i_\mu j_\nu \ i'_{\mu'} j'_{\nu'}}^{\mu\nu \ \  \ \mu'\nu'}(p-1,\alpha,\alpha') \cdot F_{k_{\Tilde{\mu}} l_{\Tilde{\nu}} \ k"_{\mu"} l"_{\nu"}}^{\Tilde{\mu}\Tilde{\nu}  \ \ \mu" \nu"}(p-1,\beta,\beta') =\\
                &= \Big(E_{i_\mu \ 1_\alpha}^{\mu} \otimes E_{j_\nu \ 1_\alpha}^{\nu}\Big) V^{(p-1)} \Big(E_{1_{\alpha'} \ i'_{\mu'}}^{\ \ \ \ \mu'} \otimes E_{1_{\alpha'} \ j'_{\nu'}}^{\ \ \ \ \nu'}\Big) \cdot \Big(E_{k_{\Tilde{\mu}} \ 1_\beta}^{\Tilde{\mu}} \otimes E_{l_{\Tilde{\nu}} \ 1_\beta}^{\Tilde{\nu}}\Big) V^{(p-1)} \Big(E_{1_{\beta'} \ k"_{\mu"}}^{\ \ \ \ \mu"} \otimes E_{1_{\beta'} \ l"_{\nu"}}^{\ \ \ \ \nu"}\Big)\\
                &= \Big(E_{i_\mu \ 1_\alpha}^{\mu} \otimes E_{j_\nu \ 1_\alpha}^{\nu}\Big) V^{(p-1)} \Big(E_{1_{\alpha'} \ 1_\beta}^{\mu'} \otimes E_{1_{\alpha'} \ 1_\beta}^{\nu'}\Big)   V^{(p-1)} \Big(E_{1_{\beta'} \ k"_{\mu"}}^{\ \ \ \ \mu"} \otimes E_{1_{\beta'} \ l"_{\nu"}}^{\ \ \ \ \nu"}\Big)\delta^{\mu'\Tilde{\mu}}\delta^{\nu'\Tilde{\nu}}\delta_{i'_{\mu'}k_{\Tilde{\mu}}}\delta_{j'_{\nu'}l_{\Tilde{\nu}}}\label{eqn:lema_part_of_sandwitch}
            \end{align}
            In the second line we apply the composition rule from~\eqref{eq:def_E}.
            In the third line, for the operators sandwiched by $V^{(p-1)}$ we apply the result of Theorem \ref{thm:wba_element} getting
            \begin{align}
                &F_{i_\mu j_\nu \ i'_{\nu'} j'_{\nu'}}^{\mu\nu \ \  \ \mu'\nu'}(p-1,\alpha,\alpha')  F_{k_{\Tilde{\mu}} l_{\Tilde{\nu}} \ k"_{\mu"} l"_{\nu"}}^{\Tilde{\mu}\Tilde{\nu}  \ \ \mu" \nu"}(p-1,\beta,\beta') =\\
                &= \Big(E_{i_\mu \ 1_\alpha}^{\mu} \otimes E_{j_\nu \ 1_\alpha}^{\nu}\Big) \Big( a^{\mu'\nu'}(\alpha',\beta) V^{(p)} + b^{\mu'\nu'}(\alpha',\beta) V^{(p-1)}\Big) \Big(E_{1_{\beta'} \ k"_{\mu"}}^{\ \ \ \ \mu"} \otimes E_{1_{\beta'} \ l"_{\nu"}}^{\ \ \ \ \nu"}\Big)\delta^{\mu'\Tilde{\mu}}\delta^{\nu'\Tilde{\nu}}\delta_{i'_{\mu'}k_{\Tilde{\mu}}}\delta_{j'_{\nu'}l_{\Tilde{\nu}}}\\
                &=\Bigg[ a^{\mu'\nu'}(\alpha',\beta)  \Big(E_{i_\mu \ 1_\alpha}^{\mu} \otimes E_{j_\nu \ 1_\alpha}^{\nu}\Big)V^{(p)}\Big(E_{1_{\beta'} \ k"_{\mu"}}^{\ \ \ \ \mu"} \otimes E_{1_{\beta'} \ l"_{\nu"}}^{\ \ \ \ \nu"}\Big)\\
                &+ b^{\mu'\nu'}(\alpha',\beta) \Big(E_{i_\mu \ 1_\alpha}^{\mu} \otimes E_{j_\nu \ 1_\alpha}^{\nu}\Big) V^{(p-1)}\Big(E_{1_{\beta'} \ k"_{\mu"}}^{\ \ \ \ \mu"} \otimes E_{1_{\beta'} \ l"_{\nu"}}^{\ \ \ \ \nu"}\Big) \Bigg]\delta^{\mu'\Tilde{\mu}}\delta^{\nu'\Tilde{\nu}}\delta_{i'_{\mu'}k_{\Tilde{\mu}}}\delta_{j'_{\nu'}l_{\Tilde{\nu}}}\\
                &=\Bigg( a^{\mu'\nu'}(\alpha',\beta) \Big( E_{i_\mu \ j_\mu}^{\mu} \otimes \id \Big)V^{(p)}\Big(E_{l"_{\mu"}\ k"_{\mu"}}^{\mu"} \otimes \id \Big)\delta^{\mu\nu}\delta^{\mu"\nu"}\\
                &+ b^{\mu'\nu'}(\alpha',\beta) \Big(E_{i_\mu \ 1_\alpha}^{\mu} \otimes E_{j_\nu \ 1_\alpha}^{\nu}\Big) V^{(p-1)}\Big(E_{1_{\beta'} \ k"_{\mu"}}^{\ \ \ \ \mu"} \otimes E_{1_{\beta'} \ l"_{\nu"}}^{\ \ \ \ \nu"} \Big)\Bigg)\delta^{\mu'\Tilde{\mu}}\delta^{\nu'\Tilde{\nu}}\delta_{i'_{\mu'}k_{\Tilde{\mu}}}\delta_{j'_{\nu'}l_{\Tilde{\nu}}}\\
                &= \Bigg(a^{\mu'\nu'}(\alpha',\beta) F^{\mu \ \ \ \mu"}_{i_\mu j_\mu \ k"_{\mu"} l"_{\mu"}}(p) \delta^{\mu,\nu}\delta^{\mu"\nu"} + b^{\mu'\nu'}(\alpha',\beta) F^{\mu\nu \ \ \mu" \nu"}_{i_\mu j_\nu \ k"_{\mu"} l"_{\nu"}}(p-1,\alpha,\beta')\Bigg)\delta^{\mu'\Tilde{\mu}}\delta^{\nu'\Tilde{\nu}}\delta_{i'_{\mu'}k_{\Tilde{\mu}}}\delta_{j'_{\nu'}l_{\Tilde{\nu}}}
            \end{align}
        \end{enumerate}
        In the third line, we apply the 'ping-pong' trick~\eqref{eq:BellV'} with the composition rule~\eqref{eq:def_E}. This completes the proof.
    \end{proof}
\end{lemma}

The relations of composition given in Lemma \ref{lemma:composition_of_fs} can be simplified in certain situations. Observe that the non-zero results of compositions of operators from \eqref{eqn:F_operator}, \eqref{eqn:F_poperator} can be obtained only when the inner indices on the left-hand sides in the composition relations~\eqref{eq:rel1}-\eqref{eq:rel4}  are the same. Indeed, we can write the following remark collecting these observations.

\begin{remark}
\label{remark_simplification}
    The composition relations~\eqref{eq:rel1}-\eqref{eq:rel4} from Lemma \ref{lemma:composition_of_fs} can be simplified taking into account only indices that give non-zero operators in composition:
    \begin{enumerate}[1)]
        \item \begin{align}
            F_{i_\mu j_\mu \ i'_{\nu} j'_{\nu}}^{\mu \ \ \ \ \nu}(p) \cdot F_{i'_{\nu} j'_{\nu} \ k"_{\nu"} l"_{\nu"}}^{\nu\ \ \ \ \nu"}(p) = m_\nu F^{\mu \ \ \ \nu"}_{i_\mu j_\mu \ k"_{\nu"}l"_{\nu"}} 
        \end{align}
        \item \begin{align}
            F_{i_\mu j_\mu \ i'_{\nu} j'_{\nu}}^{\mu \ \ \ \ \nu}(p) \cdot F_{i'_{\nu} j'_{\nu} \ k"_{\mu"} l"_{\nu"}}^{\nu\nu \ \ \mu"\nu"}(p-1,\beta,\beta') = \frac{m_\nu}{d}F^{\mu \ \ \ \mu"}_{i_\mu j_\mu \ k"_{\mu"} l"_{\mu"}}(p)\delta^{\mu"\nu"} 
        \end{align}
        \item \begin{align}
        F_{i_\mu j_\nu \ k_{\Tilde{\mu}} l_{\Tilde{\mu}}}^{\mu\nu \ \ \Tilde{\mu}\Tilde{\mu}}(p-1,\alpha,\alpha') \cdot F_{k_{\Tilde{\mu}} l_{\Tilde{\mu}} \ k"_{\nu"} l"_{\nu"}}^{\Tilde{\mu}\ \ \ \ \nu"}(p) = \frac{m_{\mu'}}{d} F^{\mu \ \ \nu"}_{i_\mu j_\mu \ k"_{\nu"} l"_{\nu"}} (p)\delta^{\mu\nu}
        \end{align}
        \item \begin{align}
           F_{i_\mu j_\nu \ i'_{\nu'} j'_{\nu'}}^{\mu\nu \ \  \ \mu'\nu'}(p-1,\alpha,\alpha') \cdot F_{i'_{\nu'} j'_{\nu'} \ k"_{\mu"} l"_{\nu"}}^{\mu'\nu'  \ \ \mu" \nu"}(p-1,\beta,\beta') &= a^{\mu'\nu'}(\alpha',\beta) F^{\mu \ \ \ \mu"}_{i_\mu j_\mu \ k"_{\mu"} l"_{\mu"}}(p) \delta^{\mu,\nu}\delta^{\mu"\nu"} \\
           &+ b^{\mu'\nu'}(\alpha',\beta) F^{\mu\nu \ \ \mu" \nu"}_{i_\mu j_\nu \ k"_{\mu"} l"_{\nu"}}(p-1,\alpha,\beta') 
        \end{align}
    \end{enumerate}
    where $m_\nu,m_{\mu'}$ are respective multiplicities of irreps $\nu,\mu'\vdash p$ in the Schur-Weyl duality, and the coefficients $a^{\mu'\nu'}(\alpha',\beta),b^{\mu'\nu'}(\alpha',\beta)\in \mathbb{R}$ are given by the Remark \ref{Remark12}.
\end{remark}

\section{Construction of irreducible matrix units in the ideal \texorpdfstring{$\mathcal{M}^{(p)}$}{Lg}}\label{sec:ideal_m}
In this section, we provide the irreducible matrix units for the maximal ideal $\mathcal{M}^{(p)}$. These matrix units have been constructed in~\cite{StudzinskiIEEE22}. However, we review the main result with the notation suites for the paper's self-consistency and reader convenience.

\begin{theorem}\label{thm:basis_for_max_ideal}
    Let $F_{i_\mu j_\mu \ i'_{\nu} j'_{\nu}}^{\mu \ \ \ \nu}(p)$ be operators given through~\eqref{eqn:F_operator} in Definition~\ref{Def:opF}. The irreducible matrix units in the maximal ideal $\mathcal{M}^{(p)}$ is given by the following set of operators 
    \begin{align}\label{eqn:def_of_Gp}
        \mathcal{G}_{i_\mu j_\mu \ i'_\nu j'_\nu }^{\mu  \ \ \ \nu}(p) :=\frac{1}{\sqrt{m_\mu m_\nu}} F_{i_\mu j_\mu \ i'_{\nu} j'_{\nu}}^{\mu \ \ \ \nu}(p)=\frac{1}{\sqrt{m_\mu m_\nu}} \Big(E_{i_\mu j_\mu}^{\mu} \otimes\id \Big) V^{(p)} \Big(E_{j'_\nu i'_\nu}^{\nu}\otimes \id\Big) 
    \end{align}
    satisfying the following composition rule
    \begin{align}
    \label{eq:ortG}
                \mathcal{G}_{i_\mu j_\mu \ i'_\nu j'_\nu }^{\mu  \ \ \ \nu}(p) \cdot \mathcal{G}_{k_{\mu'} l_{\mu'} \ k'_{\nu'} l'_{\nu'} }^{\mu'  \ \ \ \nu'}(p) =  \mathcal{G}^{\mu \ \ \ \nu'}_{i_\mu j_\mu \ k'_{\nu'}l'_{\nu'}}(p)\delta_{i'_\nu k_{\mu'}}\delta_{j'_\nu l_\nu}\delta^{\mu' \nu}
            \end{align}
    where $m_\mu,m_\nu$ are the multiplicities of the respective irreps $\mu,\nu \vdash \mu $ in the Schur-Weyl duality.
    \begin{proof}
        We split the proof into two parts:
        \begin{itemize}
            \item Showing that operators are orthogonal i.e. relation~\eqref{eq:ortG} holds. Writing the left-hand side of~\eqref{eq:ortG} explicitly and using orthogonality relation for operator $E^{\mu}_{ij}$ we have
            \begin{align}
                \mathcal{G}_{i_\mu j_\mu \ i'_\nu j'_\nu }^{\mu  \ \ \ \nu}&(p) \cdot \mathcal{G}_{k_{\mu'} l_{\mu'} \ k'_{\nu'} l'_{\nu'} }^{\mu' \ \ \ \nu'}(p)\\
                &= \frac{1}{\sqrt{m_\mu m_\nu}}\frac{1}{\sqrt{m_{\mu'} m_{\nu'}}} \Big(E_{i_\mu j_\mu}^{\mu} \otimes\id \Big) V^{(p)} \Big(E_{j'_\nu i'_\nu}^{\nu}\otimes \id\Big) \cdot \Big(E_{k_{\mu'} l_{\mu'}}^{\mu'} \otimes\id \Big) V^{(p)} \Big(E_{l'_{\nu'} k'_{\nu'}}^{\nu'}\otimes \id\Big)\\
                &= \frac{1}{\sqrt{m_\mu m_\nu}}\frac{1}{\sqrt{m_{\mu'} m_{\nu'}}} \Big(E_{i_\mu j_\mu}^{\mu} \otimes\id \Big) V^{(p)} \Big(E_{j'_\nu i'_\nu}^{\nu} E_{k_{\mu'} l_{\mu'}}^{\mu'}\otimes \id\Big)   V^{(p)} \Big(E_{l'_{\nu'} k'_{\nu'}}^{\nu'}\otimes \id\Big)\\
                &= \frac{1}{m_\nu\sqrt{m_\mu  m_{\nu'}}} \Big(E_{i_\mu j_\mu}^{\mu} \otimes\id \Big) V^{(p)} \Big(E_{j'_\nu l_{\nu}}^{\nu}\otimes \id\Big)   V^{(p)} \Big(E_{l'_{\nu'} k'_{\nu'}}^{\nu'}\otimes \id\Big)\delta_{i'_\nu k_{\mu'}}\delta^{\mu' \nu}
            \end{align}
            Now, applying Fact \ref{fact:sandwitch_of_v} to operators sandwiched by $V^{(p)}$ together with the trace of operator $E^{\nu}_{j'_\nu l_\nu}$ given by relation \eqref{eqn:composition_of_e_operators}, we reduce to
            \begin{align}
               \mathcal{G}_{i_\mu j_\mu \ i'_\nu j'_\nu }^{\mu  \ \ \ \nu}(p) \cdot \mathcal{G}_{k_{\mu'} l_{\mu'} \ k'_{\nu'} l'_{\nu'} }^{\mu' \ \ \ \nu'}(p) &= \frac{m_\nu}{m_\nu\sqrt{m_\mu  m_{\nu'}}} \Big(E_{i_\mu j_\mu}^{\mu} \otimes\id \Big)   V^{(p)} \Big(E_{l'_{\nu'} k'_{\nu'}}^{\nu'}\otimes \id\Big)\delta_{i'_\nu k_{\mu'}}\delta_{j'_\nu l_\nu}\delta^{\mu' \nu}\\
               &=  \frac{1}{\sqrt{m_\mu  m_{\nu'}}} \Big(E_{i_\mu j_\mu}^{\mu} \otimes\id \Big)   V^{(p)} \Big(E_{l'_{\nu'} k'_{\nu'}}^{\nu'}\otimes \id\Big)\delta_{i'_\nu k_{\mu'}}\delta_{j'_\nu l_\nu}\delta^{\mu' \nu}\\
               &= \mathcal{G}^{\mu \ \ \ \nu'}_{i_\mu j_\mu \ k'_{\nu'}l'_{\nu'}}(p)\delta_{i'_\nu k_{\mu'}}\delta_{j'_\nu l_\nu}\delta^{\mu' \nu}
            \end{align}
            \item Showing that element $V^{(p)}$ generating the ideal $\mathcal{M}^{(p)}$ can be expressed as a linear combination of irreducible matrix units $\mathcal{G}_{i_\mu j_\mu \ i'_\nu j'\nu }^{\mu  \ \ \ \nu}(p)$. Indeed, since $\id = \sum_\mu P_\mu$, and $P_{\mu}=\sum_{i_\mu}E^{\mu}_{i_\mu i_\mu}$, we observe that 
            \begin{align}
                V^{(p)} &= \sum_{\mu,\nu\vdash p}(P_\mu\otimes \id) V^{(p)}(P_\nu\otimes \id) = \sum_{\mu,\nu\vdash p} \sum_{i_\mu, j_\nu} \left(E^{\mu}_{i_\mu i_\mu}\otimes \id \right) V^{(p)} \left(E^{\nu}_{j_\nu j_\nu}\otimes \id\right)\\
                &=\sum_{\mu,\nu\vdash p} \sum_{i_\mu, j_\nu} \sqrt{m_\mu m_\nu}\Bigg(\frac{1}{\sqrt{m_\mu m_\nu}} \Big( E^{\mu}_{i_\mu i_\mu}\otimes \id \Big) V^{(p)} \Big(E^{\nu}_{j_\nu j_\nu} \otimes \id\Big)\Bigg)\\
                &=\sum_{\mu,\nu\vdash p} \sum_{i_\mu,j_\nu}\sqrt{m_\mu m_\nu} \mathcal{G}^{\mu \ \ \ \nu}_{i_\mu i_\mu \ j_\nu j_\nu}(p).
            \end{align}
       \end{itemize}
    \end{proof}
\end{theorem}

\section{Construction of irreducible matrix units in the ideal \texorpdfstring{$\mathcal{M}^{(p-1)}$}{Lg}}\label{sec:ideal_m-1}
In this section, we provide a construction of irreducible matrix units for the lower ideal $\mathcal{M}^{(p-1)}$.  

\begin{theorem}\label{thm:lower_ideal_operator_G}
    Let $F^{\mu \ \ \ \mu'}_{i_\mu j_{\mu} \ i'_{\mu'} j'_{\mu'}}(p), F^{\mu \nu \ \ \mu' \nu'}_{i_\mu j_\nu \ i'_{\mu'} j'_{\nu'}}(p-1,\alpha,\alpha')$ be the operators given through Definition~\ref{Def:opF}. Then the following operators span the ideal $\mathcal{M}^{(p-1)}$:
    \begin{align}
        H^{\mu \nu \ \ \mu' \nu'}_{i_\mu j_\nu \ i'_{\mu'} j'_{\nu'}}(p-1,\alpha,\alpha') = 
        d F^{\mu \nu \ \ \mu' \nu'}_{i_\mu j_\nu \ i'_{\mu'} j'_{\nu'}}(p-1,\alpha,\alpha')-F^{\mu \ \ \ \mu'}_{i_\mu j_{\mu} \ i'_{\mu'} j'_{\mu'}}(p)\delta^{\mu \nu}\delta^{\mu' \nu'},\label{eqn:def_of_op_H}
    \end{align}
     satisfying the following composition rule
     \begin{align}
     \label{eq:comp_H}
        H^{\mu \nu \ \ \mu' \nu'}_{i_\mu j_\nu \ i'_{\mu'} j'_{\nu'}}(p-1,\alpha,\alpha')\cdot  H^{\Tilde{\mu} \Tilde{\nu} \ \ \mu" \nu"}_{k_{\Tilde{\mu}} l_{\Tilde{\nu}} \ k'_{\mu"} l'_{\nu"}}(p-1,\beta,\beta') =  db^{\mu'\nu'}(\alpha',\beta) H^{\mu,\nu \ \ \mu"\nu"}_{i_\mu j_\nu \ k'_{\mu"}l'_{\nu"}}(p-1,\alpha,\beta')\delta^{\mu' \Tilde{\mu}}\delta^{\nu' \Tilde{\nu}}\delta_{i'_{\mu'} k_{\Tilde{\mu}}}\delta_{j'_{\mu'} l_{\Tilde{\nu}}}
     \end{align}
     where the coefficient $b^{\mu'\nu'}(\alpha',\beta)$ can be evaluated by expression \eqref{eqn:coefficient_b} from Theorem~\ref{thm:wba_element}.
     \end{theorem}
     \begin{proof}
        The proof proceeds similarly to the proof of Lemma \ref{lemma:composition_of_fs}, however, through proof we will restrict ourselves to inner indices of composition that are non-zero (see Remark \ref{remark_simplification}). Namely, we write the following:
    \begin{align}
        H&^{\mu \nu \ \ \mu' \nu'}_{i_\mu j_\nu \ i'_{\mu'} j'_{\nu'}}(p-1,\alpha,\alpha')\cdot H^{\mu' \nu' \ \ \mu" \nu"}_{i'_{\mu'} j'_{\nu'} \ k'_{\mu"} l'_{\nu"}}(p-1,\beta,\beta')\\
        &=\Bigg(d F^{\mu \nu \ \ \mu' \nu'}_{i_\mu j_\nu \ i'_{\mu'} j'_{\nu'}}(p-1,\alpha,\alpha') - F^{\mu \ \ \ \mu'}_{i_\mu j_{\mu} \ i'_{\mu'} j'_{\mu'}}(p)\delta^{\mu,\nu}\delta^{\mu' \nu'}\Bigg) \Bigg(d F^{\mu' \nu' \ \ \mu" \nu"}_{i'_{\mu'} j'_{\nu'} \ k'_{\mu"} l'_{\nu"}}(p-1,\beta,\beta') - F^{\mu' \ \ \ \mu"}_{i'_{\mu'} j'_{\mu'} \ k'_{\mu"} l'_{\mu"}}(p)\delta^{\mu' \nu'}\delta^{\mu" \nu"}\Bigg)\\
        &= d^2  F^{\mu \nu \ \ \mu' \nu'}_{i_\mu j_\nu \ i'_{\mu'} j'_{\nu'}}(p-1,\alpha,\alpha')F^{\mu' \nu' \ \ \mu" \nu"}_{i'_{\mu'} j'_{\nu'} \ k'_{\mu"} l'_{\nu"}}(p-1,\beta,\beta') - d F^{\mu \nu \ \ \mu' \nu'}_{i_\mu j_\nu \ i'_{\mu'} j'_{\nu'}}(p-1,\alpha,\alpha')F^{\mu' \ \ \ \mu"}_{i'_{\mu'} j'_{\mu'} \ k'_{\mu"} l'_{\mu"}}(p)\delta^{\mu' \nu'}\delta^{\mu" \nu"}\\
        &- d F^{\mu \ \ \ \mu'}_{i_\mu j_{\mu} \ i'_{\mu'} j'_{\mu'}}(p)F^{\mu' \nu' \ \ \mu" \nu"}_{i'_{\mu'} j'_{\nu'} \ k'_{\mu"} l'_{\nu"}}(p-1,\beta,\beta')\delta^{\mu \nu}\delta^{\mu' \nu'} + F^{\mu \ \ \ \mu'}_{i_\mu j_{\mu} \ i'_{\mu'} j'_{\mu'}}(p)F^{\mu' \ \ \ \mu"}_{i'_{\mu'} j'_{\mu'} \ k'_{\mu"} l'_{\mu"}}(p)\delta^{\mu \nu}\delta^{\mu' \nu'}\delta^{\mu" \nu"}
    \end{align}
     Using the Lemma \ref{lemma:composition_of_fs} we reduce the above expression to 
     \begin{align}
         H&^{\mu \nu \ \ \mu' \nu'}_{i_\mu j_\nu \ i'_{\mu'} j'_{\nu'}}(p-1,\alpha,\alpha')\cdot  H^{\mu' \nu' \ \ \mu" \nu"}_{i'_{\mu'} j'_{\nu'} \ k'_{\mu"} l'_{\nu"}}(p-1,\beta,\beta')\\
         &=\Bigl(ad^2 F^{\mu \ \ \ \mu"}_{i_\mu j_\mu \ k_{\mu"} l_{\mu"}}(p)\delta^{\mu \nu}\delta^{\mu" \nu"}+bd^2 F^{\mu\nu \ \ \mu"\nu"}_{i_\mu j_\nu \ k'_{\mu"} l'_{\nu"}}(p-1,\alpha,\beta') \Bigr) - m_{\mu'}F^{\mu \ \ \ \mu"}_{i_\mu j_\mu \ k'_{\mu"} l'_{\mu"}}(p) \delta^{\mu' \nu'}\delta^{\mu" \nu"}\delta^{\mu \nu}\\
         &- d(\Tilde{a}d+\Tilde{b})F^{\mu \ \ \ \mu"}_{i_\mu j_\mu \ k'_{\mu"} l'_{\mu"}}(p)\delta^{\mu \nu}\delta^{\mu' \nu'}\delta^{\mu" \nu"}+ m_{\mu'}F^{\mu \ \ \ \mu"}_{i_\mu j_\mu \ k'_{\mu"} l'_{\mu"}}(p)\delta^{\mu \nu}\delta^{\mu' \nu'}\delta^{\mu" \nu"}\\
         &= \Bigl(ad^2 F^{\mu \ \ \ \mu"}_{i_\mu j_\mu \ k'_{\mu"} l'_{\mu"}}(p)\delta^{\mu \nu}\delta^{\mu" \nu"}+bd^2 F^{\mu\nu \ \ \mu"\nu"}_{i_\mu j_\nu \ k'_{\mu"} l'_{\nu"}}(p-1,\alpha,\beta') \Bigr)- m_{\mu'}F^{\mu \ \ \ \mu"}_{i_\mu j_\mu \ k'_{\mu"} l'_{\mu"}}(p)\delta^{\mu,\nu}\delta^{\mu'\nu'}\delta^{\mu" \nu"}\label{eqn:comparing_coeff_ab_ab}
     \end{align}
            where term $\Tilde{a}d+\Tilde{b}$ comes from the composition of $F^{\mu \ \ \ \mu'}_{i_\mu j_{\mu} \ i'_{\mu'} j'_{\mu'}}(p)\cdot F^{\mu' \nu' \ \ \mu" \nu"}_{i'_{\mu'} j'_{\nu'} \ k'_{\mu"} l'_{\nu"}}(p-1,\beta,\beta')$ where additionally we used Lemma \ref{lemma:ad+b=m_mu}, and  we reduced the terms $m_{\mu'}F^{\mu \ \ \ \mu"}_{i_\mu j_\mu \ k'_{\mu"} l'_{\mu"}}(p)\delta^{\mu \nu}\delta^{\mu' \nu'}\delta^{\mu" \nu"}$. By combining operators $F^{\mu \nu \ \ \mu' \nu'}_{i_\mu j_\nu \ i'_{\mu'} j'_{\nu'}}(p-1,\alpha,\alpha')\cdot F^{\mu'\nu' \ \ \mu" \nu"}_{i'_{\mu'} j'_{\nu'} \ k'_{\mu"} l'_{\nu"}}(p-1,\beta,\beta')$ which is equal to writing the eq.~\eqref{eqn:lema_part_of_sandwitch} from Lemma \ref{lemma:composition_of_fs} and by Theorem \ref{thm:wba_element}, we get that the middle term sandwiched by $V^{(p-1)}$ is equal to following linear combination
            \begin{align}
                V^{(p-1)} \Big(E_{1_{\alpha'} \ 1_\beta}^{\mu'} \otimes E_{1_{\alpha'} \ 1_\beta}^{\nu'}\Big)   V^{(p-1)} = a^{\mu' \nu'}(\alpha',\beta)\cdot V^{(p)} + b^{\mu' \nu'}(\alpha',\beta)\cdot V^{(p-1)},
            \end{align}
            where with respect to Remark \ref{Remark12} coefficients $a^{\mu' \nu'}(\alpha',\beta), b^{\mu' \nu'}(\alpha',\beta)$ have a form
            \begin{align}
                a \equiv a^{\mu' \nu'}(\alpha',\beta)=\frac{1}{d(d^2-1)}\Bigg(dm_{\mu'}\delta^{\mu' \nu'}-\frac{m_{\mu'}m_{\nu'}}{m_{\alpha'}}\delta^{\alpha' \beta} \Bigg),\label{eqn:wsp_a_general}
            \end{align}
            \begin{align}
                b\equiv b^{\mu' \nu'}(\alpha',\beta)=\frac{1}{d(d^2-1)}\Bigg(d\frac{m_{\mu'}m_{\nu'}}{m_{\alpha'}}\delta^{\alpha' \beta}- m_{\mu'}\delta^{\mu' \nu'} \Bigg).\label{eqn:wsp_b_general}
            \end{align}
            
            Now observe that, by Lemma \ref{lemma:ad+b=m_mu} we can write
            \begin{align}
                ad+b = \frac{m_{\mu'}}{d}\delta^{\mu' \nu'} \ \rightarrow \ ad^2 = m_{\mu'}\delta^{\mu' \nu'}-db, \label{eqn:replacement_of_term_ad2}
            \end{align}
           Using \eqref{eqn:replacement_of_term_ad2} we can rewrite line \eqref{eqn:comparing_coeff_ab_ab} in the following manner
            \begin{align}
            H&^{\mu \nu \ \  \mu' \nu'}_{i_\mu j_\nu \ i'_{\mu'} j'_{\nu'}}(p-1,\alpha,\alpha')\cdot  H^{\mu' \nu' \ \ \mu" \nu"}_{i'_{\mu'} j'_{\nu'} \ k'_{\mu"} l'_{\nu"}}(p-1,\beta,\beta')\\    
            &= \Bigl(ad^2 F^{\mu \ \ \ \mu"}_{i_\mu j_\mu \ k'_{\mu"} l'_{\mu"}}(p)\delta^{\mu \nu}\delta^{\mu" \nu"}+bd^2 F^{\mu\nu \ \ \mu"\nu"}_{i_\mu j_\nu \ k'_{\mu"} l'_{\nu"}}(p-1,\alpha,\beta') \Bigr)- m_{\mu'}F^{\mu \ \ \ \mu"}_{i_\mu j_\mu \ k'_{\mu"} l'_{\mu"}}(p)\delta^{\mu \nu}\delta^{\mu' \nu'}\delta^{\mu" \nu"}\\
            &= \Bigl((m_{\mu'}\delta^{\mu' \nu'}-db) F^{\mu \ \ \ \mu"}_{i_\mu j_\mu \ k'_{\mu"} l'_{\mu"}}(p)\delta^{\mu \nu}\delta^{\mu" \nu"}+bd^2 F^{\mu\nu \ \ \mu"\nu"}_{i_\mu j_\nu \ k'_{\mu"} l'_{\nu"}}(p-1,\alpha,\beta') \Bigr)- m_{\mu'}F^{\mu \ \ \ \mu"}_{i_\mu j_\mu \ k'_{\mu"} l'_{\mu"}}(p)\delta^{\mu \nu}\delta^{\mu' \nu'}\delta^{\mu" \nu"}\\
            &=db\Bigg(d F^{\mu\nu \ \ \mu"\nu"}_{i_\mu j_\nu \ k'_{\mu"} l'_{\nu"}}(p-1,\alpha,\beta') -  F^{\mu \ \ \ \mu"}_{i_\mu j_\mu \ k'_{\mu"} l'_{\mu"}}(p)\delta^{\mu \nu}\delta^{\mu" \nu"}\Bigg)\\
            &= db H^{\mu \nu \ \ \mu"\nu"}_{i_\mu j_\nu \ k'_{\mu"}l'_{\mu"}}(p-1,\alpha,\beta'),
     \end{align}
     where coefficient $b\equiv b^{\mu'\nu'}(\alpha',\beta)$ has form given in \eqref{eqn:wsp_b_general}. In the final step, we show that using our definition of the operators $H^{\mu \nu \ \mu' \nu'}_{i_\mu j_\nu \ i'_{\mu'} j'_{\nu'}}(p-1,\alpha,\alpha')$ we can construct operators from Definition~\ref{Def:opF}. Indeed, we have
     \begin{align}
    H^{\mu \nu \ \ \mu' \nu'}_{i_\mu j_\nu \ i'_{\mu'} j'_{\nu'}}(p-1,\alpha,\alpha') &= 
        d F^{\mu \nu \ \ \mu' \nu'}_{i_\mu j_\nu \ i'_{\mu'} j'_{\nu'}}(p-1,\alpha,\alpha')-F^{\mu \ \ \ \mu'}_{i_\mu j_{\mu} \ i'_{\mu'} j'_{\mu'}}(p)\delta^{\mu \nu}\delta^{\mu' \nu'}\\
        &=dF^{\mu \nu \ \ \mu' \nu'}_{i_\mu j_\nu \ i'_{\mu'} j'_{\nu'}}(p-1,\alpha,\alpha')-\sqrt{m_\mu m_{\mu'}}\mathcal{G}^{\mu \ \ \ \mu'}_{i_\mu j_{\mu} \ i'_{\mu'} j'_{\mu'}}(p)\delta^{\mu \nu}\delta^{\mu' \nu'}.
     \end{align}
     From the above we get
     \begin{align}
         F^{\mu \nu \ \ \mu' \nu'}_{i_\mu j_\nu \ i'_{\mu'} j'_{\nu'}}(p-1,\alpha,\alpha')=\frac{1}{d}\left(H^{\mu \nu \ \ \mu' \nu'}_{i_\mu j_\nu \ i'_{\mu'} j'_{\nu'}}(p-1,\alpha,\alpha')+\sqrt{m_\mu m_{\mu'}}\mathcal{G}^{\mu \ \ \ \mu'}_{i_\mu j_{\mu} \ i'_{\mu'} j'_{\mu'}}(p)\delta^{\mu \nu}\delta^{\mu' \nu'}\right),
     \end{align}
     which is the desired expression. This finishes the proof.
     \end{proof}

\begin{remark}
Notice that the operators $H^{\mu \nu \ \ \mu' \nu'}_{i_\mu j_\nu \ i'_{\mu'} j'_{\nu'}}(p-1,\alpha,\alpha')$ from the above theorem, due to relation~\eqref{eq:comp_H} are almost irreducible matrix units in the ideal $\mathcal{M}^{(p-1)}$. It means they satisfy an analogous composition law to the relation from~\eqref{eqn:composition_of_e_operators} for the matrix units $E_{i_\mu j_\mu}^{\mu}$ for the algebra $\mathbb{C}[\s_p]$. Firstly, the introduced basis is overcomplete. Secondly, for $\mu'=\widetilde{\mu}$ and $\nu'=\widetilde{\nu}$ they are not orthogonal in indices $\alpha', \beta$. We shall address these problems further in this section.
\end{remark}

Before we proceed with the construction of irreducible matrix units for the ideal $\mathcal{M}^{(p-1)}$, we express the operator $V^{(p-1)}$, which generates $\mathcal{M}^{(p-1)}$, according to~\eqref{eqn:ideal_m} by the operators $H^{\mu \nu \ \ \mu' \nu'}_{i_\mu j_\nu \ i'_{\mu'} j'_{\nu'}}(p-1,\alpha,\alpha')$ and irreducible matrix units for the ideal $\mathcal{M}^{(p)}$ from Theorem~\ref{thm:basis_for_max_ideal}.

     \begin{lemma}\label{lemma:operator_H_generates_ideal_mp-1}
       The operator $V^{(p-1)}$  generating the ideal $\mathcal{M}^{(p-1)}$ from~\eqref{eqn:ideal_m} can be written in terms of  the operators from Theorem~\ref{thm:basis_for_max_ideal} and Theorem~\ref{thm:lower_ideal_operator_G} as follows:
       \begin{align}
           V^{(p-1)}=\sum_{\alpha,\beta\vdash p-1} \sum_{\substack{\mu,\mu'=\alpha+\Box\\\nu,\nu'=\beta+\Box}}  \Bigg(
                H^{\mu \mu' \ \ \nu \nu'}_{i_\mu i_{\mu'}  \ j_\nu j_{\nu'}}(p-1,\alpha,\beta)+ \sqrt{m_\mu m_\nu}\mathcal{G}^{\mu \nu \ \ \mu \nu}_{i_\mu j_\nu \  i_\mu j_\nu}(p)\delta^{\mu \mu'}\delta^{\nu \nu'}\Bigg).
       \end{align}
        Operators $H^{\mu\nu \ \ \mu'\nu'}_{i_\mu j_\mu \ i'_{\mu'} j'_{\nu'}}(p-1,\alpha,\alpha')$ span the ideal $\mathcal{M}^{(p-1)}$, while $\mathcal{G}^{\mu \nu \ \ \mu \nu}_{i_\mu j_\nu \ i_\mu j_\nu}(p)$ are irreducible matrix units for the ideal $\mathcal{M}^{(p)}$.
\begin{proof}
    From the definition of the ideal $\mathcal{M}^{(p-1)}$ we know it is generated by the operator $V^{(p-1)}$. To show that the operators $H_{i_\mu j_\nu \ i_{\mu'} j_{\nu'}}^{\mu \nu \ \  \mu'\nu'}(p-1,\alpha,\alpha')$ indeed form a basis for  $\mathcal{M}^{(p-1)}$ we must express $V^{(p-1)}$ as a linear combination of basis elements for ideals $\mathcal{M}^{(p)}, \mathcal{M}^{(p-1)}$, due to inclusion relation~\eqref{eqn:inclusions_of_m}. Indeed, since $\id = \sum_\alpha P_\alpha$, we observe that 
            \begin{align}
                V^{(p-1)} &=\sum_{\alpha,\beta\vdash p-1}\Big(P_\alpha\otimes \id_{1,p+1,\ldots,2p}\Big) V^{(p-1)} \Big(\id_{1,2,\ldots,p,2p} \otimes P_{\beta}\Big) \\ 
                &=\sum_{\alpha,\beta\vdash p-1}\Big(P_\alpha\otimes P_{\alpha} \otimes \id_{1,2p}\Big) V^{(p-1)} \Big(P_{\beta} \otimes P_{\beta} 
                \otimes \id_{1,2p}\Big) \\
                &= \sum_{\alpha,\beta\vdash p-1} \sum_{i_\alpha,j_{\beta}} \Big(\id_{1} \otimes E^{\alpha}_{i_\alpha i_{\alpha}} \otimes E^{\alpha}_{i_\alpha i_{\alpha}}\otimes \id_{2p} \Big)  V^{(p-1)} \Big(\id_{1} \otimes  E^{\beta}_{j_{\beta} j_{\beta}} \otimes E^{\beta}_{j_{\beta} j_{\beta}}\otimes \id_{2p} \Big),
            \end{align}
            where we used the fact that Young projector $P_\alpha$ is a projector \eqref{eq:Young}, and the 'ping-pong' trick from~\eqref{eq:ping-pong}. By Lemma \ref{lemma:biger_e_operator} we can write then
            \begin{align}
                V^{(p-1)} &= \sum_{\alpha,\beta\vdash p-1} \sum_{\substack{\mu,\mu'=\alpha+\Box\\\nu,\nu'=\beta+\Box}} \sum_{i_\alpha,j_{\beta}} \Big(E^{\mu}_{i_\alpha i_{\alpha}} \otimes E^{\mu'}_{i_\alpha i_{\alpha}} \Big)  V^{(p-1)} \Big( E^{\nu}_{j_{\beta} j_{\beta}}  \otimes E^{\nu'}_{j_{\beta} j_{\beta}} \Big).
            \end{align}
            In the above, we used Lemma~\ref{lemma:biger_e_operator}. 
            Now let us add to the above expression the following zero term 
            $\frac{1}{d}\Big(E^{\mu}_{i_\mu i_\mu}\otimes \id \Big) V^{(p)}\Big(E^{\nu}_{j_\nu j_\nu}\otimes \id \Big)-\frac{1}{d}\Big(E^{\mu}_{i_\mu i_\mu}\otimes \id \Big) V^{(p)}\Big(E^{\nu}_{j_\nu j_\nu}\otimes \id \Big)$. By arranging summations in appropriate order, we get the following expression:
            \begin{align}
                V^{(p-1)} &= \sum_{\alpha,\beta\vdash p-1} \sum_{\substack{\mu,\mu'=\alpha+\Box\\\nu,\nu'=\beta+\Box}} \sum_{i_\alpha,j_{\beta}} \Bigg(\Big(E^{\mu}_{i_\alpha i_{\alpha}} \otimes E^{\mu'}_{i_\alpha i_{\alpha}} \Big)  V^{(p-1)} \Big( E^{\nu}_{j_{\beta} j_{\beta}} \otimes E^{\nu'}_{j_{\beta} j_{\beta}} \Big)\\
                &-\frac{1}{d} \Big( E^{\mu}_{i_\mu i_\mu}\otimes \id \Big)V^{(p)} \Big(E^{\nu}_{j_\nu j_\nu}\otimes \id \Big)\delta^{\mu \mu'}\delta^{\nu \nu'} + \frac{1}{d}\Big(E^{\mu}_{i_\mu i_\mu}\otimes \id \Big) V^{(p)}\Big(E^{\nu}_{j_\nu j_\nu}\otimes \id \Big)\delta^{\mu \mu'}\delta^{\nu \nu'}\Bigg)\\
                &=\frac{1}{d}\sum_{\alpha,\beta\vdash p-1} \sum_{\substack{\mu,\mu'=\alpha+\Box\\\nu,\nu'=\beta+\Box}} \sum_{i_\alpha,j_{\beta}}\Bigg(d\Big(E^{\mu}_{i_\alpha i_{\alpha}} \otimes E^{\mu'}_{i_\alpha i_{\alpha}} \Big)  V^{(p-1)} \Big( E^{\nu}_{j_{\beta} j_{\beta}} \otimes E^{\nu'}_{j_{\beta} j_{\beta}} \Big)\\
                &- \Big( E^{\mu}_{i_\mu i_\mu}\otimes \id \Big)V^{(p)} \Big(E^{\nu}_{j_\nu j_\nu}\otimes \id \Big)\delta^{\mu \mu'}\delta^{\nu \nu'} +  \Big(E^{\mu}_{i_\mu i_\mu}\otimes \id \Big) V^{(p)}\Big(E^{\nu}_{j_\nu j_\nu}\otimes \id \Big)\delta^{\mu \mu'}\delta^{\nu \nu'}\Bigg)\\
                &=\frac{1}{d}\sum_{\alpha,\beta\vdash p-1} \sum_{\substack{\mu,\mu'=\alpha+\Box\\\nu,\nu'=\beta+\Box}} \sum_{i_\alpha,j_{\beta}} \Bigg(\Bigg[d\Big(E^{\mu}_{i_\alpha i_{\alpha}} \otimes E^{\mu'}_{i_\alpha i_{\alpha}} \Big)  V^{(p-1)} \Big( E^{\nu}_{j_{\beta} j_{\beta}} \otimes E^{\nu'}_{j_{\beta} j_{\beta}} \Big)\\
                &- \Big( E^{\mu}_{i_\mu i_\mu}\otimes \id \Big)V^{(p)} \Big(E^{\nu}_{j_\nu j_\nu}\otimes \id \Big)\delta^{\mu \mu'}\delta^{\nu \nu'}\Bigg] + \Big(E^{\mu}_{i_\mu i_\mu}\otimes \id \Big) V^{(p)}\Big(E^{\nu}_{j_\nu j_\nu}\otimes \id \Big)\delta^{\mu \mu'}\delta^{\nu \nu'}\Bigg).
            \end{align}
            Inserting in the last term a normalization factor, we obtain:
            \begin{align}
            \label{eq:someformula}
                V^{(p-1)}&=\frac{1}{d}\sum_{\alpha,\beta\vdash p-1} \sum_{\substack{\mu,\mu'=\alpha+\Box\\\nu,\nu'=\beta+\Box}} \sum_{i_\alpha,j_{\beta}}\Bigg(
                \Bigg[d\Big(E^{\mu}_{i_\alpha i_{\alpha}} \otimes E^{\mu'}_{i_\alpha i_{\alpha}} \Big)  V^{(p-1)} \Big( E^{\nu}_{j_{\beta} j_{\beta}} \otimes E^{\nu'}_{j_{\beta} j_{\beta}} \Big)\\
                &- \Big( E^{\mu}_{i_\mu i_\mu}\otimes \id \Big)V^{(p)} \Big(E^{\nu}_{j_\nu j_\nu}\otimes \id \Big)\delta^{\mu \mu'}\delta^{\nu \nu'}\Bigg] + \frac{\sqrt{m_\mu m_\nu}}{\sqrt{m_\mu m_\nu}
                }\Big(E^{\mu}_{i_\mu i_\mu}\otimes \id \Big) V^{(p)}\Big(E^{\nu}_{j_\nu j_\nu}\otimes \id \Big)\delta^{\mu \mu'}\delta^{\nu \nu'}\Bigg).
            \end{align}
           The second element in the square bracket, due to the composition rules~\eqref{eqn:composition_of_e_operators}, and the 'ping-pong' trick~\eqref{eq:BellV'}, can be written as:
           \begin{align}
           \Big( E^{\mu}_{i_\mu i_\mu}\otimes \id \Big)V^{(p)} \Big(E^{\nu}_{j_\nu j_\nu}\otimes \id \Big)\delta^{\mu \mu'}\delta^{\nu \nu'}=\Big( E^{\mu}_{i_\mu 1_\mu}\otimes E^{\mu'}_{1_{\mu'} i_{\mu'}} \Big)V^{(p)} \Big(E^{\nu}_{j_\nu 1_\nu}\otimes E^{\nu'}_{1_{\nu'} j_{\nu'}} \Big)\delta^{\mu \mu'}\delta^{\nu \nu'}.
           \end{align}
           Plugging the above to~\eqref{eq:someformula} and taking into account~\eqref{eqn:def_of_op_H}, we arrive to:
            \begin{align}
                V^{(p-1)}&=\sum_{\alpha,\beta\vdash p-1} \sum_{\substack{\mu,\mu'=\alpha+\Box\\\nu,\nu'=\beta+\Box}} \sum_{i_\alpha,j_{\beta}}\Bigg(\Bigg[d\Big(E^{\mu}_{i_\alpha i_{\alpha}} \otimes E^{\mu'}_{i_\alpha i_{\alpha}} \Big)  V^{(p-1)} \Big( E^{\nu}_{j_{\beta} j_{\beta}} \otimes E^{\nu'}_{j_{\beta} j_{\beta}} \Big)\\
                &-\Big( E^{\mu}_{i_\mu i_\mu}\otimes \id \Big)V^{(p)} \Big(E^{\nu}_{j_\nu j_\nu}\otimes \id \Big)\delta^{\mu \mu'}\delta^{\nu \nu'}\Bigg]+ \frac{\sqrt{m_\mu m_\nu}}{\sqrt{m_\mu m_\nu}
                }\Big(E^{\mu}_{i_\mu i_\mu}\otimes \id \Big) V^{(p)}\Big(E^{\nu}_{j_\nu j_\nu}\otimes \id \Big)\delta^{\mu \mu'}\delta^{\nu \nu'}\Bigg)\\
                &=\sum_{\alpha,\beta\vdash p-1} \sum_{\substack{\mu,\mu'=\alpha+\Box\\\nu,\nu'=\beta+\Box}} \sum_{i_\alpha,j_{\beta}} \Bigg(\Bigg[d\Big(E^{\mu}_{i_\alpha i_{\alpha}} \otimes E^{\mu'}_{i_\alpha i_{\alpha}} \Big)  V^{(p-1)} \Big( E^{\nu}_{j_{\beta} j_{\beta}} \otimes E^{\nu'}_{j_{\beta} j_{\beta}}  \Big)\\
                &-\Big( E^{\mu}_{i_\mu 1_\mu}\otimes E^{\mu'}_{ i_{\mu'} 1_{\mu'}} \Big)V^{(p)} \Big(E^{\nu}_{ 1_\nu j_\nu }\otimes E^{\nu'}_{1_{\nu'} j_{\nu'}} \Big)\delta^{\mu\mu'}\delta^{\nu\nu'}\Bigg] +\frac{\sqrt{m_\mu m_\nu}}{\sqrt{m_\mu m_\nu}
                }\Big(E^{\mu}_{i_\mu i_\mu}\otimes \id \Big) V^{(p)}\Big(E^{\nu}_{j_\nu j_\nu}\otimes \id \Big)\delta^{\mu \mu'}\delta^{\nu \nu'}\Bigg)\\
                &=\sum_{\alpha,\beta\vdash p-1} \sum_{\substack{\mu,\mu'=\alpha+\Box\\\nu,\nu'=\beta+\Box}}  \Bigg(
                H^{\mu \mu' \ \ \nu \nu'}_{i_\mu i_{\mu'}  \ j_\nu j_{\nu'}}(p-1,\alpha,\beta)+ \sqrt{m_\mu m_\nu}\mathcal{G}^{\mu \nu \ \ \mu \nu}_{i_\mu j_\nu \ i_\mu j_\nu}(p)\delta^{\mu \mu'}\delta^{\nu \nu'}\Bigg)\label{eq:spanning}.
            \end{align}
                 This finishes the proof.
\end{proof}
     \end{lemma}

     \begin{remark}
    From the proof of Lemma~\ref{lemma:operator_H_generates_ideal_mp-1}, we see that generating operator $V^{(p-1)}$ for the ideal $\mathcal{M}^{(p-1)}$ is spanned by the  operators $H^{\mu \nu \ \ \mu' \nu'}_{i_\mu j_\nu \ i'_{\mu'} j'_{\nu'}}(p-1,\alpha,\alpha')$, as well as the irreducible matrix units $\mathcal{G}^{\mu \nu \ \ \mu \nu}_{i_\mu j_\nu \ i_\mu j_\nu}(p)$ for the ideal $\mathcal{M}^{(p)}$ - see expression~\eqref{eq:spanning}. This is the property we expected, due to the chain of inclusions given in~\eqref{eqn:inclusions_of_m}.
     \end{remark}

We see that due to relation~\eqref{eq:comp_H} given in Theorem~\ref{thm:lower_ideal_operator_G} the operators $H$ are not orthonormal in the interior indices $\alpha',\beta$ - their composition is multiplied by the factor $db^{\mu'\nu'}(\alpha',\beta)$, where $\alpha'=\mu'-\Box, \beta=\nu'-\Box$. 
Consider the case when the coefficients $db^{\mu'\nu'}(\alpha',\beta)\neq 0$, only then relation~\eqref{eq:comp_H} is nontrivial. 
In the next part of this section, we are interested in finding a new set of operators that satisfy the orthonormality relation also with respect to the interior indices. We start with preliminary considerations. For fixed $\mu,\nu$ the numbers $b^{\mu\nu}(\alpha,\alpha')$ can be arranged as a matrix $B^{\mu\nu}=(b^{\mu\nu}(\alpha,\alpha'))$, where according to~\eqref{eq:wsp_b_general} from Remark~\ref{Remark12} we have:
\begin{align}
\label{eqn:wsp_b_general2}
    B^{\mu\nu}:=(b^{\mu \nu}(\alpha,\alpha')) = \Bigg(\frac{1}{d(d^2-1)}\Bigg(d\frac{m_{\mu}m_{\nu}}{m_{\alpha}}\delta^{\alpha \alpha'} - m_{\mu}\delta^{\mu \nu} \Bigg)\Bigg).
\end{align}
When $\mu=\nu$ the matrix $B^{\mu\nu}$ is a square matrix with the matrix elements of the form:
\begin{align}
\label{Sec:matB}
 B^{\mu\mu}=(b^{\mu\mu}(\alpha,\alpha'))=\Bigg(\frac{m_{\mu}}{d(d^2-1)}\Bigl(d\frac{m_{\mu}}{m_{\alpha}}\delta^{\alpha\alpha'} - 1\Bigr) \Bigg).
\end{align}
This observation follows directly from definition of the numbers $b^{\mu\nu}(\alpha,\alpha')$ in~\eqref{eqn:wsp_b_general2}. For $\mu\neq \nu$, the matrix $B^{\mu\nu}$ is rectangular. In this case, only the first term in~\eqref{eqn:wsp_b_general2} survives, and it is non-zero if and only if $\alpha=\alpha'$. Then the matrix $B^{\mu\nu}$ has the following matrix elements:
\begin{align}
\label{Sec:matBrect}
b^{\mu\nu}(\alpha,\alpha')=\frac{1}{(d^2-1)}\frac{m_{\mu}m_{\nu}}{m_{\alpha}}\delta^{\alpha\alpha'}.
\end{align}
As a result the matrix $B^{\mu\nu}$ has non-zero elements only on its main diagonal.
Actually only one element of this matrix is non-zero. Indeed, for $\alpha=\alpha'$  it must be $\alpha=\mu-\Box$ and $\alpha=\nu-\Box$ with the same requirement for $\alpha'$. We denote this property shortly as $\alpha,\alpha' \in \mu\wedge \nu$. Moreover, since we subtract only one box from $\mu$ and $\nu$ to ensure $\alpha=\alpha'$ and $\alpha,\alpha' \in \mu\wedge \nu$, the diagram  $\mu$ must be obtained from  $\nu$ by moving a one box. This also implies that for every pair $\mu, \nu$ satisfying this property there is only one common $\alpha$ for them.  Since other elements of $B^{\mu\nu}$ are equal to zero, for our further purposes, we can restrict the matrix $B^{\mu\nu}$ to be a square one-by-one matrix with the mentioned condition  $\alpha,\alpha' \in \mu\wedge \nu$ with the elements on the main diagonal given through~\eqref{Sec:matBrect}. When diagram $\nu$ cannot be obtained from $\mu$ by moving only one box the whole matrix $B^{\mu\nu}$ is a zero matrix since there are no common shapes $\alpha,\alpha' \in \mu\wedge \nu$.

Now, we discuss two main properties of the matrix $B^{\mu\nu}$ - its unitary diagonalization and invertibility. We start from the diagonalization property.
For fixed $\mu,\nu$ we define a family of unitary matrices as $U^{\mu\nu} = \{(U^{\mu\nu}_{\alpha\beta})\}$, such that a unitary matrix $(U^{\mu \nu}_{\alpha \beta})$ transforms the matrix $B^{\mu\nu}$ to its diagonal form $B_{diag}^{\mu\nu}=\left( b_{diag}^{\mu\nu} (\alpha, \alpha')\right)$ as
\begin{align}\label{eqn:diag_b_coeff0}
        b_{diag}^{\mu\nu} (\alpha, \alpha') = \sum_{\substack{\beta,\beta'\in\mu\wedge \nu}} U^{\mu \nu}_{\alpha \beta} \cdot b^{\mu\nu}(\beta,\beta') \cdot \Bigl(U^{\mu \nu} \Bigr)^\dagger_{\beta' \alpha'}
\end{align}
where according to the definition of diagonalization, non-zero elements are only obtained for $\alpha=\alpha'$. Thus for diagonal elements, we will write for simplicity just $ b_{diag}^{\mu\nu} (\alpha)$. Notice that matrix $(U^{\mu \nu}_{\alpha \beta})$ is non-trivial, i.e. different than the identity matrix, only for $\mu=\nu$. 
At the end let us stress that a non-trivial unitary that diagonalizes the matrix $B^{\mu\mu}$ always exists, since matrix $B$ from its definition is a normal matrix, i.e. it satisfies the condition $(B^{\mu\mu})^{\dagger}B^{\mu\mu}=B^{\mu\mu}(B^{\mu\mu})^{\dagger}$. In fact, due to its construction, the matrix $B^{\mu\mu}$, as well as $B^{\mu\nu}$, is a symmetric matrix. 

Now we discuss shortly the singularity (invertibility) of the matrix $B^{\mu\mu}$. 
The non-invertibility of the matrix $B^{\mu\mu}(\alpha,\alpha')$, or equivalently determinant vanishing, is described by the Corollary~\ref{Cor:29} from Appendix~\ref{App:PropB}, saying the following:
\begin{align}
\label{B0cond}
\det (B^{\mu\mu})=0 \quad \text{if and only if} \quad dm_{\mu }=m_{\alpha _{1}}+m_{\alpha _{2}}+\ldots+m_{\alpha _{k}},
\end{align} 
where $\alpha_i$, for $1\leq i\leq k$, are all Young diagrams obtained from $\mu$ by removing a single box. For the detailed proof of this property and an additional discussion, especially the connection between $\det (B^{\mu\mu})$ and theory of the symmetric polynomials, we refer to Appendix~\ref{App:PropB}. For the case $\mu\neq \nu$, the matrix $B^{\mu\nu}$ its invertible on its support. Since in this case, it is diagonal commuting its inversion is straightforward. 

Having the notion for the family of diagonalizing unitary matrices $U^{\mu \nu}=\{(U^{\mu \nu}_{\alpha \beta})\}$ and discussion about the invertibility of the matrix $B^{\mu\nu}$, we are ready to formulate the main result for this section.
\begin{theorem}
\label{thm:irrepsMp1}
Let $U^{\mu \nu}=\{(U^{\mu \nu}_{\alpha \beta})\}$ be a family of diagonalizing unitary matrices defined through relation~\eqref{eqn:diag_b_coeff0} and $H^{\mu \nu \ \ \mu' \nu'}_{i_\mu j_\nu \ i'_{\mu'} j'_{\nu'}}(p-1,\alpha,\alpha')$ be operators given by Theorem~\ref{thm:lower_ideal_operator_G}.
The irreducible matrix units in ideal $\mathcal{M}^{(p-1)}$ are given by the following operators:
\begin{align}\label{eqn:def_of_gpis}
\G^{\mu\nu \ \ \mu'\nu'}_{i_\mu j_\nu \ i'_{\mu'}j'_{\nu'}}(p-1,\beta,\beta'):=\frac{\sum_{\substack{\alpha\in \mu \wedge \nu \\ \alpha' \in \mu'\wedge \nu'}}\Bigl(U^{\mu \nu}\Bigr)^\dagger_{\alpha \beta} U^{\mu' \nu'}_{\beta' \alpha'} H^{\mu \nu \ \ \mu' \nu'}_{i_\mu j_\nu \ i'_{\mu'} j'_{\nu'}}(p-1,\alpha,\alpha')}{d\sqrt{b^{\mu\nu}_{diag}(\beta)b^{\mu'\nu'}_{diag}(\beta')}},
\end{align}
where the numbers $b^{\mu\nu}_{diag}(\beta), b^{\mu'\nu'}_{diag}(\beta')$ are non-zero eigenvalues of the matrix described by~\eqref{eqn:wsp_b_general2}. 
The above operators satisfy the following composition rule:
\begin{align}
\label{eq:compositionG}
\G^{\mu \nu \ \ \mu' \nu'}_{i_\mu j_\nu \ i'_{\mu'} j'_{\nu'}}(p-1,\beta,\beta')\cdot  \G^{\Tilde{\mu} \Tilde{\nu} \ \ \mu" \nu"}_{k_{\Tilde{\mu}} l_{\Tilde{\nu}} \ k'_{\mu"} l'_{\nu"}}(p-1,\Tilde{\beta},\Tilde{\beta'}) =  \G^{\mu,\nu \ \ \mu"\nu"}_{i_\mu j_\nu \ k'_{\mu"}l'_{\nu"}}(p-1,\beta,\Tilde{\beta'})\delta^{\beta' \Tilde{\beta}}\delta^{\mu' \Tilde{\mu}}\delta^{\nu' \Tilde{\nu}}\delta_{i'_{\mu'} k_{\Tilde{\mu}}}\delta_{j'_{\mu'} l_{\Tilde{\nu}}}
\end{align}
\end{theorem}

\begin{proof}
For fixed $\mu,\nu$, let $U^{\mu \nu}=(U^{\mu \nu}_{\alpha \beta})$ be a diagnalizing unitary operator as it is in~\eqref{eqn:diag_b_coeff0}. We define the following set of auxiliary operators:
    \begin{align}\label{eqn:new_op_g_transf}
        \forall_{\mu,\nu,\mu',\nu'\vdash p} \quad \mathcal{H}^{\mu \nu \ \ \mu' \nu'}_{i_\mu j_\nu \ i'_{\mu'} j'_{\nu'}}(p-1,\beta,\beta') := \sum_{\substack{\alpha\in \mu \wedge \nu \\ \alpha' \in \mu'\wedge \nu'}}\Bigl(U^{\mu \nu}\Bigr)^\dagger_{\alpha \beta} U^{\mu' \nu'}_{\beta \alpha'} H^{\mu \nu \ \ \mu' \nu'}_{i_\mu j_\nu \ i'_{\mu'} j'_{\nu'}}(p-1,\alpha,\alpha').
    \end{align}

    We show that the operators from expression \eqref{eqn:new_op_g_transf} satisfy composition relations similarly to those in Theorem \ref{thm:lower_ideal_operator_G}. Indeed, we have
\begin{align}
     &\mathcal{H}^{\mu \nu \ \ \mu' \nu'}_{i_\mu j_\nu \ i'_{\mu'} j'_{\nu'}}(p-1,\beta,\beta') \cdot \mathcal{H}^{\Tilde{\mu} \Tilde{\nu} \ \ \mu" \nu"}_{k_{\Tilde{\mu}} l_{\Tilde{\nu}} \ k'_{\mu"} l'_{\nu"}}(p-1,\Tilde{\beta},\Tilde{\beta'})=\label{eq:HcH}\\
     &= \sum_{\substack{\alpha\in \mu \wedge \nu \\ \alpha' \in \mu'\wedge \nu'}} \Bigl(U^{\mu \nu}\Bigr)^\dagger_{\alpha \beta} U^{\mu' \nu'}_{\beta' \alpha'}H^{\mu\nu \ \ \mu'\nu'}_{i_\mu j_\nu \ i'_{\mu'}j'_{\nu'}}(p-1,\alpha,\alpha')\sum_{\substack{ \Tilde{\alpha}\in\Tilde{\mu}\wedge\Tilde{\nu} \\ \Tilde{\alpha'}\in\mu"\wedge \nu"}} \Bigl(U^{\Tilde{\mu} \Tilde{\nu}} \Bigr)^\dagger_{\Tilde{\alpha} \Tilde{\beta}} U^{\mu" \nu"}_{\Tilde{\beta'} \Tilde{\alpha'}} H^{\Tilde{\mu} \Tilde{\nu} \ \ \mu" \nu"}_{k_{\Tilde{\mu}} l_{\Tilde{\nu}} \ k'_{\mu"} l'_{\nu"}}(p-1,\Tilde{\alpha},\Tilde{\alpha'})\\
    &= \sum_{\substack{\alpha\in \mu\wedge\nu \\ \alpha' \in \mu'\wedge \nu'}} \Bigl(U^{\mu,\nu}\Bigr)^\dagger_{\alpha,\beta} U^{\mu' \nu'}_{\beta' \alpha'} \sum_{\substack{ \Tilde{\alpha}\in\Tilde{\mu}\wedge \Tilde{\nu} \\ \Tilde{\alpha'}\in\mu"\wedge \nu"}} \Bigl(U^{\Tilde{\mu} \Tilde{\nu}}\Bigr)^\dagger_{\Tilde{\alpha} \Tilde{\beta}} U^{\mu" \nu"}_{\Tilde{\beta'} \Tilde{\alpha'}}H^{\mu\nu \ \ \mu'\nu'}_{i_\mu j_\nu \ i'_{\mu'}j'_{\nu'}}(p-1,\alpha,\alpha')H^{\Tilde{\mu} \Tilde{\nu} \ \ \mu" \nu"}_{k_{\Tilde{\mu}} l_{\Tilde{\nu}} \ k'_{\mu"} l'_{\nu"}}(p-1,\Tilde{\alpha},\Tilde{\alpha'})\\
    &= d \sum_{\substack{\alpha\in \mu\wedge \nu \\ \alpha',\Tilde{\alpha} \in \mu'\wedge \nu' \\\Tilde{\alpha'}\in\mu"\wedge \nu"}} \Bigl(U^{\mu \nu}\Bigr)^\dagger_{\alpha \beta} U^{\mu' \nu'}_{\beta' \alpha'}  \Bigl(U^{\mu' \nu'} \Bigr)^\dagger_{\Tilde{\alpha} \Tilde{\beta}} U^{\mu" \nu"}_{\Tilde{\beta'} \Tilde{\alpha'}} b^{\mu'\nu'}(\alpha',\Tilde{\alpha}) H^{\mu,\nu \ \ \mu"\nu"}_{i_\mu j_\nu \ k'_{\mu"}l'_{\nu"}}(p-1,\alpha,\Tilde{\alpha'})\delta^{\mu' \Tilde{\mu}}\delta^{\nu' \Tilde{\nu}}\delta_{i'_{\mu'} k_{\Tilde{\mu}}}\delta_{j'_{\mu'} l_{\Tilde{\nu}}}
\end{align}
where in the last line we used the fact that operators from~\eqref{eq:HcH} satisfy the orthogonality relations from Theorem \ref{thm:lower_ideal_operator_G}. Reorganizing the expression,  we can write
\begin{align}
    &\mathcal{H}^{\mu \nu \ \ \mu' \nu'}_{i_\mu j_\nu \ i'_{\mu'} j'_{\nu'}}(p-1,\beta,\beta') \cdot \mathcal{H}^{\Tilde{\mu} \Tilde{\nu} \ \ \mu" \nu"}_{k_{\Tilde{\mu}} l_{\Tilde{\nu}} \ k'_{\mu"} l'_{\nu"}}(p-1,\Tilde{\beta},\Tilde{\beta'})=\\
    &= d \sum_{\substack{\alpha\in \mu\wedge \nu \\ \Tilde{\alpha'}\in\mu"\wedge \nu"}}  \left[\sum_{\substack{\alpha',\Tilde{\alpha} \in \mu'\wedge \nu'}} U^{\mu' \nu'}_{\beta \alpha'}b^{\mu'\nu'}(\alpha',\Tilde{\alpha})  \Bigl(U^{\mu' \nu'} \Bigr)^\dagger_{\Tilde{\alpha} \Tilde{\beta}} \right]\Bigl(U^{\mu \nu}\Bigr)^\dagger_{\alpha \beta} U^{\mu" \nu"}_{\Tilde{\beta'} \Tilde{\alpha'}}H^{\mu,\nu \ \ \mu"\nu"}_{i_\mu j_\nu \ k'_{\mu"}l'_{\nu"}}(p-1,\alpha,\Tilde{\alpha'})\delta^{\mu' \Tilde{\mu}}\delta^{\nu' \Tilde{\nu}}\delta_{i'_{\mu'} k_{\Tilde{\mu}}}\delta_{j'_{\mu'} l_{\Tilde{\nu}}}\\
        & = db_{diag}^{\mu'\nu'}(\beta') \mathcal{H}^{\mu,\nu \ \ \mu"\nu"}_{i_\mu j_\nu \ k'_{\mu"}l'_{\nu"}}(p-1,\beta,\Tilde{\beta'})\delta^{\mu' \Tilde{\mu}}\delta^{\nu' \Tilde{\nu}}\delta^{\beta' \Tilde{\beta}}\delta_{i'_{\mu'} k_{\Tilde{\mu}}}\delta_{j'_{\mu'} l_{\Tilde{\nu}}}
\end{align}

Now, defining $\G^{\mu\nu \ \ \mu'\nu'}_{i_\mu j_\nu \ i'_{\mu'}j'_{\nu'}}(p-1,\beta,\beta')$ as it is in~\eqref{eqn:def_of_gpis}, we show that relations~\eqref{eq:compositionG} are indeed satisfied:
\begin{align}
&\G^{\mu \nu \ \ \mu' \nu'}_{i_\mu j_\nu \ i'_{\mu'} j'_{\nu'}}(p-1,\beta,\beta') \cdot \G^{\Tilde{\mu} \Tilde{\nu} \ \ \mu" \nu"}_{k_{\Tilde{\mu}} l_{\Tilde{\nu}} \ k'_{\mu"} l'_{\nu"}}(p-1,\Tilde{\beta},\Tilde{\beta'})=\\
&=\frac{\h^{\mu \nu \ \ \mu' \nu'}_{i_\mu j_\nu \ i'_{\mu'} j'_{\nu'}}(p-1,\beta,\beta')}{d\sqrt{b^{\mu\nu}_{diag}(\beta)b^{\mu'\nu'}_{diag}(\beta')}}\frac{\h^{\Tilde{\mu} \Tilde{\nu} \ \ \mu" \nu"}_{k_{\Tilde{\mu}} l_{\Tilde{\nu}} \ k'_{\mu"} l'_{\nu"}}(p-1,\Tilde{\beta},\Tilde{\beta'})}{d\sqrt{b^{\Tilde{\mu}\Tilde{\nu}}_{diag}(\Tilde{\beta})b^{\mu"\nu"}_{diag}(\Tilde{\beta'})}}\\
&=\frac{db^{\mu'\nu'}_{diag}(\beta')}{d^2\sqrt{b^{\mu\nu}_{diag}(\beta)b^{\mu'\nu'}_{diag}(\beta')}\sqrt{b^{\mu'\nu'}_{diag}(\beta',\beta')b^{\mu"\nu"}_{diag}(\Tilde{\beta'},\Tilde{\beta'})}}\h^{\mu,\nu \ \ \mu"\nu"}_{i_\mu j_\nu \ k'_{\mu"}l'_{\nu"}}(p-1,\beta,\Tilde{\beta'})\delta^{\mu' \Tilde{\mu}}\delta^{\nu' \Tilde{\nu}}\delta^{\beta' \Tilde{\beta}}\delta_{i'_{\mu'} k_{\Tilde{\mu}}}\delta_{j'_{\mu'} l_{\Tilde{\nu}}}\\
&=\frac{\h^{\mu,\nu \ \ \mu"\nu"}_{i_\mu j_\nu \ k'_{\mu"}l'_{\nu"}}(p-1,\beta,\Tilde{\beta'})}{d\sqrt{b^{\mu\nu}_{diag}(\beta)b^{\mu"\nu"}_{diag}(\Tilde{\beta'})}}\delta^{\mu' \Tilde{\mu}}\delta^{\nu' \Tilde{\nu}}\delta^{\beta' \Tilde{\beta}}\delta_{i'_{\mu'} k_{\Tilde{\mu}}}\delta_{j'_{\mu'} l_{\Tilde{\nu}}}\\
&=\G^{\mu,\nu \ \ \mu"\nu"}_{i_\mu j_\nu \ k'_{\mu"}l'_{\nu"}}(p-1,\beta,\Tilde{\beta'})\delta^{\mu' \Tilde{\mu}}\delta^{\nu' \Tilde{\nu}}\delta^{\beta' \Tilde{\beta}}\delta_{i'_{\mu'} k_{\Tilde{\mu}}}\delta_{j'_{\mu'} l_{\Tilde{\nu}}}.
\end{align}
The analog of the spanning property~\eqref{eq:spanning} follows from the following relation between the operators $H^{\mu \nu \ \ \mu' \nu'}_{i_\mu j_\nu \ i'_{\mu'} j'_{\nu'}}(p-1,\alpha,\alpha')$ and $\G^{\mu \nu \ \ \mu' \nu'}_{i_\mu j_\nu \ i'_{\mu'} j'_{\nu'}}(p-1,\beta,\beta')$:
\begin{align}
H^{\mu \nu \ \ \mu' \nu'}_{i_\mu j_\nu \ i'_{\mu'} j'_{\nu'}}(p-1,\alpha,\alpha')=d\sum_{\substack{\beta\in \mu \wedge \nu \\ \beta' \in \mu'\wedge \nu'}}\sqrt{b^{\mu\nu}_{diag}(\beta)b^{\mu'\nu'}_{diag}(\beta')}U^{\mu \nu}_{\beta \alpha}\left(U^{\mu' \nu'}\right)^{\dagger}_{\alpha' \beta'}\G^{\mu \nu \ \ \mu' \nu'}_{i_\mu j_\nu \ i'_{\mu'} j'_{\nu'}}(p-1,\beta,\beta').
\end{align}
Indeed, the above equality holds since we can write its right-hand side as:
\begin{align}
&d\sum_{\substack{\beta\in \mu \wedge \nu \\ \beta' \in \mu'\wedge \nu'}}\sqrt{b^{\mu\nu}_{diag}(\beta)b^{\mu'\nu'}_{diag}(\beta')}U^{\mu \nu}_{\beta \alpha}\left(U^{\mu' \nu'}\right)^{\dagger}_{\alpha' \beta'}\G^{\mu \nu \ \ \mu' \nu'}_{i_\mu j_\nu \ i'_{\mu'} j'_{\nu'}}(p-1,\beta,\beta')\\
&=d\sum_{\substack{\beta\in \mu \wedge \nu \\ \beta' \in \mu'\wedge \nu'}}\sqrt{b^{\mu\nu}_{diag}(\beta)b^{\mu'\nu'}_{diag}(\beta')}U^{\mu \nu}_{\beta \alpha}\left(U^{\mu' \nu'}\right)^{\dagger}_{\alpha' \beta'}\frac{\sum_{\substack{\Tilde{\alpha}\in \mu \wedge \nu \\ \Tilde{\alpha}' \in \mu'\wedge \nu'}}\Bigl(U^{\mu \nu}\Bigr)^\dagger_{\Tilde{\alpha} \beta} U^{\mu' \nu'}_{\beta' \Tilde{\alpha}'} H^{\mu \nu \ \ \mu' \nu'}_{i_\mu j_\nu \ i'_{\mu'} j'_{\nu'}}(p-1,\Tilde{\alpha},\Tilde{\alpha}')}{d\sqrt{b^{\mu\nu}_{diag}(\beta)b^{\mu'\nu'}_{diag}(\beta')}}\\
&=\sum_{\substack{\Tilde{\alpha}\in \mu \wedge \nu \\ \Tilde{\alpha}' \in \mu'\wedge \nu'}}\left(\sum_{\substack{\beta\in \mu \wedge \nu }}\Bigl(U^{\mu \nu}\Bigr)^\dagger_{\Tilde{\alpha} \beta} U^{\mu \nu}_{\beta \alpha}\right)\left(\sum_{\substack{ \beta' \in \mu'\wedge \nu'}}\left(U^{\mu' \nu'}\right)^{\dagger}_{\alpha' \beta'}U^{\mu' \nu'}_{\beta' \Tilde{\alpha}'}\right)H^{\mu \nu \ \ \mu' \nu'}_{i_\mu j_\nu \ i'_{\mu'} j'_{\nu'}}(p-1,\Tilde{\alpha},\Tilde{\alpha}')\\
&=\sum_{\substack{\Tilde{\alpha}\in \mu \wedge \nu \\ \Tilde{\alpha}' \in \mu'\wedge \nu'}} \delta^{\Tilde{\alpha} \alpha}\delta^{\alpha' \Tilde{\alpha}'}H^{\mu \nu \ \ \mu' \nu'}_{i_\mu j_\nu \ i'_{\mu'} j'_{\nu'}}(p-1,\Tilde{\alpha},\Tilde{\alpha}')=H^{\mu \nu \ \ \mu' \nu'}_{i_\mu j_\nu \ i'_{\mu'} j'_{\nu'}}(p-1,\Tilde{\alpha},\Tilde{\alpha}')=H^{\mu \nu \ \ \mu' \nu'}_{i_\mu j_\nu \ i'_{\mu'} j'_{\nu'}}(p-1,\alpha,\alpha').
\end{align}
This finishes the proof.
\end{proof}

\begin{remark}
\label{Rem:dissc}
Let us remind here that due to discussion presented below equation~\eqref{Sec:matBrect} and~\eqref{eqn:diag_b_coeff0} non-trivial unitary (different than identity) in~\eqref{eqn:def_of_gpis} is only in the case when $\mu=\nu$ and/or $\mu'=\nu'$. However, we keep writing $U^{\mu\nu}_{\alpha\beta}, U^{\mu'\nu'}_{\alpha'\beta'}$ to have unified expression for the operators $\G^{\mu\nu \ \ \mu'\nu'}_{i_\mu j_\nu \ i'_{\mu'}j'_{\nu'}}(p-1,\beta,\beta')$ in Theorem~\ref{thm:irrepsMp1} for all possible cases.
\end{remark}

Notice that to construct the irreducible matrix units we have to exclude all eigenvalues of the matrix $B^{\mu \mu}$ that are equal to zero. If $B^{\mu \mu}$ is singular, the operators from Theorem~\ref{thm:irrepsMp1} are linearly dependent and the basis is overcomplete.  Thus the number of elements in the basis is smaller than the number of operators from Theorem~\ref{thm:irrepsMp1}. 
One could also proceed a different route - first exclude linear dependent operators appearing in Theorem~\ref{thm:irrepsMp1},  and then the resulting matrix $B^{\mu \mu}$ will be automatically non-singular. The exclusion can be done algorithmically and it is depicted in Appendix~\ref{app:C}.


\section{Matrix units for $\mathcal{A}^d_{p,q}$ from Clebsch--Gordan tensor networks}\label{sec:tn_rep_mat_units}

In this section, we present an alternative construction for matrix units of $\mathcal{A}^d_{p,q}$ adapted to $\mathbb{C}[\s_p] \times \mathbb{C}[\s_q]$ subalgebra for arbitrary $p,q,d$. This construction is based on tensor networks built from Clebsch--Gordan transforms.

The basic idea is to treat the underlying Hilbert space of $p+q$ qudits as tensor product of defining $\mathcal{U}_\square$ and dual of the defining $\mathcal{U}_{\overline \square}$ irreps of the unitary group $U(d)$ and apply two Schur transforms to the left and right parts of the tensor product \cite{grinko2023gelfandtsetlinbasispartiallytransposed,grinko2025phd}:
\begin{equation}
      (\mathbb{C}^d)^{\otimes p+q} = \mathcal{U}_\square \otimes \dotsc \otimes \mathcal{U}_\square \otimes \mathcal{U}_{\overline \square} \otimes \dotsc \otimes \mathcal{U}_{\overline \square} \\ \simeq^{U_{\mathrm{Sch}} \otimes U_{\mathrm{dSch}}}
  \Big( \bigoplus_{\substack{\mu \vdash p \\ \operatorname{ht}(\mu) \leq d}} \mathcal{U}_{\mu} \otimes \mathcal{S}_{\mu} \Big) \otimes \Big( \bigoplus_{\substack{\mu \vdash q \\ \operatorname{ht}(\mu) \leq d}} \mathcal{U}_{\overline \mu} \otimes \mathcal{S}_{\mu} \Big),
\end{equation}
where $U_{\mathrm{Sch}}$ is normal Schur transform comprised of standard CG transforms $\mathrm{CG}^+$ and $U_{\mathrm{dSch}}$ is the dual Schur transform comprised of dual CG transforms $\mathrm{CG}^-$, and notation $\simeq^{U_{\mathrm{Sch}} \otimes U_{\mathrm{dSch}}}$ means that the isomorphism is achieved via a change of basis unitary $U_{\mathrm{Sch}} \otimes U_{\mathrm{dSch}}$.

Recall that Clebsch--Gordan transform is defined as a unitary which block diagonalizes tensor product of two irreps of the unitary group:
\begin{equation}
    \mathcal{U}_\lambda \otimes \mathcal{U}_\mu \simeq^{\mathrm{CG}} \bigoplus_\nu \mathcal{U}_\nu \otimes \mathbb{C}^{c_{\lambda,\mu}^{\nu}},
\end{equation}
where $\lambda$, $\mu$, $\nu$ represent highest weights of the correponding simple Weyl modules, $c_{\lambda,\mu}^{\nu}$ is the Littlewood-Richardson coefficient, and
this isomorphism is achieved via Clebsch--Gordan unitary $\mathrm{CG}$, i.e. tensor product of two irreducible representations $U_\lambda$, $U_\mu$ of a unitary matrix $U$ is block-diagonalized as
\begin{equation}
    \mathrm{CG} (U_\lambda \otimes U_\mu) \mathrm{CG}^\dagger = \bigoplus_\nu U_\nu \otimes I_{c_{\lambda,\mu}^{\nu}}.
\end{equation}

\begin{figure}[!h]
    \centering
    \includegraphics[width=0.17\textwidth, page=1]{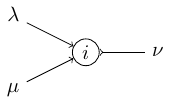}
    $\qquad \qquad$
    \includegraphics[width=0.17\textwidth, page=2]{tensor_networks.pdf}
    $\qquad \qquad$
    \includegraphics[width=0.17\textwidth, page=3]{tensor_networks.pdf}
    \caption{Tensor network representation of general Clebsch--Gordan tensor $\mathrm{CG}$ and two special tensors $\mathrm{CG}^+$ and $\mathrm{CG}^-$. Arrows indicate input and output irreps.}
    \label{fig:CG_tns}
\end{figure}

\begin{figure}[!h]
    \centering
    \includegraphics[width=0.5\textwidth, page=7]{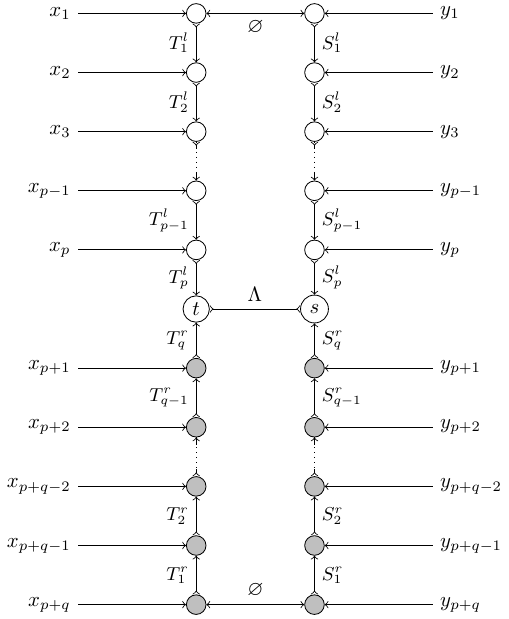}
    \caption{Matrix units for $\mathcal{A}^d_{p,q}$ adapted to $\mathbb{C}[\s_p] \times \mathbb{C}[\s_q]$ represented as tensor network of Clebsch--Gordan tensors. 
    Multiplicities $t$ and $s$ can take values from $1$ to $c^{\Lambda}_{T^l_p,T^r_q}$ and $c^{\Lambda}_{S^l_p,S^r_q}$ correspondingly; $x$ and $y$ represent standard basis vectors.}
    \label{fig:Matrix_units_new}
\end{figure}

We can think of $\mathrm{CG}$ as tensor, labelled by two input irreps $\lambda$, $\mu$ and two outputs $\nu$, $i$, where $i \in \{1,\dotsc,c_{\lambda,\mu}^{\nu}\}$ represents multiplicity.
For our purposes, we will distinguish two special cases $\mathrm{CG}^+$ and $\mathrm{CG}^-$:
$\mathrm{CG}^+$ corresponds to taking tensor product between an arbitrary irrep $\lambda$ and the defining irrep $\square := (1,0,0,\dotsc,0) = (1,0^{d-1})$, while $\mathrm{CG}^-$ correponds to taking tensor product between $\lambda$ and $\bar{\square} := (0,0,\dotsc,0,-1) = (0^{d-1},-1)$, which is the dual of the defining irrep, see  Figure~(\ref{fig:CG_tns}).

Using tensor network representation of $\mathrm{CG}$ transforms from Figure~(\ref{fig:CG_tns}), we can represent matrix units $E^\Lambda_{\mathfrak{t},\mathfrak{s}}$ for $\mathcal{A}^d_{p,q}$ adapted to $\mathbb{C}[\s_p] \times \mathbb{C}[\s_q]$ as in Figure~(\ref{fig:Matrix_units_new}).
In this language, the matrix units are labelled by data $\mathfrak{t} = (\Lambda,T^l,T^r,t)$ and $\mathfrak{s} = (\Lambda,S^l,S^r,s)$, where $T^l, S^l$ are SYTs with $p$ boxes, $T^r, S^r$ are SYTs with $q$ boxes, and  
$t$ and $s$ are multiplicity labels; $t \in \{1,\dotsc,c^{\Lambda}_{T^l_p,T^r_q}\}$ and $s \in \{1,\dotsc,c^{\Lambda}_{S^l_p,S^r_q}\}$.

In a similar spirit, the spanning operators $F_{i_\mu j_\nu \ i'_{\mu'} j'_{\nu'}}^{\mu\nu \ \ \mu'\nu'}(p-1,\alpha,\alpha')$ introduced in Equation~(\ref{eqn:F_poperator}) could be represented as tensor networks, see Figure~(\ref{fig:Matrix_units_original}). 
As explained in the previous sections, they are not orthogonal, so an additional orthogonalisation process is required for $F_{i_\mu j_\nu \ i'_{\mu'} j'_{\nu'}}^{\mu\nu \ \ \mu'\nu'}(p-1,\alpha,\alpha')$, which then results in the construction of another basis $H^{\mu \nu \ \ \mu' \nu'}_{i_\mu j_\nu \ i'_{\mu'} j'_{\nu'}}(p-1,\alpha,\alpha')$, and the final orthogonal basis $\G^{\mu\nu \ \ \mu'\nu'}_{i_\mu j_\nu \ i'_{\mu'}j'_{\nu'}}(p-1,\beta,\beta')$ from Theorem~\ref{thm:irrepsMp1}. In contrast, the $E^\Lambda_{\mathfrak{t},\mathfrak{s}}$ operators from Figure~(\ref{fig:Matrix_units_new}) are orthogonal matrix units by construction, so no additional orthogonalisation is needed.

The alternative basis in Figure~\ref{fig:Matrix_units_new} is useful because it can easily provide a method to compute matrix elements of all generators of the algebra $\mathcal{A}^d_{p,q}$. 
In particular, the generators of permutations inside $\mathcal{A}^d_{p,q}$ are represented according to Young--Yamanouchi basis since the basis by construction is adapted to $\mathbb{C}[\s_p] \times \mathbb{C}[\s_q]$. The formula for the Young--Yamanouchi basis can be found in \cite{grinko2023gelfandtsetlinbasispartiallytransposed}. Moreover, in all $\Lambda$ irreps $(\mathfrak{s},\mathfrak{t})$-matrix elements of the only ``new'' generator $V^{(1)}$ can be easily computed by contracting the tensor network, presented in Figure \ref{fig:Matrix_units_new}, and they are given by the following expression:
\begin{align}
\label{eq:contr_gen_method}
    \frac{1}{m_\Lambda} \Tr \left[E^\Lambda_{\mathfrak{t},\mathfrak{s}} V^{(1)} \right] = \frac{1}{m_\Lambda} \prod_{i=1}^{q-1} \delta_{T^r_i,S^r_i} \prod_{j=1}^{p-1} \delta_{T^l_j,S^l_j} \includegraphics[valign=c, scale=0.9, page=8]{tensor_networks.pdf}
\end{align}
For example, all generators of $\mathcal{A}^d_{p,q}$ look like for the case $p=q=d=3$ is provided in Appendix \ref{app:tensor_network_example_gens}. Moreover, we can use these representation matrices to compute the twirl of $V^{(k)}$ for $k=1,2,3$ over $\s_p \times \s_q$. We call the resulting matrices $\rho(k)$. The result is provided in Appendix \ref{app:tensor_network_example_rhos}. In the next section, we proceed to study the spectrum of operators $\rho(k)$, based on the method developed Sections \ref{sec:ideal_m}, \ref{sec:ideal_m-1}.

\begin{figure}[!t]
    \centering
    \includegraphics[width=0.6\textwidth, page=6]{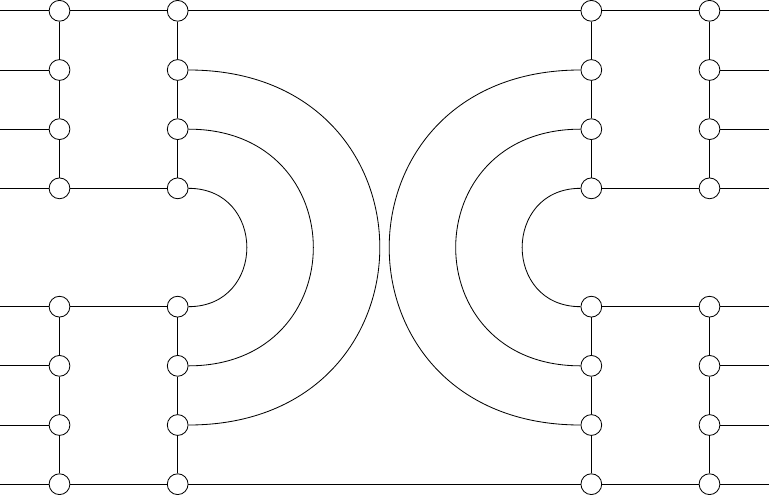}
    \caption{Spanning operators $F_{i_\mu j_\nu \ i'_{\mu'} j'_{\nu'}}^{\mu\nu \ \ \mu'\nu'}(p-1,\alpha,\alpha')$ from Equation~(\ref{eqn:F_poperator}) for $\mathcal{A}^d_{p,p}$ from  adapted to $\mathbb{C}[\s_p] \times \mathbb{C}[\s_p]$ represented as tensor network of Clebsch--Gordan tensors. }
    \label{fig:Matrix_units_original}
\end{figure}

\section{Calculating spectra of Walled Brauer Algebra elements}\label{sec:spectra_of_wba_elements}

In this section, we will investigate the properties of the two operators of a special form, which are closely connected with multi-port-based teleportation schemes (MPBT).  These operators are of the following form:
\begin{align}\label{eqn:rho_operator1}
    \rho(p) := \frac{1}{(p!)^2}\sum_{\sigma_1,\sigma_2 \in \s_p} \Big(V_{\sigma_1}\otimes V_{\sigma_2}\Big) V^{(p)} \Big(V_{\sigma_{1}^{-1}}\otimes V_{\sigma_{2}^{-1}}\Big),
\end{align}
\begin{align}\label{eqn:rho_operator2}
    \rho(p-1) := \frac{1}{(p!)^2}\sum_{\sigma_1,\sigma_2 \in \s_p} \Big(V_{\sigma_{1}}\otimes V_{\sigma_{2}}\Big) V^{(p-1)} \Big(V_{\sigma_{1}^{-1}}\otimes V_{\sigma_{2}^{-1}}\Big),
\end{align}
where $V_{\sigma_{1}},V_{\sigma_{2}}$ are operators that permute $p$ systems according to permutation $\sigma_1, \sigma_2\in \s_p$ respectively. Operators $\rho(p)$ and $\rho(p-1)$ belong to the algebra of partially permuted operators $\mathcal{A}^{d}_{p,p}$. In other words, the operators are the matrix representations of the elements from the walled Brauer algebra $\mathcal{B}^d_{p,p}$. To explore the properties of $\rho(p), \rho(p-1)$, we will heavily rely on the orthogonal irreducible basis for commutant $U^{\otimes p}\otimes \overline{U}^{\otimes p}$, or equivalently for the algebra $\mathcal{A}^{d}_{p,p}$ discussed in the previous sections. In particular, we will exploit an irreducible orthogonal basis for the ideals $\mathcal{M}^{(p)}$, $\mathcal{M}^{(p-1)}$ discussed in sections~\ref{sec:ideal_m}, and~\ref{sec:ideal_m-1} respectively.

Our consideration we start by introducing a notion of the twirl operator $\operatorname{twirl}:(\mathbb{C}^d)^{\otimes p}\otimes (\mathbb{C}^d)^{\otimes p} \rightarrow (\mathbb{C}^d)^{\otimes p}\otimes (\mathbb{C}^d)^{\otimes p}$ given as

\begin{align}\label{eqn:def_of_twirl}
    \forall X\in (\mathbb{C}^d)^{\otimes p}\otimes (\mathbb{C}^d)^{\otimes p}\quad \operatorname{twirl}(X):=\frac{1}{(p!)^2}\sum_{\sigma_1,\sigma_2\in\s_p}V_{\sigma_1}\otimes V_{\sigma_2} X V_{\sigma_1^{-1}}\otimes V_{\sigma_2^{-1}},
\end{align}
where $V_{\sigma_{1}},V_{\sigma_{2}}$ is a permutation operators for permutations $\sigma_1,\sigma_2\in \s_p$. The above definition clearly shows that operators~\eqref{eqn:rho_operator1},~\eqref{eqn:rho_operator2} can be expressed through~\eqref{eqn:def_of_twirl} as
\begin{align}
\rho(p)=\operatorname{twirl}(V^{(p)}),\qquad \rho(p-1)=\operatorname{twirl}(V^{(p-1)}).
\end{align}
We will use this notation later on to simplify notation. Now, we prove a more general fact regarding the $twir$ operator behavior under the trace operator with irreducible matrix units~\eqref{eqn:basis_Eij} for the group algebra $\mathbb{C}[\s_p]$. Namely, we have the following:

\begin{lemma}\label{fact:twirl_of_general_operators}
    Let us consider the $\operatorname{twirl}$ operator given through~\eqref{eqn:def_of_twirl}. For arbitrary operators $X,Y\in (\mathbb{C}^d)^{\otimes p}\otimes (\mathbb{C}^d)^{\otimes p}$, and the irreducible matrix units $E_{i_\mu j_\mu}^\mu,E_{k_\nu l_\nu}^{\nu}$ from~\eqref{eqn:basis_Eij} for the algebra $\mathbb{C}[\s_p]$, the following equality holds
    \begin{equation}
    \label{eq:twirl}
        \tr(\operatorname{twirl}(X)E_{i_\mu j_\mu}^\mu\otimes E_{k_\nu l_\nu }^{\nu} Y E_{i'_{\mu'} j'_{\mu'}}^{\mu'}\otimes E_{k'_{\nu'} l'_{\nu'}}^{\nu'}) = \sum_{r_\mu,s_\nu=1}^{d_\mu,d_\nu}\frac{\tr(X E_{r_\mu j_\mu}^\mu\otimes E_{s_\nu l_\nu }^\nu Y E_{i'_{\mu} r_\mu}^\mu\otimes E_{k'_{\nu} s_\nu}^{\nu})}{d_\mu d_\nu}\delta^{\mu \mu'}\delta^{\nu \nu'} \delta_{i_\mu j'_{\mu'}}\delta_{k_{\nu} l'_{\nu'}}.
    \end{equation}
\end{lemma}

\begin{proof} 
    The proof goes by straightforward calculations by exploiting the definition of $\operatorname{twirl}(X)$ given by \eqref{eqn:def_of_twirl}:
    \begin{align}
    &\tr(\operatorname{twirl}(X)E_{i_\mu j_\mu}^\mu\otimes E_{k_\nu l_\nu }^{\nu} Y E_{i'_{\mu'}j'_{\mu'}}^{\mu'}\otimes E_{k'_{\nu'}l'_{\nu'}}^{\nu'}) =\\
        &\frac{1}{(p!)^2}\sum_{\sigma_1,\sigma_2\in\s_p}\tr(XV_{\sigma_1^{-1}}\otimes V_{\sigma_2^{-1}} E_{i_\mu j_\mu}^\mu \otimes E_{k_\nu l_\nu}^{\nu} Y E_{i'_{\mu'}j'_{\mu'}}^{\mu'}\otimes E_{k'_{\nu'}l'_{\nu'}}^{\nu'} V_{\sigma_1}\otimes V_{\sigma_2})\\
        &=\frac{1}{(p!)^2}\sum_{\sigma_1,\sigma_2\in\s_p}\sum_{\substack{\mu,\nu\\\mu',\nu'}}\sum_{\substack{r_\mu,s_\nu=1 \\ r'_{\mu'},s'_{\nu'}=1}}^{d_\mu, d_{\mu'},d_{\nu},d_{\nu'}}\varphi_{r_\mu i_\mu}^\mu(\sigma_1^{-1})\varphi_{s_\nu k_\nu}^\nu(\sigma_2^{-1})\varphi_{j'_{\mu'}r'_{\mu'}}^{\mu'}(\sigma_1)\varphi_{l'_{\nu'}s'_{\nu'}}^{\nu'}(\sigma_2)\tr(XE_{r_\mu j_\mu}^\mu\otimes E_{s_\nu l_\nu}^\nu Y E_{i'_{\mu'}r'_{\mu'}}^{\mu'}\otimes E_{k'_{\nu'}s'_{\nu'}}^{\nu'})\label{eql:line1}\\
        &=\sum_{\substack{r_\mu,s_\nu=1 \\ r'_{\mu'},s'_{\nu'}=1}}^{d_\mu, d_{\mu'},d_{\nu},d_{\nu'}}\delta^{\mu \mu'}\delta^{\nu \nu'}\delta_{r_\mu r'_{\mu'}}\delta_{i_\mu j'_{\mu'}}\delta_{s_\nu s'_{\nu'}}\delta_{k_\nu l'_{\nu'}}\frac{\tr(XE_{r_\mu j_\mu}^\mu\otimes E_{s_\nu l_\nu}^\nu Y E_{i'_{\mu'}r'_{\mu'}}^{\mu'}\otimes E_{k'_{\nu'}s'_{\nu'}}^{\nu'})}{d_\mu d_\nu}\label{eql:line2}\\
        &=\sum_{r_\mu,s_\nu=1}^{d_\mu,d_\nu}  \frac{\tr(XE_{r_\mu j_\mu}^\mu\otimes E_{s_\nu l_\nu}^\nu Y E_{i'_{\mu} r_{\mu}}^\mu \otimes E_{k'_{\nu}s_{\nu}}^\nu)}{d_\mu d_\nu}\delta^{\mu \mu'}\delta^{\nu \nu'}\delta_{i_\mu j'_{\mu'}}\delta_{k_\nu l'_{\nu'}}.
    \end{align}
    In equation~\eqref{eql:line1} we use properties from~\eqref{eq:actionVonE} together with~\eqref{eqn:composition_of_e_operators}. In equation~\eqref{eql:line2} we apply 
    the orthogonality relation for irreps of the form $ \sum_{\sigma \in \s_p}\varphi _{i_\mu j_\mu}^{\alpha }( \sigma^{-1}) \varphi _{k_\nu l_\nu}^{\beta }(\sigma)=\frac{p!}{d_{\alpha}}\delta^{\alpha \beta}\delta_{j_\mu k_\nu}\delta _{i_\mu l_\nu}.$
\end{proof}
Now let us discuss the consequences of Fact~\ref{fact:twirl_of_general_operators} in the particular choices of the operators $\operatorname{twirl}(X),Y$. First, let us take as an operator $Y$ operator $V^{(p)}$ or $V^{(p-1)}$ which by Definition~\ref{Def:opF} contribute in constructing irreducible matrix units in ideals $\mathcal{M}^{(p)}$ and $\mathcal{M}^{(p-1)}$.  Then due to the result of  Fact~\ref{fact:twirl_of_general_operators}, only their certain subset survives tracing with $\operatorname{twirl}(X)$, where $X$ is for the time being an arbitrary operator. Being more strict, without loss of generality, since we are interested in non-zero contributions when calculating mentioned traces, we can restrict ourselves to the following subset from Definition~\ref{Def:opF}:
\begin{align}
    F_{i_\mu j_\mu \ i'_{\nu} j'_{\nu}}^{\mu \ \ \ \ \nu}(p) \mapsto F_{i_\mu j_\mu \ i_\mu j_\mu}^{\mu \ \ \ \ \mu}(p) = \Big(E_{i_\mu j_\mu}^{\mu} \otimes\id \Big) V^{(p)} \Big(E_{j_\mu i_\mu}^{\mu}\otimes \id\Big)
\end{align}
\begin{align}
\label{eq:map2}
    F_{i_\mu j_\nu \ i'_{\mu'} j'_{\nu'}}^{\mu\nu \ \ \mu'\nu'}(p-1,\alpha,\alpha')\mapsto F_{i_\mu j_\nu \ i_\mu j_\nu}^{\mu\nu \ \ \mu\nu}(p-1,\alpha,\alpha') =  E^{\mu}_{i_\mu 1_\alpha}\otimes E^{\nu}_{j_\nu 1_\alpha} V^{(p-1)} E^{\ \ \ \mu}_{1_{\alpha'} i_\mu}\otimes E^{\ \ \ \nu}_{1_{\alpha'} j_\nu}. 
\end{align}
Indeed, for $Y=V^{(p)}$, the operator $\operatorname{twirl}(X)$ is multiplying by $F_{i_\mu j_\mu \ i'_{\nu} j'_{\nu}}^{\mu \ \ \ \ \nu}(p)$, and we have:
    \begin{align}
        \tr(\operatorname{twirl}(X)\left(E_{i_\mu j_\mu}^\mu\otimes E_{k_\nu l_\nu }^{\nu}\right) V^{(p)} \left(E_{i'_{\mu'}j'_{\mu'}}^{\mu'}\otimes E_{k'_{\nu'}l'_{\nu'}}^{\nu'}\right))=\sum_{r_\mu,s_\mu=1}^{d_\mu}\frac{\tr(XF_{r_\mu s_\mu \  r_\mu s_\mu }^{\mu  \ \ \  \ \mu}(p))}{d_\mu^2}\delta^{\mu \mu'}\delta^{\nu \nu'} \delta_{i_\mu j'_{\mu'}}\delta_{k_{\nu} l'_{\nu'}}.
    \end{align}
For $Y=V^{(p-1)}$ to get mapping from~\eqref{eq:map2} we must change a little bit sandwiching in~\eqref{eq:twirl} and write some of the indices in the PRIR notation given in~\eqref{eqn:e_prir_notation}. In the other words, we multiply $\operatorname{twirl}(X)$ by the operator $F_{i_\mu j_\nu \ i'_{\mu'} j'_{\nu'}}^{\mu\nu \ \ \mu'\nu'}(p-1,\alpha,\alpha')$ from~\eqref{eqn:F_poperator}, and get the following:
\begin{align}
\tr(\operatorname{twirl}(X)\Big(E_{i_\mu \ 1_\alpha}^{\mu} \otimes E_{j_\nu \ 1_\alpha}^{\nu}\Big) V^{(p-1)} \Big(E_{1_{\alpha'} \ i'_{\mu'}}^{\ \ \ \ \mu'} \otimes E_{1_{\alpha'} \ j'_{\nu'}}^{\ \ \ \ \nu'}\Big)) =\sum_{r_\mu,s_\nu=1}^{d_\mu,d_\nu}\frac{\tr(X F^{\mu\nu \ \ \mu \nu}_{r_\mu s_\nu \ r_{\mu}s_{\nu}}(p-1,\alpha,\alpha'))}{d_\mu d_\nu}\delta^{\mu \mu'}\delta^{\nu \nu'} \delta_{i_\mu i'_{\mu'}}\delta_{j_{\nu} j'_{\nu'}}
\end{align}
This directly translates to the choice of irreducible matrix units in $\mathcal{M}^{(p)}$ and $\mathcal{M}^{(p-1)}$ giving a non-zero contribution to discussed above traces:
\begin{align}
\mathcal{G}_{i_\mu j_\mu \ \ i'_\nu j'_\nu }^{\mu \ \ \ \ \  \nu}(p) \mapsto \mathcal{G}_{i_\mu j_\mu \ \ i_\nu j_\nu }^{\mu \ \ \ \ \ \nu}(p), \qquad \G^{\mu\nu \ \ \mu'\nu'}_{i_\mu j_\nu \ i'_{\mu'}j'_{\nu'}}(p-1,\alpha,\alpha')\mapsto \G^{\mu\nu \ \ \mu \nu}_{i_\mu j_\nu \ i_{\mu}j_{\nu}}(p-1,\alpha,\alpha').
\end{align}
Having the above observations and taking $\operatorname{twirl}(X)$ to be $\rho(p)$ from~\eqref{eqn:rho_operator1} or $\rho(p-1)$ from~\eqref{eqn:rho_operator2}, we can evaluate the non-zero matrix elements of the operators being under interest using irreducible matrix units of the ideals $\mathcal{M}^{(p)}$ and $\mathcal{M}^{(p-1)}$. We summarize our findings in the following theorem:

\begin{theorem}\label{thm:spec_of_rho_Gp}
    The matrix elements of operator $\rho(p), \rho(p-1)$, given through~\eqref{eqn:rho_operator1},~\eqref{eqn:rho_operator2} respectively in the irreducible bases from Theorem~\ref{thm:basis_for_max_ideal} and Theorem~\ref{thm:irrepsMp1} are the following:
\begin{enumerate}[1)]
\item 
\begin{align}
\label{eq:1}
\forall \mu \quad \forall i_\mu,j_\mu \quad \tr( \rho(p) \mathcal{G}^{\mu \ \ \ \mu}_{i_\mu j_\mu \ i_\mu j_\mu}(p)) =\frac{m_\mu}{d_{\mu}},
\end{align}
where $m_{\mu},d_{\mu}$ are multiplicity and dimension of the irrep $\mu \vdash p$ in the Schur-Weyl duality.
\item  
\begin{align}
\label{eq:2}
  \forall \mu \quad   \forall i_\mu,j_\mu \quad \tr( \rho(p-1) \mathcal{G}^{\mu \ \ \ \mu}_{i_\mu j_\mu \ i_\mu j_\mu}(p)) = \frac{1}{d}\frac{m_\mu}{d_\mu},
    \end{align}
    where $m_\mu, d_\mu$ are multiplicity and the dimension of the irrep $\mu \vdash p$ in the Schur-Weyl duality.
\item 
    \begin{align}
    \label{eq:3}
    \forall \mu,\nu \quad \forall i_\mu,j_\nu \quad \tr( \rho(p) \mathcal{G}^{\mu\nu \ \ \mu\nu}_{i_\mu j_\nu \ i_\mu j_\nu}(p-1,\alpha,\alpha')) = 0,
    \end{align}
\item  $\forall \mu,\nu \quad \forall i_\mu,j_\nu$
    \begin{align}\label{eqn:4}
         &\tr( \rho(p-1) \mathcal{G}^{\mu\nu \ \ \mu\nu}_{i_\mu j_\nu \ i_\mu j_\nu}(p-1,\alpha,\alpha'))\\
         &= \frac{\sum_{\substack{\alpha\in \mu \wedge \nu \\ \alpha' \in \mu\wedge \nu}}\Bigl(U^{\mu \nu}\Bigr)^\dagger_{\alpha \beta} U^{\mu \nu}_{\beta' \alpha'}}{d \sqrt{b^{\mu\nu}(\alpha) b^{\mu\nu}(\alpha')}}\Bigg(\frac{1}{d^2 - 1} \frac{1}{d_\mu d_\nu}\Bigg[ d m_\mu^2 d_\mu \delta^{\mu \nu} - \frac{m_\mu^3}{m_\alpha}d_\alpha \delta^{\mu \nu}\delta^{\beta \alpha}\delta^{\beta' \alpha} \delta^{\beta \beta'}\\
            &+ d\frac{(m_\mu m_\nu)^2}{m_\alpha m_{\alpha'}}d_\alpha \delta^{\alpha \alpha'}\delta^{\beta \alpha}\delta^{\beta' \alpha} \delta^{\beta \alpha'}\delta^{\beta' \alpha'} - \frac{m_\mu^3}{m_{\alpha'}}d_{\alpha'} \delta^{\mu \nu}\delta^{\beta \alpha'}\delta^{\beta' \alpha'}\Bigg]- \frac{1}{d}\frac{m_\mu^2}{d_\mu} \delta^{\mu \nu}\Bigg),
    \end{align}
    where $m_\mu,m_\nu$ are multiplicities and $d_\mu,d_\nu$ are the dimensions of the irreps $\mu\vdash p$, $\nu\vdash p$ respectively in the Schur-Weyl duality.
\end{enumerate}
\end{theorem}
    
    \begin{proof}
    We will prove the statement of the theorem point by point. 
    \begin{itemize}
        \item Proof of equality~\eqref{eq:1}.\\ Applying result of Fact~\ref{fact:twirl_of_general_operators} when $X=V^{(p)}$ and $Y=\mathcal{G}_{i_\mu j_\mu \ i_\nu j_\nu }^{\mu  \ \ \ \nu}(p)$, we have $\forall i_\mu,j_\mu$:
        \begin{align}
        \tr( \rho(p) \mathcal{G}^{\mu \ \ \ \mu}_{i_\mu j_\mu \ i_\mu j_\mu}(p)) &=\frac{1}{m_\mu  d_{\mu}^2}\sum_{r_\mu,s_\mu=1}^{d_{\mu}}\tr(V^{(p)}\left(E^\mu_{r_\mu s_\mu}\otimes \id\right)V^{(p)}\left(E^\mu_{s_\mu r_\mu}\otimes \id\right))\\
        &=\frac{1}{m_\mu  d_{\mu}^2}\sum_{r_\mu,s_\mu=1}^{d_{\mu}}\tr(\tr(E^\mu_{r_\mu s_\mu})V^{(p)}\left(E^\mu_{s_\mu r_\mu}\otimes \id\right)) \label{eq228a}\\
        &=\frac{1}{d_{\mu}^2}\sum_{r_\mu,s_\mu=1}^{d_{\mu}}\delta_{r_\mu s_\mu}\tr(V^{(p)}\left(E^\mu_{s_\mu r_\mu}\otimes \id\right)) \\
        &=\frac{1}{d_{\mu}^2}\sum_{r_\mu=1}^{d_{\mu}}\tr(V^{(p)}\left(E^\mu_{r_\mu r_\mu}\otimes \id\right))\label{eq228b}\\
        &=\frac{1}{d_{\mu}^2}\tr(V^{(p)}\left(P^\mu\otimes \id\right))\label{eq228c}\\
       &=\frac{1}{d_{\mu}^2}\tr(\tr_{p+1,\ldots, 2p}(V^{(p)})P^\mu)\\
     &=\frac{1}{d_{\mu}^2}\tr(P^\mu)=\frac{m_\mu}{d_{\mu}}. \label{eq228d}
    \end{align}
    In equation~\eqref{eq228a} we use Fact~\ref{fact:sandwitch_of_v}. In equation~\eqref{eq228b} we use the definition of Young projector~\eqref{eq:Young}, and in ~\eqref{eq228d} we apply trace rule~\eqref{eq:Young}.
    \item Proof of equality~\eqref{eq:2}.\\ Applying result of Fact~\ref{fact:twirl_of_general_operators} when $X=V^{(p-1)}$ and $Y=\mathcal{G}_{i_\mu j_\mu \ i_\mu j_\mu }^{\mu  \ \ \ \mu}(p)$, we have $\forall i_\mu,j_\mu$:
        \begin{align}
            \tr(\rho(p-1) \mathcal{G}^{\mu \ \ \ \mu}_{i_\mu j_\mu \ i_\mu j_\mu}(p)) &=  \frac{1}{m_\mu d_\mu^2} \sum_{r_\mu,s_\mu =1}^{d_\mu}\tr(V^{(p-1)} \left(E^\mu_{r_\mu s_\mu}\otimes \id\right)V^{(p)}\left(E^\mu_{s_\mu r_\mu}\otimes \id\right))
        \end{align}
        By Proposition \ref{thm:wba_element} and the fact that $V^{(p)} = V^{(p-1)} \otimes V^{(1)}$ we write following
        \begin{align}
            \tr( \operatorname{twirl}(V^{(p-1)}) \mathcal{G}^{\mu \ \ \ \mu}_{i_\mu j_\mu \ i_\mu j_\mu}(p)) &=  \frac{1}{m_\mu d_\mu^2} \sum_{r_\mu,s_\mu=1}^{d_\mu}\tr( (a V^{(p)} + b V^{(p-1)})V' (E^{\mu}_{s_\mu r_\mu}\otimes \id) )\\
            &= \frac{1}{m_\mu d_\mu^2} \sum_{r_\mu,s_\mu=1}^{d_\mu}\tr( (a dV^{(p)} + b V^{(p)}) (E^{\mu}_{s_\mu r_\mu}\otimes \id) )\\
            &=  \frac{1}{m_\mu d_\mu^2} \sum_{r_\mu,s_\mu=1}^{d_\mu}(a d + b)\tr(V^{(p)} (E^{\mu}_{s_\mu r_\mu}\otimes \id))\label{eqn:element_of_rho_in_gp},
        \end{align}
        where the coefficients $a$ and $b$ with respect to the Theorem \ref{thm:wba_element} are given by
        \begin{align}
            a &= \frac{1}{d^3-d}\Bigl(d\tr\Big(V^{(p)} E_{r_\mu s_\mu}^\mu\otimes \id \Big) - \tr\Big(V^{(p-1)} E_{r_\mu s_\mu}^\mu\otimes \id \Big) \Bigr)\\
            b&= \frac{1}{d^3-d}\Bigl(d\tr\Big(V^{(p-1)} E_{r_\mu s_\mu}^\mu\otimes \id \Big) - \tr\Big(V^{(p)} E_{r_\mu s_\mu}^\mu\otimes \id \Big) \Bigr)
        \end{align}
       By Lemma \ref{lemma:ad+b=m_mu} we can replace the term $ad+b$ with $m_\mu/d$, observe that there is no delta because coefficients $a$ and $b$ have only one label $\mu$, hence we reduce 
        \eqref{eqn:element_of_rho_in_gp} to 
        \begin{align}
           \tr( \rho(p-1) \mathcal{G}^{\mu \ \ \ \mu}_{i_\mu j_\mu \ i_\mu j_\mu}(p))  &= \frac{1}{m_\mu d_\mu^2} \sum_{r_\mu,s_\mu=1}^{d_\mu} \frac{m_\mu}{d}\tr(V^{(p)} (E^{\mu}_{s_\mu r_\mu}\otimes \id))\\
            &=  \frac{1}{m_\mu d_\mu^2} \sum_{r_\mu,s_\mu=1}^{d_\mu} \frac{m_\mu^2}{d}\delta_{r_\mu s_\mu}\label{eqn:trace_of_rho_p-1_G_transformation}\\
            &= \frac{1}{d}\frac{m_\mu}{d_\mu}
        \end{align}
        where in the line \eqref{eqn:trace_of_rho_p-1_G_transformation} we used a following fact that $\sum_{r_\mu,s_\mu}^{d_\mu} \delta_{r_\mu s_\mu} = d_\mu$.
        
        \item Proof of equality \eqref{eq:3}.\\ Here we are going to prove that for every element $H^{\mu \nu \ \ \mu \nu}_{i_\mu j_\nu \ i_\mu j_\nu}( p-1,\alpha,\alpha')$ given by Theorem \ref{thm:lower_ideal_operator_G} traced with the operator $\rho(p-1)$ is zero. Due to Theorem \ref{thm:irrepsMp1} the irreducible matrix units $\mathcal{G}^{\mu\nu \ \ \mu\nu}_{i_\mu j_\nu \ i_\mu j_\nu}(p-1,\alpha,\alpha')$ in ideal $\mathcal{M}^{(p-1)}$ are written as linear combinations of the operators $H^{\mu \nu \ \ \mu \nu}_{i_\mu j_\nu \ i_\mu j_\nu}(p-1,\alpha,\alpha')$. Indeed, if the overlap of the operator $H^{\mu \nu \ \ \mu \nu}_{i_\mu j_\nu \ i_\mu j_\nu}(p-1,\alpha,\alpha')$ and $\rho(p-1)$ is zero then the linear combination will also be zero, hence it is sufficient to only prove the first part of a statement. Applying result of Fact~\ref{fact:twirl_of_general_operators} when $X=V^{(p)}$ and $Y=H_{i_\mu j_\nu \ i_\mu j_\nu }^{\mu \nu \ \ \mu \nu}(p-1,\alpha,\alpha')$, we have $\forall i_\mu,j_\nu$:
        \begin{align}
             \tr( \rho(p) H^{\mu\nu \ \ \mu\nu}_{i_\mu j_\nu \ i_\mu j_\nu}(p-1,\alpha,\alpha')) &= 
             \tr( \rho(p) \Big(d F^{\mu \nu \ \ \mu \nu}_{i_\mu j_\nu \ i_\mu j_\nu}(p-1,\alpha,\alpha') - F^{\mu \ \ \ \mu}_{i_\mu j_\mu \ i_\mu j_\mu}(p)\delta^{\mu \nu} \Big)  )\label{eqn:p_element_theorem1}
        \end{align}
        where
        \begin{align}\label{eqn:fp-1_theorem}
            F^{\mu \nu \ \ \mu \nu}_{i_\mu j_\nu \ i_\mu j_\nu}(p-1,\alpha,\alpha') = E^{\mu}_{i_\mu 1_\alpha}\otimes E^{\nu}_{j_\nu 1_\alpha} V^{(p-1)} E^{\ \ \ \mu}_{1_{\alpha'} i_\mu}\otimes E^{\ \ \ \nu}_{1_{\alpha'} j_\nu}
        \end{align}
        \begin{align}\label{eqn:fp_theorem}
            F^{\mu \ \ \ \mu}_{i_\mu j_\mu \ i_\mu j_\mu}(p) = \Big(E_{i_\mu j_\mu}^{\mu}\otimes 
            \id\Big) V^{(p)} \Big(E^{\mu}_{j_\mu i_\mu}\otimes \id\Big).
        \end{align}
        Implementing \eqref{eqn:fp-1_theorem} and \eqref{eqn:fp_theorem} into \eqref{eqn:p_element_theorem1} we get
        \begin{align}
             &\tr( \rho(p)  H^{\mu\nu \ \ \mu\nu}_{i_\mu j_\nu \ i_\mu j_\nu}(p-1,\alpha,\alpha'))\label{eqn:p_element_therem_calculaterd0} \\
             &=  \Bigg(d \tr( \rho(p)  E^{\mu}_{i_\mu 1_\alpha}\otimes E^{\nu}_{j_\nu 1_\alpha} V^{(p-1)} E^{\ \ \ \mu}_{1_{\alpha'} i_\mu}\otimes E^{\ \ \ \nu}_{1_{\alpha'} j_\nu} )\label{eqn:p_element_therem_notcalculated1} \\
            &- \tr( \rho(p) \Big(E_{i_\mu j_\mu}^{\mu}\otimes 
            \id\Big) V^{(p)} \Big(E^{\mu}_{j_\mu i_\mu}\otimes \id\Big) )\delta^{\mu \nu}\Bigg)\label{eqn:p_element_therem_calculated1}
        \end{align}
        Let us take care of each term separately, namely first term \eqref{eqn:p_element_therem_notcalculated1} we can write as follows
        \begin{align}
            d &\tr(\rho(p)  E^{\mu}_{i_\mu 1_\alpha}\otimes E^{\nu}_{j_\nu 1_\alpha} V^{(p-1)} E^{\ \ \ \mu}_{1_{\alpha'} i_\mu}\otimes E^{\ \ \ \nu}_{1_{\alpha'} j_\nu} )\\
            &= \frac{d}{d_\mu d_\nu}\sum_{r_\mu,s_\nu = 1}^{d_\mu d_\nu} \tr(V^{(p)} E^{\mu}_{r_\mu 1_\alpha}\otimes E^{\nu}_{s_\nu 1_\alpha} V^{(p-1)} E^{\ \ \ \mu}_{1_{\alpha'} r_\mu}\otimes E^{\ \ \ \nu}_{1_{\alpha'} s_\nu} )\\
            &=\frac{d}{d_\mu d_\nu}\sum_{r_\mu,s_\nu = 1}^{d_\mu d_\nu} \tr(V^{(1)}\otimes V^{(p-1)} E^{\mu}_{r_\mu 1_\alpha}\otimes E^{\nu}_{s_\nu 1_\alpha} V^{(p-1)} E^{\ \ \ \mu}_{1_{\alpha'} r_\mu}\otimes E^{\ \ \ \nu}_{1_{\alpha'} s_\nu} )
        \end{align}
        The term sandwiched by $V^{(p-1)}$ can be written in terms of coefficients $a$ and $b$ according to Theorem \ref{thm:wba_element}, hence we get
        \begin{align}
            \frac{d}{d_\mu d_\nu}&\sum_{r_\mu,s_\nu = 1}^{d_\mu d_\nu} \tr(V^{(1)}\otimes V^{(p-1)} E^{\mu}_{r_\mu 1_\alpha}\otimes E^{\nu}_{s_\nu 1_\alpha} V^{(p-1)} E^{\ \ \ \mu}_{1_{\alpha'} r_\mu}\otimes E^{\ \ \ \nu}_{1_{\alpha'} s_\nu} )\\
            &=\frac{d}{d_\mu d_\nu}\sum_{r_\mu,s_\nu = 1}^{d_\mu d_\nu} \tr(V^{(1)}(a V^{(p)}+b V^{(p-1)}) E^{\ \ \ \mu}_{1_{\alpha'} r_\mu}\otimes E^{\ \ \ \nu}_{1_{\alpha'} s_\nu} )\\
            &= \frac{d}{d_\mu d_\nu}\sum_{r_\mu,s_\nu = 1}^{d_\mu d_\nu} \tr((a d V^{(p)}+b V^{(p)}) E^{\ \ \ \mu}_{1_{\alpha'} r_\mu}\otimes E^{\ \ \ \nu}_{1_{\alpha'} s_\nu} )\\
            &= \frac{d}{d_\mu d_\nu}\sum_{r_\mu,s_\nu = 1}^{d_\mu d_\nu} \tr((a d +b)V^{(p)} E^{\ \ \ \mu}_{1_{\alpha'} r_\mu}\otimes E^{\ \ \ \nu}_{1_{\alpha'} s_\nu} )
        \end{align}
        By Lemma \ref{lemma:ad+b=m_mu} we can replace $(ad+b)=\frac{m_\mu}{d}\delta^{\mu\nu}$ and we get that 
        \begin{align}
            d &\tr( \operatorname{twirl}(V^{(p)}) E^{\mu}_{i_\mu 1_\alpha}\otimes E^{\nu}_{j_\nu 1_\alpha} V^{(p-1)} E^{\ \ \ \mu}_{1_{\alpha'} i_\mu}\otimes E^{\ \ \ \nu}_{1_{\alpha'} j_\nu} ) = \frac{m_\mu}{d_\mu d_\nu}\sum_{r_\mu,s_\nu = 1}^{d_\mu d_\nu} \tr(E^{\ \ \ \mu}_{1_{\alpha'} r_\mu}\otimes E^{\ \ \ \nu}_{1_{\alpha'} s_\nu} V^{(p)})
        \end{align}
        The trace value can be evaluated by Lemma \ref{lem:trace_values}, getting us 
        \begin{align}
            d \tr( \rho(p)  E^{\mu}_{i_\mu 1_\alpha}\otimes E^{\nu}_{j_\nu 1_\alpha} V^{(p-1)} E^{\ \ \ \mu}_{1_{\alpha'} i_\mu}\otimes E^{\ \ \ \nu}_{1_{\alpha'} j_\nu} ) &= \frac{m_\mu}{d_\mu d_\nu}\sum_{r_\mu,s_\nu = 1}^{d_\mu d_\nu} m_\mu \delta^{\mu \nu}\delta_{r_\mu s_\nu}\\
            &=\frac{m_\mu}{d_\mu d_\nu} m_\mu d_\mu \delta^{\mu \nu}\\
            &= \frac{m_\mu^2}{d_{\mu}}\delta^{\mu \nu}\label{eqn:tr_vp_fp-1}
        \end{align}
        The second part of evaluating term \eqref{eqn:p_element_therem_calculated1} proceeds the same way as proof of relation \eqref{eq:1}, hence we can write that
        \begin{align}
            \tr( \rho(p) \Big(E_{i_\mu j_\mu}^{\mu}\otimes 
            \id\Big) V^{(p)} \Big(E^{\mu}_{j_\mu i_\mu}\otimes \id\Big)\delta^{\mu \nu} ) = \frac{m_\mu^2}{d_\mu}\delta^{\mu \nu}.\label{eqn:tr_vp_fp}
        \end{align}
        Implementing the results \eqref{eqn:tr_vp_fp-1} and \eqref{eqn:tr_vp_fp} into \eqref{eqn:p_element_therem_calculaterd0} we get that 
        \begin{align}
             \tr( \rho(p)  H^{\mu\nu \ \ \mu\nu}_{i_\mu j_\nu \ i_\mu j_\nu}(p-1,\alpha,\alpha')) = \Bigg(\frac{m_\mu^2}{d_{\mu}}\delta^{\mu \nu} - \frac{m_\mu^2}{d_\mu}\delta^{\mu \nu}\Bigg) = 0
        \end{align}
        Due to Theorem \ref{thm:irrepsMp1} the irreducible matrix units in the ideal $\mathcal{M}^{(p-1)}$ is build by operator $\mathcal{G}^{\mu\nu \ \ \mu\nu}_{i_\mu j_\nu \ i_\mu j_\nu}(p-1,\alpha,\alpha')$ which are the linear combinations of operators $H^{\mu\nu \ \ \mu\nu}_{i_\mu j_\nu \ i_\mu j_\nu}(p-1,\alpha,\alpha')$, hence we conclude that 
        \begin{align}
             \tr(\rho(p) \mathcal{G}^{\mu\nu \ \ \mu\nu}_{i_\mu j_\nu \ i_\mu j_\nu}(p-1,\alpha,\alpha'))  = 0
        \end{align}
        which finishes the proof.
        
        \item Proof of equality \eqref{eqn:4}.\\ Applying result of Fact~\ref{fact:twirl_of_general_operators} when $X=V^{(p-1)}$ and $Y=\mathcal{G}_{i_\mu j_\nu \ i_\mu j_\nu }^{\mu \nu \ \ \mu \nu}(p-1,\alpha,\alpha')$, we have $\forall i_\mu,j_\nu$:
        \begin{align}
             &\tr( \rho(p-1) \mathcal{G}^{\mu\nu \ \ \mu\nu}_{i_\mu j_\nu \ i_\mu j_\nu}(p-1,\alpha,\alpha')) \\
             &= 
             \sum_{\substack{\alpha\in \mu \wedge \nu \\ \alpha' \in \mu\wedge \nu}}\Bigl(U^{\mu \nu}\Bigr)^\dagger_{\alpha \beta} U^{\mu \nu}_{\beta' \alpha'}\tr( \rho(p-1) \Bigg(\frac{d F^{\mu \nu \ \ \mu \nu}_{i_\mu j_\nu \ i_\mu j_\nu}(p-1,\alpha,\alpha') - F^{\mu \ \ \ \mu}_{i_\mu j_\mu \ i_\mu j_\mu}(p)\delta^{\mu \nu}}{d \sqrt{b^{\mu\nu}(\alpha) b^{\mu\nu}(\alpha')}}  \Bigg)  ).\label{eqn:p_element_theorem}
        \end{align}
        Implementing \eqref{eqn:fp-1_theorem} and \eqref{eqn:fp_theorem} into \eqref{eqn:p_element_theorem} we get
        \begin{align}
             &\tr( \rho(p-1) \mathcal{G}^{\mu\nu \ \ \mu\nu}_{i_\mu j_\nu \ i_\mu j_\nu}(p-1,\alpha,\alpha')) \\
             &=
            \frac{\sum_{\substack{\alpha\in \mu \wedge \nu \\ \alpha' \in \mu\wedge \nu}}\Bigl(U^{\mu \nu}\Bigr)^\dagger_{\alpha \beta} U^{\mu \nu}_{\beta' \alpha'}}{d \sqrt{b^{\mu\nu}(\alpha) b^{\mu\nu}(\alpha')}} \Bigg(d \tr( \rho(p-1) E^{\mu}_{i_\mu 1_\alpha}\otimes E^{\nu}_{j_\nu 1_\alpha} V^{(p-1)} E^{\ \ \ \mu}_{1_{\alpha'} i_\mu}\otimes E^{\ \ \ \nu}_{1_{\alpha'} j_\nu} )\label{eqn:p_element_therem_notcalculated} \\
            &- \tr( \rho(p-1)\Big(E_{i_\mu j_\mu}^{\mu}\otimes 
            \id\Big) V^{(p)} \Big(E^{\mu}_{j_\mu i_\mu}\otimes \id\Big) )\delta^{\mu \nu}\Bigg)\label{eqn:p_element_therem_calculated}
        \end{align}
        We split the equation and we will calculate each term separately and at the end, we will merge the results. The term \eqref{eqn:p_element_therem_calculated} we have calculated in the part of the proof concerning \eqref{eq:2}, hence we only take an interest in \eqref{eqn:p_element_therem_notcalculated}. By Fact \ref{fact:twirl_of_general_operators} we write the following
        \begin{align}
            &\tr( \rho(p-1) E^{\mu}_{i_\mu 1_\alpha}\otimes E^{\nu}_{j_\nu 1_\alpha} V^{(p-1)} E^{\ \ \ \mu}_{1_{\alpha'} i_\mu}\otimes E^{\ \ \ \nu}_{1_{\alpha'} j_\nu} ) = \frac{1}{d_\mu d_\nu}\sum_{r_\mu,s_\nu = 1}^{d_\mu d_\nu} \tr(V^{(p-1)} E^{\mu}_{r_\mu 1_\alpha}\otimes E^{\nu}_{s_\nu 1_\alpha} V^{(p-1)} E^{\ \ \ \mu}_{1_{\alpha'} r_\mu}\otimes E^{\ \ \ \nu}_{1_{\alpha'} s_\nu} )
        \end{align}
        By Theorem \ref{thm:wba_element} we have
        \begin{align}
            &\tr( \rho(p-1) E^{\mu}_{i_\mu 1_\alpha}\otimes E^{\nu}_{j_\nu 1_\alpha} V^{(p-1)} E^{\ \ \ \mu}_{1_{\alpha'} i_\mu}\otimes E^{\ \ \ \nu}_{1_{\alpha'} j_\nu} ) \\
            &= \frac{1}{d_\mu d_\nu}\sum_{r_\mu,s_\nu = 1}^{d_\mu d_\nu} \tr((a\cdot V^{(p)}+b\cdot V^{(p-1)})E^{\ \ \ \mu}_{1_{\alpha'} r_\mu}\otimes E^{\ \ \ \nu}_{1_{\alpha'} s_\nu} )\\
            &= \frac{1}{d_\mu d_\nu}\sum_{r_\mu,s_\nu = 1}^{d_\mu d_\nu} \Bigl( a\tr(V^{(p)} E^{\ \ \ \mu}_{1_{\alpha'} r_\mu}\otimes E^{\ \ \ \nu}_{1_{\alpha'} s_\nu}) + b\tr(V^{(p-1)} E^{\ \ \ \mu}_{1_{\alpha'} r_\mu}\otimes E^{\ \ \ \nu}_{1_{\alpha'} s_\nu} ) \Bigr)\label{eqn:calculations_of_trace_p-1}
        \end{align}
        Where coefficients $a$ and $b$ according to the Theorem \ref{thm:wba_element} have a form
\begin{align}
    b&= \frac{1}{d^3-d}\Bigl( d\tr(E^\mu_{r_\mu 1_\alpha}\otimes E^{\nu}_{s_\nu1_\alpha}V^{(p-1)})-\tr(E^\mu_{r_\mu1_\alpha}\otimes E^{\nu}_{s_\nu1_\alpha}V^{(p)})\Bigr)\label{eqn:coefficient_b_of_trrho}\\
    a&= \frac{1}{d^3-d}\Bigl(d\tr( E^{\mu}_{r_\mu1_\alpha}\otimes E^{\nu}_{s_\nu1_\alpha}V^{(p)}) - \tr(E^\mu_{r_\mu1_\alpha}\otimes E^{\nu}_{s_\nu1_\alpha}V^{(p-1)}) \Bigr)\label{eqn:coefficient_a_of_trrho}
\end{align}
To evaluate the trace values we expand the indices in half-PRIR notation \eqref{eqn:e_halfprir_notation}. Let $r=(\mu,r_\beta,\beta)$,  $s = (\nu,s_{\beta'},\beta')$ and $\alpha,\alpha'\in\mu$, $\alpha,\alpha'\in\nu$ then according to the Lemma \ref{lem:trace_values} we get that 
\begin{align}
    \tr( E_{r_\beta 1_\alpha}^\mu \otimes E_{s_{\beta'} 1_\alpha}^\nu V^{(p-1)}) &= \frac{m_\mu m_\nu}{m_\alpha}\delta_{r_\beta  s_{\beta'}}\delta^{\beta \alpha}\delta^{\beta' \alpha}\label{eqn:trace_value_of_vp-1}\\
     \tr(E_{r_\mu 1_\alpha}^\mu \otimes E_{s_\nu 1_\alpha}^\nu V^{(p)}) &=m_\mu
     \delta^{\mu \nu}\delta_{r_\mu s_\nu}\label{eqn:trace_value_of_vp}
\end{align}
Collecting results \eqref{eqn:trace_value_of_vp-1} and \eqref{eqn:trace_value_of_vp} and implementing them into \eqref{eqn:calculations_of_trace_p-1} together with coefficients $a$ \eqref{eqn:coefficient_a_of_trrho} and $b$ \eqref{eqn:coefficient_b_of_trrho}, we get the following
\begin{align}
    &\sum_{r_\mu,s_\nu=1}^{d_\mu,d_\nu} \Bigl[a \tr(V^{(p)}E_{1_{\alpha'}r_\mu}^{\mu}\otimes E_{1_{\alpha'}s_\nu}^{\nu}) + b\tr(V^{(p-1)} E_{1_{\alpha'}r_\mu}^{\mu}\otimes E_{1_{\alpha'}s_\nu}^{\nu})\Bigr]\frac{1}{d_\mu d_\nu}\\
    &= \frac{1}{d^3-d} \sum_{r_\mu,s_\nu=1}^{d_\mu,d_\nu}\Bigg[ \Bigl( d  m_\mu \delta^{\mu \nu}\delta_{r_\mu s_\nu} - \frac{m_\mu m_\nu}{m_\alpha}\delta_{r_\beta s_{\beta'}} \delta^{\beta \alpha}\delta^{\beta' \alpha}\Bigr) m_\mu \delta^{\mu \nu}\delta_{r_\mu s_\nu} \\
    &+ \Bigl( d \frac{m_\mu m_\nu}{m_\alpha}\delta_{r_\alpha s_{\alpha}}\delta^{\beta \alpha}\delta^{\beta' \alpha} - m_\mu \delta^{\mu \nu}\delta_{r_\mu,s_\nu} \Bigr) \frac{m_\mu m_\nu}{m_{\alpha'}}\delta_{r_{\alpha'} s_{\alpha'}}\delta^{\beta \alpha'}\delta^{\beta' \alpha'} \Bigg]\frac{1}{d_\mu d_\nu}\\
    &= \frac{1}{d^3-d}  \sum_{r_\mu,s_\nu=1}^{d_\mu,d_\nu}\Bigg[ d  m_\mu^2 \delta^{\mu \nu}\delta_{r_\mu s_\nu}  - \frac{m_\mu^3}{m_\alpha}\delta^{\mu \nu}\delta_{r_\beta s_{\beta'}} \delta^{\beta \alpha}\delta^{\beta' \alpha} \delta^{\beta \beta'}\\
    &+ d \frac{(m_\mu m_\nu)^2}{m_\alpha m_{\alpha'}}\delta_{r_\beta s_{\beta'}}\delta_{r_{\beta} s_{\beta'}} \delta^{\beta \alpha}\delta^{\beta' \alpha} \delta^{\beta \alpha'}\delta^{\beta' \alpha'}  - \frac{m_\mu^3}{m_{\alpha'}} \delta^{\mu \nu}\delta_{r_\mu s_\nu} \delta^{\beta \alpha'}\delta^{\beta' \alpha'}\Bigg]\frac{1}{d_\mu d_\nu}\label{eqn:fact:twirlwithfp-1}\\
    &=\frac{1}{d^3-d}\Bigg[dm_\mu^2 d_\mu \delta^{\mu \nu} - \sum_{\beta\in\mu,\beta'\in\nu}\sum_{r_\beta,s_{\beta'}=1}^{d_\beta,d_{\beta'}} \frac{m_\mu^3}{m_\alpha} \delta_{r_\beta s_{\beta'}} \delta^{\beta \alpha}\delta^{\beta' \alpha} \delta^{\beta \beta'}\label{eqn:fact_of_rewriting_sum_in_prir1} \\
    &+\sum_{\beta\in\mu,\beta'\in\nu}\sum_{r_\beta,s_{\beta'}=1}^{d_\beta,d_{\beta'}}d\frac{(m_\mu m_\nu)^2}{m_\alpha m_{\alpha'}}\delta_{r_\beta s_{\beta'}}\delta_{r_{\beta} s_{\beta'}} \delta^{\beta \alpha}\delta^{\beta' \alpha} \delta^{\beta \alpha'}\delta^{\beta' \alpha'}  - \sum_{\beta\in\mu,\beta'\in\nu}\sum_{r_\beta,s_{\beta'}=1}^{d_\beta,d_{\beta'}} \frac{m_\mu^3}{m_{\alpha'}} \delta^{\mu \nu}\delta_{r_\beta s_{\beta'}} \delta^{\beta \alpha'}\delta^{\beta' \alpha'}\Bigg] \frac{1}{d_\mu d_\nu}\label{eqn:fact_of_rewriting_sum_in_prir2}\\
    &= \frac{1}{d^3 - d} \Bigg[ d m_\mu^2 d_\mu \delta^{\mu \nu} - \frac{m_\mu^3}{m_\alpha}d_\alpha \delta^{\mu \nu} + d \frac{(m_\mu m_\nu)^2}{m_\alpha m_{\alpha'}}d_\alpha \delta^{\alpha \alpha'} - \frac{m_\mu^3}{m_{\alpha'}}d_{\alpha'} \delta^{\mu \nu}\Bigg]\frac{1}{d_\mu d_\nu}\label{eqn:thm_eq4_last_part_of_merge}
\end{align}
In the lines \eqref{eqn:fact_of_rewriting_sum_in_prir1}-\eqref{eqn:fact_of_rewriting_sum_in_prir2} we rewrite the sums written in normal notation, using the PRIR notation \eqref{eqn:e_prir_notation} in the following way
\begin{align}
    \sum_{r,s=1}^{d_\mu d_\nu} \rightarrow \sum_{\beta\in\mu,\beta'\in\nu}\sum_{r_\beta,s_{\beta'}}^{d_\beta, d_{\beta'}},
\end{align}
so we can use the fact that combining such sums with deltas, especially $\delta_{r_\alpha s_\alpha}$ gives us the dimension of irrep $\alpha$, i.e. $d_\alpha$. Combining the results \eqref{eqn:thm_eq4_last_part_of_merge} together with \eqref{eq:3} and implementing them into \eqref{eqn:p_element_theorem} we finally arrived at
\begin{align}
            &\tr( \rho(p-1) \mathcal{G}^{\mu\nu \ \ \mu\nu}_{i_\mu j_\nu \ i_\mu j_\nu}(p-1,\alpha,\alpha')) \\
            &= \frac{\sum_{\substack{\alpha\in \mu \wedge \nu \\ \alpha' \in \mu\wedge \nu}}\Bigl(U^{\mu \nu}\Bigr)^\dagger_{\alpha \beta} U^{\mu \nu}_{\beta' \alpha'}}{d \sqrt{b^{\mu\nu}(\alpha) b^{\mu\nu}(\alpha')}}\Bigg(d \tr( \rho(p-1) E^{\mu}_{i_\mu 1_\alpha}\otimes E^{\nu}_{j_\nu 1_\alpha} V^{(p-1)} E^{\ \ \ \mu}_{1_{\alpha'} i_\mu}\otimes E^{\ \ \ \nu}_{1_{\alpha'} j_\nu} ) \\
            &- \tr( \rho(p-1) \Big(E_{i_\mu j_\mu}^{\mu}\otimes 
            \id\Big) V^{(p)} \Big(E^{\mu}_{j_\mu i_\mu}\otimes \id\Big) )\delta^{\mu \nu} \Bigg)\\
            &=\frac{\sum_{\substack{\alpha\in \mu \wedge \nu \\ \alpha' \in \mu\wedge \nu}}\Bigl(U^{\mu \nu}\Bigr)^\dagger_{\alpha \beta} U^{\mu \nu}_{\beta' \alpha'}}{d \sqrt{b^{\mu\nu}(\alpha) b^{\mu\nu}(\alpha')}}\Bigg(\frac{1}{d^2 - 1} \frac{1}{d_\mu d_\nu}\Bigg[ d m_\mu^2 d_\mu \delta^{\mu \nu} - \frac{m_\mu^3}{m_\alpha}d_\alpha \delta^{\mu \nu}\delta^{\beta \alpha}\delta^{\beta' \alpha} \delta^{\beta \beta'}\\
            &+ d\frac{(m_\mu m_\nu)^2}{m_\alpha m_{\alpha'}}d_\alpha \delta^{\alpha \alpha'}\delta^{\beta \alpha}\delta^{\beta' \alpha} \delta^{\beta \alpha'}\delta^{\beta' \alpha'} - \frac{m_\mu^3}{m_{\alpha'}}d_{\alpha'} \delta^{\mu \nu}\delta^{\beta \alpha'}\delta^{\beta' \alpha'}\Bigg]- \frac{1}{d}\frac{m_\mu^2}{d_\mu} \delta^{\mu \nu}\Bigg)
        \end{align}
    \end{itemize}
    \end{proof}


\section{Example: analytical spectra of operators $\rho(k)$ from the algebra \texorpdfstring{$\mathcal{A}^{3}_{3,3}$}{Lg}}\label{sec:example}


In this section, we apply our analytical results from the previous sections to a specific case in which a number of particles are fixed at $6$, i.e. we take $p=3$ with local dimension $d=3$. In particular, we construct irreducible matrix units from Theorems~\ref{thm:basis_for_max_ideal} and~\ref{thm:irrepsMp1} for the ideals $\mathcal{M}^{(3)}$, and $\mathcal{M}^{(2)}$ respectively. Next, we evaluate all non-zero matrix elements of the operators $\rho(3), \rho(2)$ from~\eqref{eqn:rho_operator1},~\eqref{eqn:rho_operator2} in the constructed irreducible  bases. In the latter, we show that these matrix elements are the eigenvalues of the respective operators, and the constructed irreducible operator basis is their eigenbasis. This case can also be treated numerically using a brute-force approach, which computes the spectrum of the mentioned operators in a computational basis. The results are summarized in Table~\ref{arry:eig_of_rho} and will be used later for comparison with the analytical approach.

\begin{table}[ht]
    \centering
    \begin{tabular}{|c|c|c|}    
        \hline
        \textbf{1-arc} & \textbf{2-arcs} & \textbf{3-arcs} \\ \hline
        $0.1111 \, (\times 108)$ & \textcolor{cyan}{$0.1667 \, (\times 32)$} & $1.0000 \, (\times 1)$ \\ \hline
        $0.3333 \, (\times 133)$ & \textcolor{Green}{$0.2303 \, (\times 32)$} & $4.0000 \, (\times 4)$ \\ \hline
        $0.4444 \, (\times 108)$ & \textcolor{BrickRed}{$0.3333 \, (\times 1)$}    & $10.0000 \, (\times 1)$ \\ \hline
        $0.6667 \, (\times 104)$ & \textcolor{Green}{$0.6030 \, (\times 32)$} & \\ \hline
        $0.7778 \, (\times 27)$  & \textcolor{cyan}{$0.8333 \, (\times 32)$}  & \\ \hline
        $1.0000 \, (\times 36)$  & \textcolor{BrickRed}{$1.3333 \, (\times 4)$}    & \\ \hline
        $1.3333 \, (\times 8)$   & \textcolor{cyan}{$1.6667 \, (\times 8)$}   & \\ \hline
        $1.6667 \, (\times 1)$   & \textcolor{BrickRed}{$3.3333 \, (\times 1)$}    & \\ \hline
    \end{tabular}
      \caption{Eigenvalues of $\rho(p)$, for $p=1,2,3$ categorized by number of arcs. In this section, we focus on the analytics for 3-arcs $(p=3)$ and for 2-arcs $(p=2)$. For the completeness, we also include the 1-arc case corresponding to $p=1$. In the case of $p=2$, each color represents the different irreducible matrix units taken to calculate the overlap with the operator $\rho(2)$. Detailed explanation in the main text.}
      \label{arry:eig_of_rho}
\end{table}

Eigenvalues for the operator $\rho(3)$ are known from the previous paper~\cite{StudzinskiIEEE22}, where the maximal ideal $\mathcal{M}^{(3)}$ has been considered in the details. Indeed, from Theorem~\ref{thm:spec_of_rho_Gp}, equation~\eqref{eq:3} we see that the operator $\rho(3)$ is not supported on the ideal $\mathcal{M}^{(2)}$ and only operators from Section~\ref{sec:ideal_m} can give non-trivial contribution. 
Due to this, we are interested in the middle column in Table~\ref{arry:eig_of_rho} as it represents the eigenvalues of $\rho(2)$. The full block structure of the operator $\rho(2)$ in the irreducible basis is depicted in Figure~\ref{fig:structure_diagram}. In the rest of the section, we focus on proving this block structure analytically. 

For this particular case all possible Young frames are $\mu=\{(1,1,1),(2,1),(3)\}$ and $\alpha=\mu-\Box = \{(1,1),(2)\}$. To compress the notation, we will use the following abbreviation $(1,1,1)=(1^3), (1,1)=(1^2)$. 

\begin{figure}[ht]
\centering\includegraphics[width=1\textwidth]{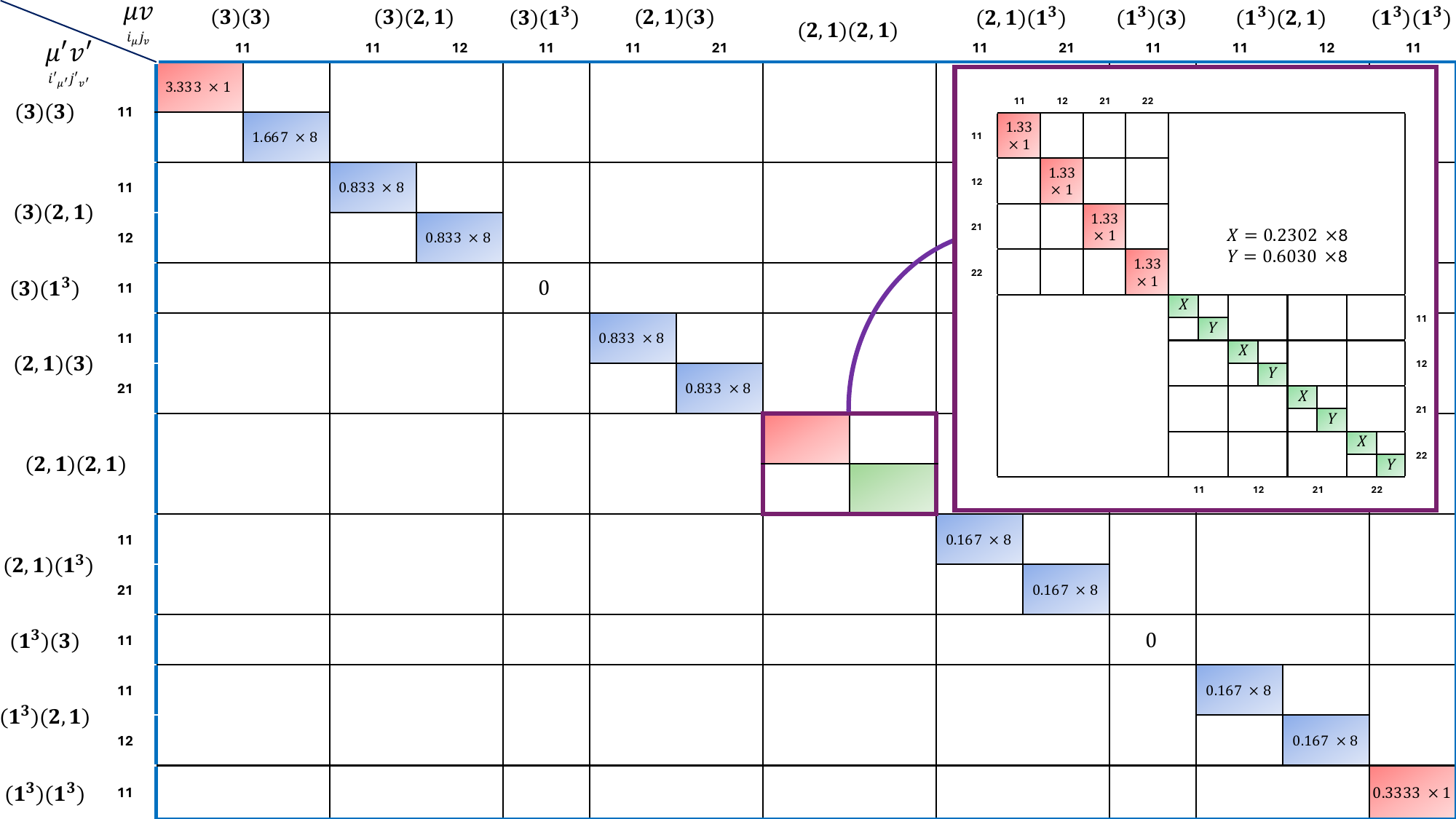}
    \caption{Graphic representation of evaluated matrix elements of the operator $\rho(2)$  given by \eqref{eqn:rho_operator2} for a case of $p=3$ with local dimension $d=3$, which turns out to be eigenvalues from Table \ref{arry:eig_of_rho}, with multiplicities indicates by $\times \#$ (i.e. in first row and column we have $3.333 \times 1$ which reads that the eigenvalue is 3.333 with multiplicity $1$). Operator $\rho(2)$ is supported on both ideals $\mathcal{M}^{(3)}$ and $\mathcal{M}^{(2)}$ which we colored columns corresponding to red and blue respectively. Additionally, we distinguished by the color green block which is orthogonalized between interior indices $\alpha = \mu - \Box$ and live on ideal $\mathcal{M}^{(2)}$. Each block is indexed by a pair of upper indices $\mu\nu$ and $\mu'\nu'$ and lower indices indicating the position in given irrep pair $i_\mu j_\nu$ and $i'_{\mu'}j'_{\nu'}$ where in our case $\mu,\nu,\mu',\nu' = \{(3),(2,1),(1^3)\}$. Every empty element we treat as a zero element.}
    \label{fig:structure_diagram}
\end{figure}

Respective values of multiplicities $m_{(\cdot)}$ and dimensions $d_{(\cdot)}$ of particular irreps $(\cdot)$ in the Schur-Weyl duality are the following:
\begin{eqnarray}\label{array:values_of_m_and_d}
\begin{array}{ccc}
 m_{(1^3)}=1 & m_{(2,1)}=8 & m_{(3)} = 10\\
d_{(1^3)} = 1& d_{(2,1)} = 2 & d_{(3)} = 1\\
m_{(1^2)}=3 & &m_{(2)} = 6  \\
d_{(1^2)}=1& & d_{(2)} = 1 
\end{array}
\end{eqnarray}

Due to Theorem~\ref{thm:spec_of_rho_Gp} we see that the operator $\rho(2)$ is non-trivially supported on both ideals $\mathcal{M}^{(3)},\mathcal{M}^{(2)}$.
For the irreducible matrix units given by Theorem \ref{thm:basis_for_max_ideal} for the highest ideal $\mathcal{M}^{(3)}$, we have three operators corresponding to each possible Young frames $\mu = \{(3),(2,1),(1^3)\}$. By Theorem \ref{thm:spec_of_rho_Gp} for $\mu=\{(3),(1^3)\}$ basis operators are 1-dimensional, so the eigenvalues of $\rho(2)$ are
\begin{align}
    \tr(\rho(2)\cdot \mathcal{G}^{(1^3) \ (1^3)}_{11 \ \  11}(3)) =\frac{1}{3}\cdot \frac{1}{1} \approx \textcolor{BrickRed}{0.3333}, \qquad \tr(\rho(2)\cdot \mathcal{G}^{(3) \ (3)}_{11 \ \ 11}(3)) = \frac{1}{3}\cdot \frac{10}{1} \approx\textcolor{BrickRed}{3.3333}.
\end{align}
For the case when $\mu=(2,1)$ the situation is more complicated, since in~\eqref{eq:2} we have $1\leq i_{(2,1)},j_{(2,1)}\leq 2$, so the corresponding matrix representation $\rho^{(2,1)(2,1)}(2)$ is 4-dimensional:
\begin{align}
\begin{pmatrix}
        \tr(\rho(2)\mathcal{G}^{(2,1) \ (2,1)}_{11 \ \  \ 11}(3)) & \cdot & \cdot & \cdot\\
        \cdot & \tr(\rho(2) \mathcal{G}^{(2,1) \ (2,1)}_{12 \ \ \ 12}(3)) & \cdot & \cdot \\
        \cdot & \cdot & \tr(\rho(2) \mathcal{G}^{(2,1) \ (2,1)}_{21 \ \ \ 21}(3)) & \cdot \\
        \cdot & \cdot & \cdot & \tr(\rho(2) \mathcal{G}^{(2,1) \ (2,1)}_{22 \ \ \ 22}(3))
    \end{pmatrix},
\end{align}
where dots denote zeros and non-trivial trace values we have only on the main diagonal. All these values, due to expression~\eqref{eq:2} from Theorem~\ref{thm:spec_of_rho_Gp} are equal to $\frac{1}{3}\cdot \frac{8}{2}\approx \textcolor{BrickRed}{1.3333}$. What is more we see that the operators $\{\mathcal{G}^{(2,1) \ (2,1)}_{11 \ \ \ 11}(3),\mathcal{G}^{(2,1) \ (2,1)}_{12 \ \ \ 12}(3),\mathcal{G}^{(2,1) \ (2,1)}_{21 \ \ \ 21}(3), \mathcal{G}^{(2,1) \ (2,1)}_{22 \ \ \ 22}(3)\}$ are eigen-operators for the operator $\rho(2)$.

Next, we aim to find eigenvalues of the operator $\rho(2)$ in the ideal $\mathcal{M}^{(2)}$. To do so, let us first observe that expression~\eqref{eqn:4} from Theorem~\ref{thm:spec_of_rho_Gp} implies that non-zero values of the trace can only be obtained for  $(\mu\neq \nu \wedge \alpha'=\alpha)$ or $(\mu= \nu \wedge (\alpha'\neq \alpha \lor \alpha'=\alpha))$. Let us concentrate on the first case. 
In fact, in this case, since we do not have a problem with orthogonality in the interior indices $(\alpha,\alpha')$ we can put unitary matrices $(U^{\mu\nu}_{\alpha,\beta})=\id$ from \eqref{eqn:diag_b_coeff0}.
The set of operators satisfying these conditions are the following:
\begin{align}
&\G^{(1^3)(2,1) \ (1^3)(2,1)}_{11 \ \ \  \ \ \ \ \ \ 11}(2,(1^2),(1^2)), \quad \G^{(1^3)(2,1) \ (1^3)(2,1)}_{12 \ \ \  \ \ \ \ \ \ 12}(2,(1^2),(1^2)),\label{eq: basis_contribution1}\\
&\G^{(2,1)(1^3)\ (2,1)(1^3)}_{11 \ \ \  \ \ \ \ \ \ 11}(2,(1^2),(1^2)), \quad \G^{(2,1)(1^3)\ (2,1)(1^3)}_{21 \ \ \  \ \ \ \ \ \ 21}(2,(1^2),(1^2)),\\
&\G^{(3)(2,1) \ (3)(2,1)}_{11 \ \ \  \ \ \ \ \ \ 11}(2,(2),(2)), \quad  \G^{(3)(2,1) \ (3)(2,1)}_{12 \ \ \  \ \ \ \ \ \ 12}(2,(2),(2)),\\
&\G^{(2,1)(3) \ (2,1)(3)}_{11 \ \ \  \ \ \ \ \ \ 11}(2,(2),(2)), \quad \G^{(2,1)(3) \ (2,1)(3)}_{21 \ \ \  \ \ \ \ \ \ 21}(2,(2),(2)).\label{eq: basis_contribution2}
\end{align}
 Then the respective overlaps of the operator $\rho(2)$ with basis operators from~\eqref{eq: basis_contribution1}-~\eqref{eq: basis_contribution2} by Theorem \ref{thm:spec_of_rho_Gp}, are the following:
\begin{align}
   \tr(\rho(2) \cdot \G^{(2,1)(1^3) \ (2,1)(1^3)}_{11 \ \ \  \ \ \ \ \ \ 11}(2,(1^2),(1^2)))
    &= \tr(\rho(2) \cdot \G^{(2,1)(1^3) \ (2,1)(1^3)}_{2 1 \ \ \  \ \ \ \ \ \ 2 1}(2,(1^2),(1^2))) 
    = \frac{4}{3} \approx \textcolor{cyan}{1.3333},\\
    \tr(\rho(2) \cdot \G^{(1^3)(2,1) \ (1^3)(2,1)}_{1 1 \ \ \  \ \ \ \ \ \ 1 1}(2,(1^2),(1^2)))&= \tr(\rho(2) \cdot \G^{(1^3)(2,1) \ (1^3)(2,1)}_{1 2 \ \ \  \ \ \ \ \ \ 1 2}(2,(1^2),(1^2)))
    = \frac{4}{3} \approx \textcolor{cyan}{1.3333},\\
    \tr(\rho(2) \cdot \G^{(3)(2,1)\ (3)(2,1)}_{1 1 \ \ \  \ \ \ \ \ \ 1 1}(2,(2),(2))) 
    &= \tr(\rho(2) \cdot \G^{(3)(2,1) \ (3)(2,1)}_{1 2 \ \ \  \ \ \ \ \ \ 1 2}(2,(2),(2)))
    = \frac{20}{3} \approx \textcolor{cyan}{6.6667}, \\
     \tr(\rho(2) \cdot \G^{(2,1)(3) \ (2,1)(3)}_{1 1 \ \ \  \ \ \ \ \ \ 1 1}(2,(2),(2))) 
    &= \tr(\rho(2) \cdot \G^{(2,1)(3) \ (2,1)(3)}_{2 1 \ \ \  \ \ \ \ \ \ 2 1}(2,(2),(2)))
    =\frac{20}{3} \approx \textcolor{cyan}{6.6667}.
\end{align}
Taking into account multiplicities $\tr(\G^{\mu\nu \ \ \mu\nu}_{i_\mu j_\nu \ i_\mu j_\nu}(p-1,\alpha,\alpha))$ of the above operators which are in every case above equal to 8, we obtain eigenvalues equal to \textcolor{cyan}{$1.3333/8 \approx  0.1667$} and \textcolor{cyan}{$6.6667/8 \approx 0.8333$}.  The mentioned multiplicities can be computed using the trace rules from~\eqref{eqn:def_of_op_H} by using Lemma~\ref{L:23} and Lemma~\ref{lemma:trace_of_f_p-1} from the Appendix.  
Again we see that normalized operators are eigenoperators for the operator $\rho(2)$ with respective eigenvalues.

Now we consider the third case, when $(\mu= \nu \wedge (\alpha'\neq \alpha \lor \alpha'=\alpha))$. We have the following possibilities:
\begin{align}
&((3)=(3) \wedge (2)=(2))\quad \wedge \quad ((1^3)=(1^3) \wedge (1^2)=(1^2)), \label{eq:line1}\\
&((2,1)=(2,1),(2)=(2))\quad \wedge \quad ((2,1)=(2,1), (1^2)=(1^2)), \label{eq:line2}\\
&((2,1)=(2,1),(2)\neq(1^2))\quad \wedge \quad ((2,1)=(2,1), (1^2)\neq(2)). \label{eq:line3}
\end{align}
For the line~\eqref{eq:line1} we deal with operators which are $\G^{(3)(3) \ (3)(3)}_{11 \ \ \ \ \ 11}(2,(2),(2))$, $ \G^{(1^3)(1^3) \ (1^3)(1^3)}_{11 \ \ \ \ \ \ \ 11}(2,(1^2),(1^2))$. First of all, from~\eqref{eqn:def_of_op_H} we can verify straightforwardly that the latter operator is a zero operator and does not contribute. For the operator $\G^{(3)(3) \ (3)(3)}_{11 \ \ \ \ \ 11}(2,(2),(2))$ the corresponding 1-dimensional space appears 8 times. Thus, we have:
\begin{align}
       \frac{1}{8}\tr(\rho(2) \cdot \G^{(3)(3) \ (3)(3)}_{11 \ \ \ \ \ 11}(2,(2),(2))) = \frac{5}{3} \approx \textcolor{cyan}{1.6667}.
\end{align}
For operators associated with~\eqref{eq:line2},~\eqref{eq:line3} situation is more complicated. Here, due to Theorem~\ref{thm:spec_of_rho_Gp} we have only four operators 
\begin{align}
   &\G^{(2,1)(2,1)\ (2,1)(2,1)}_{11 \ \ \  \ \ \ \ \ \ 11}(2,(1^2),(1^2)), \ \G^{(2,1)(2,1)\ (2,1)(2,1)}_{12 \ \ \  \ \ \ \ \ \ 12}(2,(1^2),(2)),\\
   &\G^{(2,1)(2,1) \ (2,1)(2,1)}_{21 \ \ \  \ \ \ \ \ \ 21}(2,(2),(1^2)), \ \G^{(2,1)(2,1) \ (2,1)(2,1)}_{22 \ \ \  \ \ \ \ \ \ 22}(2,(2),(2))
\end{align}
giving non-zero contribution to the trace with the operator $\rho(2)$, but taking into account~\eqref{eq:comp_H} in Theorem~\ref{thm:lower_ideal_operator_G} they are not orthogonal. For this reason, we cannot use them for obtaining eigenvalues of $\rho(2)$.  In this case, we must first apply Theorem~\ref{thm:irrepsMp1} to obtain orthogonality between the interior indices $\alpha =\mu -\Box$. Let us start from constructing matrix $B^{(2,1)(2,1)}=(b^{(2,1)(2,1)}(\alpha,\alpha'))$, where $b^{(2,1)(2,1)}(\alpha,\beta)$ can be evaluated from~\eqref{eqn:coefficient_b}. In this case, the matrix $B^{(2,1)(2,1)}$ has the following form
\begin{align}\label{eqn:b_matrix}
   B^{(2,1)(2,1)}(\alpha,\alpha')= \Bigg(\frac{1}{d(d^2-1)}\Bigl(d\frac{m_\mu m_\nu}{m_\alpha}\delta_{\alpha,\alpha'} - m_\mu\Bigr) \Bigg) = \left(\begin{array}{cc}
    1 & -\frac{1}{3} \\
    -\frac{1}{3} & \frac{7}{3}
\end{array}\right).
\end{align}
The unitary matrix $U^{(2,1)(2,1)}_{\alpha \alpha'}$ from~\eqref{eqn:diag_b_coeff0} which diagonalizes the matrix~\eqref{eqn:b_matrix} can be find numerically, and has a form
\begin{align}
    U^{(2,1) (2,1)}_{\alpha \beta}=\left(\begin{array}{cc}
    -\sqrt{\frac{1}{10}(5+2\sqrt{5})} & -\sqrt{\frac{1}{10}(5-2\sqrt{5})} \\
   -\sqrt{\frac{1}{10}(5-2\sqrt{5})} & \sqrt{\frac{1}{10}(5+2\sqrt{5})}
\end{array}\right).
\end{align}
and produces the following diagonal matrix $B_{diag}^{(2,1) (2,1)}(\beta)$
\begin{align}\label{eqn:b_matrix_diag}
    B_{diag}^{(2,1) (2,1)}(\beta)= \left(\begin{array}{cc}
    \frac{5-\sqrt{5}}{3} & 0 \\
    0 & \frac{5+\sqrt{5}}{3}
\end{array}\right),
\end{align}
Having the unitary matrix~\eqref{thm:irrepsMp1}, we can construct now orthonormal operator basis $\G^{\mu\nu \ \ \mu'\nu'}_{i_\mu j_\nu \ i'_{\mu'}j'_{\nu'}}(p-1,\beta,\beta')$ from Theorem~\ref{thm:spec_of_rho_Gp}:
\begin{align}\label{eqn:def_of_gpis1}
&\G^{(2,1)(2,1) \ (2,1)(2,1)}_{11 \ \ \  \ \ \ \ \ \ 11}(2,(1^2),(1^2)), \ \G^{(2,1)(2,1) \ (2,1)(2,1)}_{11 \ \ \  \ \ \ \ \ \ 11}(2,(1^2),(2)),\\
&\G^{(2,1)(2,1) \ (2,1)(2,1)}_{11 \ \ \  \ \ \ \ \ \ 11}(2,(2),(1^2)), \  \G^{(2,1)(2,1) \ (2,1)(2,1)}_{11 \ \ \  \ \ \ \ \ \ 11}(2,(2),(2)),\\
&\G^{(2,1)(2,1) \ (2,1)(2,1)}_{12 \ \ \  \ \ \ \ \ \ 12}(2,(1^2),(1^2)), \ \G^{(2,1)(2,1) \ (2,1)(2,1)}_{12 \ \ \  \ \ \ \ \ \ 12}(2,(1^2),(2)),\\
&\G^{(2,1)(2,1) \ (2,1)(2,1)}_{12 \ \ \  \ \ \ \ \ \ 12}(2,(2),(1^2)), \ \G^{(2,1)(2,1) \ (2,1)(2,1)}_{12 \ \ \  \ \ \ \ \ \ 12}(2,(2),(2)),\\
&\G^{(2,1)(2,1) \ (2,1)(2,1)}_{21 \ \ \  \ \ \ \ \ \ 21}(2,(1^2),(1^2)), \  \G^{(2,1)(2,1) \ (2,1)(2,1)}_{21 \ \ \  \ \ \ \ \ \ 21}(2,(1^2),(2)),\\
&\G^{(2,1)(2,1) \ (2,1)(2,1)}_{21 \ \ \  \ \ \ \ \ \ 21}(2,(2),(1^2)), \ \G^{(2,1)(2,1) \ (2,1)(2,1)}_{21 \ \ \  \ \ \ \ \ \ 21}(2,(2),(2)),\\
&\G^{(2,1)(2,1) \ (2,1)(2,1)}_{22 \ \ \  \ \ \ \ \ \ 22}(2,(1^2),(1^2)), \ \G^{(2,1)(2,1) \ (2,1)(2,1)}_{22 \ \ \  \ \ \ \ \ \ 22}(2,(1^2),(2)),\\
&\G^{(2,1)(2,1) \ (2,1)(2,1)}_{22 \ \ \  \ \ \ \ \ \ 22}(2,(2),(1^2)), \ \G^{(2,1)(2,1) \ (2,1)(2,1)}_{22 \ \ \  \ \ \ \ \ \ 22}(2,(2),(2)).
\end{align}
All above operators have multiplicities equal to 8, so the corresponding overlaps with the operator $\rho(2)$ are the following:
\begin{align}
    &\frac{1}{8}\tr(\rho(2)\cdot\G^{(2,1)(2,1) \ (2,1)(2,1)}_{i_{(2,1)} j_{(2,1)} \ i_{(2,1)} j_{(2,1)} }(2,(1^2),(1^2))) = \frac{1}{12}\Big(5-\sqrt{5}\Big) \approx \textcolor{Green}{0.2303},\\
    &\frac{1}{8}\tr(\rho(2)\cdot \G^{(2,1)(2,1) \ (2,1)(2,1)}_{i_{(2,1)} j_{(2,1)} \ i_{(2,1)} j_{(2,1)} }(2,(1^2),(2))) = 0, \\
    &\frac{1}{8}\tr(\rho(2)\cdot \G^{(2,1)(2,1) \ (2,1)(2,1)}_{i_{(2,1)} j_{(2,1)} \ i_{(2,1)} j_{(2,1)} }(2,(2),(1^2)))= 0,\\
    &\frac{1}{8}\tr(\rho(2)\cdot \G^{(2,1)(2,1) \ (2,1)(2,1)}_{i_{(2,1)} j_{(2,1)} \ i_{(2,1)} j_{(2,1)} }(2,(2),(2))) = \frac{1}{12}\Big(5+\sqrt{5}\Big) \approx \textcolor{Green}{0.6030},
\end{align}
for $i_{(2,1)},j_{(2,1)}\in\{1,2\}$. Thus we see that operators $\G^{(2,1)(2,1) \ (2,1)(2,1)}_{i_{(2,1)} j_{(2,1)} \ i_{(2,1)} j_{(2,1)} }(2,\beta,\beta))$, for $\beta \in \{(1^2),(2)\}$ are indeed eigenoperators of the operator $\rho(2)$ with 32 eigenvalues equal to 0.2303 and 32 eigenvalues equal to 0.6030. These are the last non-zero eigenvalues of the operator $\rho(2)$ from Table~\ref{arry:eig_of_rho}.

\section{Conclusions}
In this paper, we develop a mathematical toolkit for the algebra of partially transposed permutation operators, which provides a matrix representation of the walled Brauer algebra. The main contribution of this work is the construction of irreducible matrix units for the ideal $\mathcal{M}^{(p-1)}$ that are group-adapted to the subalgebra $\mathbb{C}[\mathfrak{S}_p] \times \mathbb{C}[\mathfrak{S}_p]$, in contrast to previous constructions relying on the Gelfand--Tsetlin approach. We present two complementary methods: the first is based on the representation theory of the symmetric group, while the second relies on the representation theory of the unitary group.

The results based on the symmetric-group approach are contained in Theorem~\ref{thm:lower_ideal_operator_G} and Theorem~\ref{thm:irrepsMp1}. In particular, in Lemma~\ref{lemma:operator_H_generates_ideal_mp-1} we show how to express the partially transposed permutation operator $V^{(p-1)}$ in terms of the constructed basis. This result allows one to generate the entire ideal using only the established irreducible matrix units. We further discuss several structural properties of the obtained basis and derive a number of trace and twirl rules that are useful in practical calculations for quantum information tasks with symmetries.

As an application of the developed methods, we study matrix elements of the operators $\rho(p)$ and $\rho(p-1)$ arising in variants of multi-port-based teleportation protocols. These operators correspond to averaged (twirled) versions of $V^{(p)}$ and $V^{(p-1)}$ with respect to the cross action of $\mathbb{C}[\mathfrak{S}_p]\times\mathbb{C}[\mathfrak{S}_p]$, see Theorem~\ref{thm:spec_of_rho_Gp}. The resulting fully analytical expressions are given in terms of dimensions and multiplicities of irreducible representations appearing in the Schur--Weyl duality and can, in principle, be efficiently evaluated using dedicated software. This allows us to provide an explicit example illustrating how the formalism works in practice. For $p=3$ and local dimension $d=3$, we analytically evaluate the matrix elements of the operators $\rho(3)$ and $\rho(2)$. In this case, the investigated operators are block-diagonal in the constructed irreducible basis, and the resulting matrix elements correspond to their eigenvalues, see Section~\ref{sec:example}. This demonstrates that the developed tools enable a fully analytical diagonalization of an important class of operators.

The results based on the representation theory of the unitary group rely on the decomposition of tensor products into irreducible representations and on the corresponding Clebsch--Gordan coefficients. This framework provides a constructive method for obtaining matrix elements of operators from the algebra $\mathcal{A}_{p,q}^d$ through tensor-network contractions built from unitary-group intertwiners. However, this approach requires explicit knowledge of Littlewood--Richardson multiplicities and therefore becomes computationally demanding for larger systems. For this reason, the unitary-group-based construction is best suited for problems involving a moderate number of particles and should be viewed as complementary to the symmetric-group-based approach developed in the previous sections. The detailed discussion is presented in Section~\ref{sec:tn_rep_mat_units}, illustrated with examples in Appendix~\ref{app:tensor_network_example_gens} (group-adapted generators for $\mathcal{A}_{3,3}^3$), and Appendix~\ref{app:tensor_network_example_rhos} (Decomposition of $\rho(k)$ into all group-adapted irreps).

The methods presented in this paper can be efficiently applied to problems with symmetries generated by the underlying walled Brauer algebra. At the same time, several directions remain open for further development of the formalism.

The first problem is to extend the symmetric-group-based construction to the case $p\neq q$. However, in view of results obtained in the context of multi-port-based teleportation protocols in~\cite{StudzinskiIEEE22}, this extension appears to be mainly technical, and we leave it for future work.

The second problem concerns the construction of irreducible matrix units without the necessity of identifying zero eigenvalues of the matrix $B^{\mu\mu}$ and applying the algorithmic procedure described in Appendix~\ref{app:C}. In particular, it would be desirable to develop a direct method for identifying linear dependencies among the operators appearing in Theorem~\ref{thm:irrepsMp1} and for determining the resulting dimensions of irreducible representations.

The third natural direction is the investigation of lower ideals $\mathcal{M}^{(p-2)}$ and, more generally, $\mathcal{M}^{(p-q)}$ for $2\leq q\leq p$, together with the construction of irreducible bases spanning these ideals. While the present formalism can in principle be generalized to such cases, the resulting constructions become technically involved. A systematic development of irreducible representation theory for the walled Brauer algebra that leads to a simplification of both notation and methodology remains an important open problem. These directions will be addressed in future work.

From a broader perspective, the presented constructions contribute to the systematic understanding of the walled Brauer algebra through explicit operator bases adapted to underlying symmetries. By combining symmetric-group techniques with complementary unitary-group methods, the framework developed here provides tools that can be used in a variety of problems involving mixed tensor representations and partial transposition. We expect that these methods will be useful in further studies of diagram algebras appearing in quantum information theory and related areas.

\section*{Acknowledgments}
T.M. is supported through grant Sonata 16, UMO-2020/39/D/ST2/01234 from the Polish National Science Centre. M.S. and M.H. acknowledge support by the IRA Programme, project no. FENG.02.01-IP.05-0006/23, financed by the
FENG program 2021-2027, Priority FENG.02, Measure FENG.02.01., with the support of the FNP. 
D.G.\ acknowledges support by NWO grant NGF.1623.23.025 (“Qudits in theory and experiment”) and NWO Vidi grant (Project No.\ VI.Vidi.192.109).

\appendix
\section{Additional Lemmas}
\begin{lemma}
\label{L:23}
    For the operators $F_{i_\mu j_\mu \ i'_{\nu} j'_{\nu}}^{\mu \ \ \ \nu}(p)$ given through Definition \eqref{eqn:F_operator} the following holds
    \begin{align}
       \tr(F_{i_\mu j_\mu \ i'_{\nu} j'_{\nu}}^{\mu \ \ \ \nu}(p)) = m_\mu \delta^{\mu \nu}\delta_{i'_\nu i_\mu}\delta_{j'_\mu j_\mu},
    \end{align}
    where the number $m_\mu$ denotes the multiplicity of irrep $\mu \vdash p$ in the Schur-Weyl duality.
    \end{lemma}
    \begin{proof}
        The proof is based on straightforward calculations. Using Definition \eqref{eqn:F_operator} we have
        \begin{align}
            \tr(F_{i_\mu j_\mu \ i'_{\nu} j'_{\nu}}^{\mu \ \ \ \nu}(p))
            &= \tr( (E^{\mu}_{i_\mu j_\mu} \otimes \id )V^{(p)} (E^{\nu}_{j'_\nu i'_\nu}\otimes \id) )\\
            &= \tr( (E^{\nu}_{j'_\nu i'_\nu} E^{\mu}_{i_\mu j_\mu}\otimes \id ) V^{(p)} )\\
            &= \tr(E^\mu_{j'_\mu j_\mu} ) \delta^{\mu \nu}\delta_{i'_\nu i_\mu}\\
            &= m_\mu  \delta^{\mu \nu}\delta_{i'_\nu i_\mu}\delta_{j'_\mu j_\mu},
        \end{align}
        where in the second equality we use cyclicity of the trace and in the third one the trace rule from~\eqref{eq:def_E}.
    \end{proof}

\begin{lemma}\label{lemma:trace_of_f_p-1}
    For the operators $F_{i_\mu j_\nu \ i'_{\mu'} j'_{\nu'}}^{\mu\nu \ \ \mu'\nu'}(p-1)$ given through Definition \eqref{eqn:F_poperator} the following holds
    \begin{align}
        \forall_{\alpha \in \mu \wedge \alpha \in \nu} \quad \tr(F_{i_\mu j_\nu \ i'_{\mu'} j'_{\nu'}}^{\mu\nu \ \ \mu'\nu'}(p-1)) = \frac{m_\mu m_\nu}{m_\alpha} \delta^{\alpha \alpha'}\delta^{\mu \mu'}\delta^{\nu \nu'}\delta_{i'_{\mu'} i_\mu}\delta_{j'_{\nu'} j_\nu}
    \end{align}
    where $m_\alpha,m_\mu,m_\nu$ are multiplicities of respective irreps $\alpha \vdash (p-1)$, and $\mu,\nu\vdash p$ in the Schur-Weyl duality.
    \end{lemma}
    \begin{proof}
        The proof is based on straightforward calculations by exploiting  Definition \eqref{eqn:F_poperator}:
        \begin{align}
            \tr(F_{i_\mu j_\nu \ i'_{\mu'} j'_{\nu'}}^{\mu\nu \ \ \mu'\nu'}(p-1)) &= \tr(\Big(E_{i_\mu \ 1_\alpha}^{\mu} \otimes E_{j_\nu \ 1_\alpha}^{\nu}\Big) V^{(p-1)}\Big( E_{1_{\alpha'} \ i'_{\mu'}}^{\ \ \ \ \mu'} \otimes E_{1_{\alpha'} \ j'_{\nu'}}^{\ \ \ \ \nu'} \Big))\\
            &= \tr( \Big(E_{1_{\alpha'} \ i'_{\mu'}}^{\ \ \ \ \mu'} E_{i_\mu \ 1_\alpha}^{\mu} \otimes E_{1_{\alpha'} \ j'_{\nu'}}^{\ \ \ \ \nu'} E_{j_\nu \ 1_\alpha}^{\nu}\Big) V^{(p-1)} )\\
            &= \tr( \Big(E_{1_{\alpha'} 1_\alpha}^{\mu}\otimes E^{\nu}_{1_{\alpha'} 1_\alpha}\Big)V^{(p-1)}) \delta^{\mu \mu'}\delta^{\nu \nu'}\delta_{i'_{\mu'} i_\mu}\delta_{j'_{\nu'} j_\nu}\label{eqn:trace_of_fp-1}
        \end{align}
        where in the second line we use cyclicity of the trace, and in the third line composition rule from~\eqref{eq:def_E}.
Since the operator $V^{(p-1)}$ acts trivially on the systems $1$ and $2p$ we write the trace from~\eqref{eqn:trace_of_fp-1} as
        \begin{align}
            \tr(\Big( E_{1_{\alpha'} 1_\alpha}^{\mu} \otimes E^{\nu}_{1_{\alpha'} 1_\alpha}\Big) V^{(p-1)}) &=  \tr( \left[\tr_1 \Big(E_{1_{\alpha'} 1_\alpha}^{\mu}\Big) \otimes \tr_{2p}\Big(E^{\nu}_{1_{\alpha'} 1_\alpha}\Big)\right] V^{(p-1)})\\
            &= \frac{m_\mu}{m_\alpha}\frac{m_\nu}{m_\alpha}\tr(\Big(E_{1_{\alpha'} 1_\alpha}^{\mu}\otimes E^{\nu}_{1_{\alpha'} 1_\alpha}\Big)V^{(p-1)})\\
             &= \frac{m_\mu}{m_\alpha}\frac{m_\nu}{m_\alpha}\tr( E_{1_{\alpha'} 1_\alpha}^{\mu} \Big(E^{\nu}_{1_{\alpha'} 1_\alpha}\Big)^T \tr_{p+1,\ldots,2p-1}\Big(V^{(p-1)}\Big))\\
            &= \frac{m_\mu}{m_\alpha}\frac{m_\nu}{m_\alpha}\tr( E^{\alpha}_{1_\alpha 1_\alpha} E^\alpha_{1_\alpha 1_\alpha} ) \delta^{\alpha \alpha'}\\
            &= \frac{m_\mu m_\nu}{m_\alpha} \delta^{\alpha \alpha'}.
        \end{align}
        In the first line, we apply Lemma~\ref{L3a}, in the second line, we use the 'ping-pong' trick~\eqref{eq:BellV'}. In the third line, we use the composition relation from~\eqref{eq:def_E}, while in the fourth one, we again apply the composition relation~\eqref{eq:def_E} together with the trace rule from the same expression.
    \end{proof}

\section{Further properties of the matrix \texorpdfstring{$B^{\mu\mu}$}{Lg}  and its connection to symmetric polynomials}
\label{App:PropB}

As we explained in Section~\ref{sec:ideal_m-1}, one of the important features of the matrix $B^{\mu\mu}(\alpha,\alpha')$ that must be understood is the vanishing of its determinant. In this particular case, all matrix elements of $B^{\mu\mu}=(b^{\mu\mu}(\alpha,\alpha'))$ are described by the following expression:
\begin{align}
\label{App:matB}
B^{\mu\mu}=(b^{\mu\mu}(\alpha,\alpha'))=\Bigg(\frac{m_\mu}{d(d^2-1)}\Bigl(d\frac{m_\mu}{m_\alpha}\delta^{\alpha,\alpha'} - 1\Bigr) \Bigg).
\end{align}
In this appendix, we present all pieces of information required to prove Proposition~\ref{Prop:28} with the conclusion in Corollary~\ref{Cor:29}. The main result for this section is contained in Theorem~\ref{thm:main27}.
We start from the classical statement for symmetric polynomials.
\begin{theorem}
\label{thm:fundament}
(Fundamental theorem on symmetric polynomials~\cite{Tignol}) A polynomial $%
W(x_{1},x_{2},\ldots,x_{k})$ is symmetric in $
x_{1},x_{2},\ldots,x_{k}$ if and only if $W(x_{1},x_{2},\ldots,x_{k})$ is a
polynomial is elementary symmetric polynomials $s_{1},s_{2},\ldots,s_{k}$, i.e. $W(x_{1},x_{2},\ldots,x_{k})=W^{\prime }(s_{1},s_{2},\ldots,s_{k})$, where 
\begin{align}
&s_{1}=x_{1}+x_{2}+\ldots +x_{k},\quad s_{2}=\sum_{i<j}x_{i}x_{j},\\
&s_{k-1}=\sum_{i_{1}<\cdots<i_{k-1}}x_{i_{1}}\cdots x_{i_{k-1}},\quad
s_{k}=x_{1}x_{2}\cdots x_{k}.
\end{align}
\end{theorem}

\begin{definition}
For $a,x_{i}\in \mathbb{R}$ let us define the following matrix $B=(b_{ij})$:
\begin{equation}
B(a,x_{1},x_{2},\ldots,x_{k})=(b_{ij})=\left( x_{i}\delta _{ij}-a\right) ,
\end{equation}%
which has an explicit form given as
\begin{align}
B(a,x_{1},x_{2},\ldots,x_{k})=\left( 
\begin{array}{cccc}
x_{1}-a & -a & \cdots & -a \\ 
-a & x_{2}-a & \cdots & -a \\ 
\vdots & \vdots & \ddots & \vdots \\ 
-a & -a & \cdots & x_{k}-a
\end{array}
\right). 
\end{align}
\end{definition}
We see that choosing $x_i \delta_{ij} \mapsto \frac{m_\mu}{m_\alpha}\delta^{\alpha \alpha'}$ and $a \mapsto1$ we reproduce, up to a global factor $\frac{m_\mu}{d(d^2-1)}$, the matrix $B^{\mu\mu}$ from~\eqref{App:matB}.
Having the above general definition of the matrix $B(a,x_{1},x_{2},\ldots,x_{k})$, we are in a position to formulate the first result for this section.

\begin{lemma}
\label{Lemma26}
Determinant of the matrix $B(a,x_{1},x_{2},\ldots,x_{k})$ is a symmetric
polynomial in variables $x_{1},x_{2},\ldots,x_{k}$, i.e. we have 
\begin{equation}
\det [B(a,x_{\sigma (1)},x_{\sigma (2)},\ldots,x_{\sigma (k)})]=\det
[B(a,x_{1},x_{2},\ldots,x_{k})]\quad \forall \sigma \in \s_k
\end{equation}
\end{lemma}

\begin{proof}
Let $A(\sigma
)=(\delta _{i\sigma_{(j)}})\in M(k,\mathbb{R}),$ for  $\sigma \in \s_k$, be a natural matrix representation of the group $\s_k$, then we have 
\begin{equation}
A[(pq)]B(a,x_{1},\ldots,x_{p},\ldots,x_{q},\ldots,x_{k})A[(pq)]=B(a,x_{1},\ldots,x_{q},\ldots,x_{p}\ldots,x_{k})
\end{equation}
for any transposition $(pq)\in \s_k$ this implies that 
\begin{equation}
\det B(a,x_{1},\ldots,x_{p},\ldots,x_{q},\ldots,x_{k})=\det
B(a,x_{1},\ldots,x_{q},\ldots,x_{p},\ldots,x_{k})\quad \forall p,q=1,\ldots,k
\end{equation}
and from this, it follows the statement of the lemma because each permutation 
$\sigma \in \s_k$ is a product of transpositions.
\end{proof}
Now we are ready to prove the main result for this appendix, being the main technical tool for Proposition~\ref{Prop:28} and Corollary~\ref{Cor:29}.

\begin{theorem}
\label{thm:main27}
For any $a,x_{i}\in \mathbb{R}$ we have 
\begin{equation}
\det
[B(a,x_{1},x_{2},\ldots,x_{k})]=s_{k}-as_{k-1}=x_{1}x_{2}\cdots x_{k}-a%
\sum_{i_{1}<\cdots<i_{k-1}}x_{i_{1}}\cdots x_{i_{k-1}}.
\end{equation}
\end{theorem}

\begin{proof}
First we consider simple case of matrix $B(a,x_{1},x_{2},\ldots,x_{k})$, when $
x_{1}=x_{2}=\cdots=x_{k}=x\in \mathbb{R}$. In this simple case, we have
\begin{equation}
B(a,x,x,\ldots,x)=x\mathbf{1}_{k}-aJ,
\end{equation}%
where $J=(j_{ab}):j_{ab}=1 \ \forall a,b=1,\ldots,k$. It is easy to check that it has only
two eigenvalues: $x-ak$ with eigenvector $v^{t}=(1,1,\ldots,1)$ and
multiplicity one and $x$ with eigenvectors $w=(w_{i}):\sum_{i}w_{i}=0$ with
multiplicity $k-1$. Therefore
\begin{equation}
\label{eq:ap9}
\det B(a,x,x,\ldots,x)=x^{k-1}(x-ak)=x^{k}-akx^{k-1}.
\end{equation}
on the other hand from Theorem~\ref{thm:fundament}  and Lemma~\ref{Lemma26} we deduce that 
\begin{equation}
\label{eq:ap10}
\det B(a,x_{1},\ldots,x_{p},\ldots,x_{q},\ldots,x_{k})=\sum_{i=1}^{k}q_{i}s_{i}+s_{0},\quad
s_{0}\in \mathbb{R},
\end{equation}
where $q_{i}\in \mathbb{R}$ are the coefficients that we need to determine. Taking into account
that 
\begin{equation}
s_{l}|_{x_{i}=x}=\binom{k}{l}x^{l}
\end{equation}
we get from eq.~\eqref{eq:ap9} and~\eqref{eq:ap10} we obtain the following expression
\begin{equation}
x^{k}-akx^{k-1}=\sum_{i=1}^{k}q_{i}\binom{k}{l}x^{l}+s_{0},\quad \forall
x\in \mathbb{R}
\end{equation}
and from this, we obtain
\begin{equation}
q_{k}=1,\quad q_{k-1}=-a,\quad q_{i}=0,\quad i<k-1,
\end{equation}
which yields the result.
\end{proof}

Applying this formula to our matrix $B(a=1,\frac{dm_{\mu }}{m_{\alpha
_{1}}},\frac{dm_{\mu }}{m_{\alpha _{2}}},\ldots,\frac{dm_{\mu }}{m_{\alpha _{k}}}
)$ we get

\begin{proposition}
\label{Prop:28}
\begin{equation}
\det [B(1,\frac{dm_{\mu }}{m_{\alpha _{1}}},\frac{dm_{\mu }}{
m_{\alpha _{2}}},\ldots,\frac{dm_{\mu }}{m_{\alpha _{k}}})]=d^{k-1}(m_{\mu
})^{2k-1}\left( \frac{dm_{\nu }}{\sqcap _{i}m_{\alpha _{i}}}
-\sum_{i_{1}<\cdots <i_{k-1}}\frac{1}{m_{\alpha _{i_{1}}}m_{\alpha
_{i_{2}}}\cdots m_{\alpha _{i_{k-1}}}}\right) .
\end{equation}
\end{proposition}

Since multiplication by a global factor $\frac{m_\mu}{d(d^2-1)}$ does not change the determinant of $B(a=1,\frac{dm_{\mu }}{m_{\alpha
_{1}}},\frac{dm_{\mu }}{m_{\alpha _{2}}},\ldots,\frac{dm_{\mu }}{m_{\alpha _{k}}}
)$, the above Proposition~\ref{Prop:28} implies the following:

\begin{corollary}
\label{Cor:29}
\begin{equation}
\det [B^{\mu\mu }(1,\frac{dm_{\mu }}{m_{\alpha _{1}}},\frac{dm_{\mu }}{%
m_{\alpha _{2}}},\ldots,\frac{dm_{\mu }}{m_{\alpha _{k}}})]=0
\end{equation}
if and only if
\begin{equation}
\label{eq:det0exp}
dm_{\mu }=m_{\alpha _{1}}+m_{\alpha _{2}}+\ldots+m_{\alpha _{k}}.
\end{equation}
\end{corollary}

Here we present explicit examples of how the matrix $B^{\mu\mu}$ defined in \eqref{App:matB} evolves with a particular number of particles. Let us choose the dimension $d=3$ and we will start from $p=3$ up to $p=6$. 

\begin{example}
Let $p=d=3$, then we have three possible diagrams
\begin{align}
    \ytableausetup
{boxsize=1em}
\ytableausetup
{aligntableaux=top}
\mu = \Biggl\{ \ydiagram{3}, \quad \ydiagram{2,1},\quad  \ydiagram{1,1,1} \Biggr\}.
\end{align}
We provided analysis in Section \ref{sec:example}, it is worth noting that matrix $B^{(2,1)(2,1)}$ given in \eqref{eqn:b_matrix} is not singular and one can evaluate that $\det(B^{(2,1)(2,1)})=2.223$. However for $\mu=(1^3)$ then we know from list \eqref{array:values_of_m_and_d} that $m_{(1^3)}=1$ and we have only one possible $\alpha,\alpha'=(1^3)-\Box=(1^2)$ for which $m_{(1^2)}=3$, then we calculate
\begin{align}
     B^{(1^3)(1^3)}=b^{(1^3)(1^3)}((1^2),(1^2))=\frac{1}{28}\Bigg(3\cdot \frac{1}{3}-1\Bigg)=0,
\end{align}
hence the $1\times 1$ matrix $B^{(1^3),(1^3)}$ is singular moreover condition \eqref{eq:det0exp} from Corollary \ref{Cor:29} is also satisfied.
\end{example}

\begin{example}
Let $p=4$ and $d=3$, then we have four possible Young diagrams 
\begin{align}
    \ytableausetup
{boxsize=1em}
\ytableausetup
{aligntableaux=top}
\mu = \Biggl\{ \ydiagram{4}, \quad \ydiagram{3,1},\quad  \ydiagram{2,2},\quad  \ydiagram{2,1,1} \Biggr\}. 
\end{align}
We start the analysis by choosing diagram $\mu=(2,1,1)$. For such choice of $\mu$ we have two possible $\alpha,\alpha'=\mu-\Box=\{(2,1),(1^3)\}$, one can evaluate all the dimensions and multiplicities of $\mu$ and $\alpha's$ using \eqref{eqn:equation_on_d} and $\eqref{eqn:equation_on_m}$ respectively and obtain that $m_{(2,1,1)}=3$, $m_{(2,1)}=8$ and $m_{(1^3)}=1$. The matrix $B^{(2,1,1),(2,1,1)}$ is $2\times 2$ of following form
\begin{align}
   B^{(2,1,1)(2,1,1)}= \left(
\begin{array}{cc}
 1 & -\frac{1}{8} \\
 -\frac{1}{8} & \frac{1}{64} \\
\end{array}
\right),\end{align}
the condition of Corollary \eqref{eq:det0exp} is met moreover the determinant of $ B^{(2,1,1),(2,1,1)}$ is zero. Now, take $\mu=(2,2)$, then we only have one way of removing a box, which is $\alpha=(2,1)$. One can find that $m_{(2,2)}=6$ and $m_{(2,1)}=8$, the matrix $B^{(2,2),(2,2)}$ is one dimensional and we observe that condition \eqref{eq:det0exp} is not met here and determinant or the value of $B$ is equal to $\frac{5}{16}$. The similar behavior repeats for the rest of possible diagrams $\mu$, meaning for $\mu=\{(3,1),(4)\}$ matrices $B^{\mu\mu}$ are non-singular. 
\end{example}
\begin{example}
Let us consider $p=5$ and $d=3$, for which we have five possible diagrams
\begin{align}
    \ytableausetup
{boxsize=1em}
\ytableausetup
{aligntableaux=top}
\mu = \Biggl\{ \ydiagram{5}, \quad \ydiagram{4,1},\quad  \ydiagram{3,2},\quad  \ydiagram{3,1,1},\quad  \ydiagram{2,2,1} \Biggr\}.
\end{align}
Similar to the previous examples, we start with diagram $\mu=(2,2,1)$. Then the possible $\alpha's$ are $\alpha,\alpha'=\{(2,2),(2,2,1)\}$ for which one reads their multiplicities $m_{(2,2)}=6$, $m_{(2,2,1)}=3$  and $m_{(2,1,1)}=3$. The condition \eqref{eq:det0exp} is satisfied, moreover the matrix $B^{(2,2,1)(2,2,1)}$ is two-dimensional of the following form
\begin{align}
    B^{(2,2,1)(2,2,1)}=\left(
\begin{array}{cc}
 \frac{1}{4} & -\frac{1}{8} \\
 -\frac{1}{8} & \frac{1}{16} \\
\end{array}
\right),
\end{align}
it is easy to see that determinant is zero meaning matrix $B^{(2,2,1)(2,2,1)}$ is singular. Now however, if we choose $\mu=(3,1,1)$ for whom $\alpha,\alpha'=\{(2,1,1),(3,1)\}$ and check the values of multiplicities which are $m_{(3,1,1)}=6$, $m_{(2,1,1)}=3$ and $m_{(3,1)}=15$ respectively. Then, $B^{(3,1,1)(3,1,1)}$ is $2\times 2$ matrix of the following form
\begin{align}
   B^{(3,1,1)(3,1,1)}= \left(
\begin{array}{cc}
 \frac{5}{4} & -\frac{1}{4} \\
 -\frac{1}{4} & \frac{1}{20} \\
\end{array}
\right),\end{align}
for which Corollary \ref{Cor:29} is met. If we take next diagram which is $\mu=(3,2)$ for whom $\alpha,\alpha'= \{(3,1),(2,2)\}$ and multiplicities are following $m_{(3,2)}=15$, $m_{(2,2)}=6$ and $m_{(3,1)}=15$ we observe that condition \eqref{eq:det0exp} is not met and
\begin{align}
    B^{(3,2)(3,2)}=\left(
\begin{array}{cc}
 \frac{65}{16} & -\frac{5}{8} \\
 -\frac{5}{8} & \frac{5}{4} \\
\end{array}
\right),
\end{align}
whose determinant is a rational number equal to $\frac{75}{16}$, meaning it is not singular. The story goes similarly for the rest of possible diagrams $\mu=\{(4,1),(5)\}$, matrices $B^{\mu\mu}$ for such choices of $\mu$ are non-singular.
\end{example}
\begin{example}
Let $p=6$ and $d=3$, for which we have seven possible diagrams
\begin{align}
    \ytableausetup
{boxsize=1em}
\ytableausetup
{aligntableaux=top}
\mu = \Biggl\{ \ydiagram{6}, \quad \ydiagram{5,1},\quad  \ydiagram{4,2},\quad  \ydiagram{4,1,1},\quad  \ydiagram{3,3},\quad  \ydiagram{3,2,1},\quad  \ydiagram{2,2,2} \Biggr\}.
\end{align}
Starting with diagram $\mu=(2,2,2)$ we have only one possible $\alpha,\alpha'=(2,2,1)$ for which $m_{(2,2,2)}=1$ and $m_{(2,2,1)}=1$ which allows as to immediately identify that coefficient $B^{(2,2,2)(2,2,2)}=b^{(2,2,2)(2,2,2)}((2,2,1),(2,2,1))=0$ and the condition \eqref{eq:det0exp} is met. For next diagram $\mu=(3,2,1)$ we have three possible ways to remove a box, so $\alpha,\alpha'=\{(3,2),(3,1,1),(2,2,1)\}$ for which we read the multiplicities $m_{(3,2,1)}=8$, $m_{(3,2)}=15$, $m_{(3,1,1)}=6$ and $m_{(2,2,1)}=3$ respectively. The condition \eqref{eq:det0exp} reads
\begin{align}
    3\cdot 8 = 3+6+15,
\end{align}
which is true and the $B^{(3,2,1),(3,2,1)}$ is a $3\times 3$ matrix of the form
\begin{align}
   B^{(3,2,1)(3,2,1)}= \left(
\begin{array}{ccc}
 \frac{1}{5} & -\frac{1}{3} & -\frac{1}{3} \\
 -\frac{1}{3} & \frac{7}{3} & -\frac{1}{3} \\
 -\frac{1}{3} & -\frac{1}{3} & 1 \\
\end{array}
\right),
\end{align}
one can evaluate the determinant which is equal to zero. In next step we take diagram $\mu=(3,3)$ for which $\alpha=(3,2)$ and multiplicities presents as follows $m_{(3,3)}=10$ and $m_{(3,2)}=15$. We see that Corollary \ref{Cor:29} is not met, moreover the value of $B^{(3,3)(3,3)}=b^{(3,3)(3,3)}((3,2),(3,2)) = \frac{5}{12}$. The next diagram is $\mu=(4,1,1)$ and $\alpha,\alpha'=\{(4,1),(3,1,1)\}$ with respective multiplicities $m_{(4,1,1)}=10$,  $m_{(4,1)}=24$ and $m_{(3,1,1)}=6$. Plugging these multiplicities into \eqref{eq:det0exp} we observe that condition holds and the matrix $B^{(4,1,1)(4,1,1)}$ has form
\begin{align}
   B^{(4,1,1)(4,1,1)}= \left(
\begin{array}{cc}
 \frac{5}{3} & -\frac{5}{12} \\
 -\frac{5}{12} & \frac{5}{48} \\
\end{array}
\right)
\end{align}
which one can check that the determinant is zero. One can check that the rest of diagrams $\mu=\{(4,2),(5,1),(6)\}$ has non-singular matrices $B^{\mu\mu}$. 
\end{example}

\section{Orthonormal basis construction in the case of singular matrix \texorpdfstring{$B^{\mu\mu}$}{Lg}}\label{app:C}
In this section, we present a general scheme for the orthonormal basis construction when the matrix $B^{\mu\mu}$ is singular, so when it has at least one zero eigenvalue. This happens when the condition~\eqref{eq:det0exp} from Corollary~\ref{Cor:29} holds. We present here the general considerations, and start from the following definition:

\begin{definition}
\label{App:def1}
Let $A=(a_{ij})\in M(m,\mathbb{C})$, then the algebra $X_{A}$ is defined as 
\begin{align}
X_{A}=\operatorname{span}_{\mathbb{C}}\{x_{ij}:i,j=1,\ldots,m\} 
\end{align}
where 
\begin{align}
x_{ij}x_{kl}=a_{jk}x_{il}. 
\end{align}
and we do not assume that the elements $\{x_{ij}\}$ are linearly
independent, thus the algebra $X_{A}$ is a complex finite-dimensional
algebra of the dimension at most $m^{2}$.
\end{definition}

The properties of the algebra $X_{A}$ depend on the properties of
the matrix $A$. We have

\begin{theorem}
Suppose that the matrix $A$ in the algebra $X_{A}$ is invertible then we
have two possibilities:
\begin{enumerate}[a)]
    \item $X_{A}=\{0\}$, i.e. the algebra $X_{A}$ a zero algebra,
    \item If $X_{A}\neq \{0\}$ then algebra $X_{A}$ is isomorphic to the matrix
algebra $M(m,\mathbb{C})$ and the elements $\{x_{ij}:i,j=1,\ldots,m\}$ are linearly independent, in particular $x_{ij}\neq 0:i,j=1,\ldots,m$, and form the basis of $X_{A}$. and in this case the unit of the algebra $X_{A}$ is of the form
\begin{align}
\mathbf{1}=\sum_{i,j=1,\ldots,m}(a_{ij}^{-1})x_{ij}, 
\end{align}
where $A^{-1}=(a_{ij}^{-1})$.
\end{enumerate}
\end{theorem}

\begin{proof}
\bigskip Suppose that $x_{ij}=0$ for some indices $i,j=1,\ldots,m$. Then from
multiplication law for the algebra $X_{A}$ we get
\begin{align}
\forall k,l=1,\ldots,m,\quad x_{ij}x_{kl}=a_{jk}x_{il}=0 
\end{align}
and because the matrix $A=(a_{ij})$ is invertible (so it has no zero columns
or zero rows) then we get 
\begin{align}
x_{ij}=0\Rightarrow \forall k,l=1,\ldots,m\quad x_{ik}=x_{li}=0 
\end{align}
and consequently $\forall k,l=1,\ldots,m\quad x_{kl}=0$. If $X_{A}\neq \{0\}$
then defining a new basis 
\begin{align}
y_{kj}:=\sum_{i=1,\ldots,m}(a_{ik}^{-1})x_{ij} 
\end{align}
we get 
\begin{align}
y_{ij}y_{kl}=\delta _{jk}y_{il}. 
\end{align}
\end{proof}

So we see that the invertibility of the matrix $A$ strongly
determines the properties of the vectors $\{x_{ij}:i,j=1,\ldots,m\}$ which span
the algebra $X_{A}$. We have also

\begin{theorem}
\label{thm:above}
Suppose the the vectors $\{x_{ij}:i,j=1,\ldots,m\}$ which span the algebra $
X_{A}$ are linearly independent and $\det (A)=0$ then there exist in the
algebra $X_{A}$ a nonzero properly nilpotent element and consequently the
algebra $X_{A}$ is not semisimple.
\end{theorem}

\begin{proof}
From the assumption $\det (A)=0$ it follows that there exists
a nonzero vector $u=(u_{1},u_{2},\ldots,u_{m})\in \mathbb{C}^{m}$, such that 
\begin{align}
Au=0\Leftrightarrow \sum_{k=1,\ldots,m}(a_{ik})u_{k}=0:i=1,\ldots,m. 
\end{align}
Consider now, for an arbitrary $l=1,\ldots,m$, an element $w_{l}=
\sum_{k=1,\ldots,m}u_{k}x_{kl}\in X_{A}$ which from the first assumption is
nonzero. From the multiplication law for the algebra $X_{A}$ we get
\begin{align}
x_{ij}w_{l}=\sum_{k=1,\ldots,m}u_{k}x_{ij}x_{kl}=
\sum_{k=1,\ldots,m}u_{k}a_{jk}x_{il}=0\quad \forall i,j=1,\ldots,m. 
\end{align}
which means that the nonzero elements $w_{l}$ are properly nilpotent and
therefore from the Definition~\ref{App:def1} the algebra $X_{A}$ is not semisimple.
\end{proof}

From the above Theorem~\ref{thm:above}, it follows:

\begin{corollary}
If the algebra $X_{A}$ is semisimple and $\det (A)=0,$ then the vectors $
\{x_{ij}:i,j=1,\ldots,m\}$ which span the algebra $X_{A}$ are linearly
dependent.
\end{corollary}

In this case, a natural question arises, how to reduce the set of
linearly dependent vectors $\{x_{ij}: i,j=1,\ldots,m\}$ to the set of linearly
independent ones. If the algebra $X_{A}$ is semisimple then we have the
following method of constructing the basis of $X_{A}$:

\begin{theorem}{\bf (Algorithm for orthonormal basis construction)}
Let the algebra $X_{A}$ be a semisimple algebra such that 
\begin{align}
X_{A}=\operatorname{span}_{\mathbb{C}}\{x_{ij}:i,j=1,\ldots,m\}, 
\end{align}
where 
\begin{align}
x_{ij}x_{kl}=a_{jk}x_{il}. 
\end{align}
The algorithm for orthogonal basis construction is as follows:
\begin{enumerate}
\item We solve diagonalisation problem for the matrix $A=(a_{ij})$, i.e.
\begin{align}
Z^{-1}AZ=diag(\lambda _{1}\neq 0,\ldots,\lambda _{p}\neq
0,0,\ldots,0)\Leftrightarrow \sum_{jk}z_{ij}^{-1}a_{jk}z_{kl}=\lambda _{i}\delta
_{il},\quad Z\in M(m,\mathbb{C}). 
\end{align}
\item If $p=m$, the matrix $A$ is invertible and we end in the generic (non-singular) case. The algorithm terminates here.
\item In the other case, we define new elements of the
algebra $X_{A}$ as
\begin{align}
y_{sr}:=\sum_{jk}z_{rj}^{-1}x_{ij}z_{is}. 
\end{align}
\item We divide the set of element $\{y_{sr}\}$ into two groups. To the first group we assign elements $\{y_{sr}\}$ with the following properties 
\begin{align}
y_{sr}=y_{rs}=0\qquad \forall s=1,\ldots,m,\quad r>p. 
\end{align}
We assign the remaining non-zero vectors $\{y_{ij}:i,j=1,\ldots,p\}$ to the second group. These elements form the basis of
the algebra $X_{A},$which we call a reduced basis, and they satisfy the
following multiplication rule 
\begin{align}
y_{ij}y_{kl}=\lambda _{j}\delta _{jl}y_{il},\quad i,j,k,l=1,\ldots,p. 
\end{align}
\item We rescale the basis vectors $\{y_{ij}:i,j=1,\ldots,p\}$ in the following way
\begin{align}
y_{ij}\mapsto f_{ij}=\frac{1}{\sqrt{\lambda _{i}\lambda _{j}}}y_{ij} 
\end{align}
\item The above rescaling gives a new basis of the algebra $X_{A}$, which satisfies the matrix
multiplication rule 
\begin{align}
f_{ij}f_{kl}=\delta _{jl}f_{il},\quad i,j,k,l=1,\ldots,p 
\end{align}
showing that the algebra $X_{A}$ is isomorphic with the matrix
algebra $M(p,\mathbb{C})$.
\end{enumerate}
\end{theorem} 

\newpage
\section{Example: generators of $\mathcal{A}^3_{3,3}$ in the $\mathbb{C}[\s_3] \times \mathbb{C}[\s_3]$--adapted basis}\label{app:tensor_network_example_gens}

In this section, we present explicit representation matrices of all $\mathcal{A}^3_{3,3}$ generators in all irreps $\Lambda$. First, the contraction generator irrep matrices are presented in Figure~\ref{fig:A333_contraction_all_irreps}, computed according to the methodology, described in eq. \eqref{eq:contr_gen_method}.

\begin{figure}[H]
  \centering
  \captionsetup{labelfont=bf}
  \small 

  \begin{subfigure}{0.31\textwidth}
    \centering
    \caption{$\Lambda = (0,0,0)$}\label{fig:tab1}
    \vspace{5pt}
    \adjustbox{max width=\linewidth}{$\left(
    \begin{array}{cccccc}
     \frac{1}{3} & \frac{2 \sqrt{2}}{3} & 0 & 0 & 0 & 0 \\
     \frac{2 \sqrt{2}}{3} & \frac{8}{3} & 0 & 0 & 0 & 0 \\
     0 & 0 & 0 & 0 & 0 & 0 \\
     0 & 0 & 0 & 0 & 0 & 0 \\
     0 & 0 & 0 & 0 & \frac{4}{3} & \frac{2 \sqrt{5}}{3} \\
     0 & 0 & 0 & 0 & \frac{2 \sqrt{5}}{3} & \frac{5}{3} \\
    \end{array}
    \right)$}
  \end{subfigure}

    \medskip
  \begin{subfigure}{0.95\textwidth}
    \centering
    \caption{$\Lambda = (1,0,-1)$}\label{fig:tab2}
    \vspace{5pt}
    \adjustbox{max width=\linewidth}{$\left(
    \begin{array}{ccccccccccccccccc}
     \frac{1}{3} & 0 & \frac{1}{3} & 0 & \frac{\sqrt{7}}{3} & 0 & 0 & 0 & 0 & 0 & 0 & 0 & 0 & 0 & 0 & 0 & 0 \\
     0 & \frac{1}{3} & 0 & 0 & 0 & -\frac{1}{3 \sqrt{7}} & 0 & 0 & 0 & -\frac{2\sqrt{5}}{3\sqrt{7}}  & 0 & 0 & \frac{\sqrt{5}}{3} & 0 & 0 & 0 & 0 \\
     \frac{1}{3} & 0 & \frac{1}{3} & 0 & \frac{\sqrt{7}}{3} & 0 & 0 & 0 & 0 & 0 & 0 & 0 & 0 & 0 & 0 & 0 & 0 \\
     0 & 0 & 0 & \frac{1}{3} & 0 & 0 & -\frac{1}{3 \sqrt{7}} & 0 & 0 & 0 & -\frac{2\sqrt{5}}{3\sqrt{7}}  & 0 & 0 & 0 & -\frac{\sqrt{5}}{3} & 0 & 0 \\
     \frac{\sqrt{7}}{3} & 0 & \frac{\sqrt{7}}{3} & 0 & \frac{7}{3} & 0 & 0 & 0 & 0 & 0 & 0 & 0 & 0 & 0 & 0 & 0 & 0 \\
     0 & -\frac{1}{3 \sqrt{7}} & 0 & 0 & 0 & \frac{1}{21} & 0 & 0 & 0 & \frac{2 \sqrt{5}}{21} & 0 & 0 & -\frac{\sqrt{5}}{3\sqrt{7}} & 0 & 0 & 0 & 0 \\
     0 & 0 & 0 & -\frac{1}{3 \sqrt{7}} & 0 & 0 & \frac{1}{21} & 0 & 0 & 0 & \frac{2 \sqrt{5}}{21} & 0 & 0 & 0 & \frac{\sqrt{5}}{3\sqrt{7}} & 0 & 0 \\
     0 & 0 & 0 & 0 & 0 & 0 & 0 & \frac{5}{21} & 0 & 0 & 0 & -\frac{4 \sqrt{5}}{21} & 0 & -\frac{\sqrt{5}}{3\sqrt{7}} & 0 & \frac{\sqrt{5}}{3\sqrt{7}} & -\frac{2\sqrt{5}}{3\sqrt{7}}  \\
     0 & 0 & 0 & 0 & 0 & 0 & 0 & 0 & 0 & 0 & 0 & 0 & 0 & 0 & 0 & 0 & 0 \\
     0 & -\frac{2\sqrt{5}}{3\sqrt{7}}  & 0 & 0 & 0 & \frac{2 \sqrt{5}}{21} & 0 & 0 & 0 & \frac{20}{21} & 0 & 0 & -\frac{10}{3 \sqrt{7}} & 0 & 0 & 0 & 0 \\
     0 & 0 & 0 & -\frac{2\sqrt{5}}{3\sqrt{7}}  & 0 & 0 & \frac{2 \sqrt{5}}{21} & 0 & 0 & 0 & \frac{20}{21} & 0 & 0 & 0 & \frac{10}{3 \sqrt{7}} & 0 & 0 \\
     0 & 0 & 0 & 0 & 0 & 0 & 0 & -\frac{4 \sqrt{5}}{21} & 0 & 0 & 0 & \frac{16}{21} & 0 & \frac{4}{3 \sqrt{7}} & 0 & -\frac{4}{3 \sqrt{7}} & \frac{8}{3 \sqrt{7}} \\
     0 & \frac{\sqrt{5}}{3} & 0 & 0 & 0 & -\frac{\sqrt{5}}{3\sqrt{7}} & 0 & 0 & 0 & -\frac{10}{3 \sqrt{7}} & 0 & 0 & \frac{5}{3} & 0 & 0 & 0 & 0 \\
     0 & 0 & 0 & 0 & 0 & 0 & 0 & -\frac{\sqrt{5}}{3\sqrt{7}} & 0 & 0 & 0 & \frac{4}{3 \sqrt{7}} & 0 & \frac{1}{3} & 0 & -\frac{1}{3} & \frac{2}{3} \\
     0 & 0 & 0 & -\frac{\sqrt{5}}{3} & 0 & 0 & \frac{\sqrt{5}}{3\sqrt{7}} & 0 & 0 & 0 & \frac{10}{3 \sqrt{7}} & 0 & 0 & 0 & \frac{5}{3} & 0 & 0 \\
     0 & 0 & 0 & 0 & 0 & 0 & 0 & \frac{\sqrt{5}}{3\sqrt{7}} & 0 & 0 & 0 & -\frac{4}{3 \sqrt{7}} & 0 & -\frac{1}{3} & 0 & \frac{1}{3} & -\frac{2}{3} \\
     0 & 0 & 0 & 0 & 0 & 0 & 0 & -\frac{2\sqrt{5}}{3\sqrt{7}}  & 0 & 0 & 0 & \frac{8}{3 \sqrt{7}} & 0 & \frac{2}{3} & 0 & -\frac{2}{3} & \frac{4}{3} \\
    \end{array}
    \right)$}
  \end{subfigure}

  \medskip

  \begin{subfigure}{0.42\textwidth}
    \centering
    \caption{$\Lambda = (2, 0, -2)$}\label{fig:tab3}
    \vspace{5pt}
    \adjustbox{max width=\linewidth}{$\left(
    \begin{array}{ccccccccc}
     0 & 0 & 0 & 0 & 0 & 0 & 0 & 0 & 0 \\
     0 & 0 & 0 & 0 & 0 & 0 & 0 & 0 & 0 \\
     0 & 0 & 0 & 0 & 0 & 0 & 0 & 0 & 0 \\
     0 & 0 & 0 & \frac{4}{9} & 0 & -\frac{4 \sqrt{2}}{9} & 0 & \frac{4 \sqrt{2}}{9} & -\frac{2 \sqrt{7}}{9} \\
     0 & 0 & 0 & 0 & 0 & 0 & 0 & 0 & 0 \\
     0 & 0 & 0 & -\frac{4 \sqrt{2}}{9} & 0 & \frac{8}{9} & 0 & -\frac{8}{9} & \frac{2 \sqrt{14}}{9} \\
     0 & 0 & 0 & 0 & 0 & 0 & 0 & 0 & 0 \\
     0 & 0 & 0 & \frac{4 \sqrt{2}}{9} & 0 & -\frac{8}{9} & 0 & \frac{8}{9} & -\frac{2 \sqrt{14}}{9}  \\
     0 & 0 & 0 & -\frac{2 \sqrt{7}}{9} & 0 & \frac{2 \sqrt{14}}{9} & 0 & -\frac{2 \sqrt{14}}{9}  & \frac{7}{9} \\
    \end{array}
    \right)$}
  \end{subfigure}
  \hfill
  \begin{subfigure}{0.27\textwidth}
    \centering
    \caption{$\Lambda = (1,1,-2)$}\label{fig:tab4}
    \vspace{5pt}
    \adjustbox{max width=\linewidth}{$
    \left(
    \begin{array}{ccccccc}
     \frac{1}{3} & 0 & -\frac{2}{3} & 0 & 0 & \frac{2}{3} & 0 \\
     0 & 0 & 0 & 0 & 0 & 0 & 0 \\
     -\frac{2}{3} & 0 & \frac{4}{3} & 0 & 0 & -\frac{4}{3} & 0 \\
     0 & 0 & 0 & 0 & 0 & 0 & 0 \\
     0 & 0 & 0 & 0 & 0 & 0 & 0 \\
     \frac{2}{3} & 0 & -\frac{4}{3} & 0 & 0 & \frac{4}{3} & 0 \\
     0 & 0 & 0 & 0 & 0 & 0 & 0 \\
    \end{array}
    \right)$}
  \end{subfigure}
  \hfill
  \begin{subfigure}{0.24\textwidth}
    \centering
    \caption{$\Lambda = (2,-1,-1)$}\label{fig:tab5}
    \vspace{5pt}
    \adjustbox{max width=\linewidth}{$
    \left(
    \begin{array}{ccccccc}
     0 & 0 & 0 & 0 & 0 & 0 & 0 \\
     0 & 0 & 0 & 0 & 0 & 0 & 0 \\
     0 & 0 & \frac{4}{3} & 0 & \frac{2}{3} & \frac{4}{3} & 0 \\
     0 & 0 & 0 & 0 & 0 & 0 & 0 \\
     0 & 0 & \frac{2}{3} & 0 & \frac{1}{3} & \frac{2}{3} & 0 \\
     0 & 0 & \frac{4}{3} & 0 & \frac{2}{3} & \frac{4}{3} & 0 \\
     0 & 0 & 0 & 0 & 0 & 0 & 0 \\
    \end{array}
    \right)$}
  \end{subfigure}

  \medskip

  \begin{subfigure}{0.2\textwidth}
    \centering
    \caption{$\Lambda = (3, 0, -3)$}\label{fig:tab6}
    \vspace{5pt}
    \adjustbox{max width=\linewidth}{$\left(
    \begin{array}{c}
     0 \\
    \end{array}
    \right)$}
  \end{subfigure}
  \hfill
  \begin{subfigure}{0.2\textwidth}
    \centering
    \caption{$\Lambda =(2,1,-3)$}\label{fig:tab7}
    \vspace{5pt}
    \adjustbox{max width=\linewidth}{$
    \left(
    \begin{array}{cc}
     0 & 0 \\
     0 & 0 \\
    \end{array}
    \right)$}
  \end{subfigure}
  \hfill
  \begin{subfigure}{0.2\textwidth}
    \centering
    \caption{$\Lambda =(3,-1,-2)$}\label{fig:tab8}
    \vspace{5pt}
    \adjustbox{max width=\linewidth}{$
    \left(
    \begin{array}{cc}
     0 & 0 \\
     0 & 0 \\
    \end{array}
    \right)$}
  \end{subfigure}

  \caption{All irrep matrices for the contraction generator $V^{(1)}$ of $\mathcal{A}^3_{3,3}$ in the symmetric group adapted basis.}
  \label{fig:A333_contraction_all_irreps}
\end{figure}

\newpage
Next, we present irrep matrices for all transposition generators, adapted to the Young--Yamanouchi basis of $\mathbb{C}[\s_3] \times \mathbb{C}[\s_3]$.

\begin{figure}[H]
  \centering
  \captionsetup{labelfont=bf}
  \small 

  \begin{subfigure}{0.31\textwidth}
    \centering
    \caption{$\Lambda = (0,0,0)$}\label{fig:tab1}
    \vspace{5pt}
    \adjustbox{max width=\linewidth}{$\left(
    \begin{array}{cccccc}
     -1 & 0 & 0 & 0 & 0 & 0 \\
     0 & -1 & 0 & 0 & 0 & 0 \\
     0 & 0 & -1 & 0 & 0 & 0 \\
     0 & 0 & 0 & 1 & 0 & 0 \\
     0 & 0 & 0 & 0 & 1 & 0 \\
     0 & 0 & 0 & 0 & 0 & 1 \\
    \end{array}
    \right)$}
  \end{subfigure}

    \medskip
  \begin{subfigure}{0.95\textwidth}
    \centering
    \caption{$\Lambda = (1,0,-1)$}\label{fig:tab2}
    \vspace{5pt}
    \adjustbox{max width=\linewidth}{$\left(
\begin{array}{ccccccccccccccccc}
 -1 & 0 & 0 & 0 & 0 & 0 & 0 & 0 & 0 & 0 & 0 & 0 & 0 & 0 & 0 & 0 & 0 \\
 0 & -1 & 0 & 0 & 0 & 0 & 0 & 0 & 0 & 0 & 0 & 0 & 0 & 0 & 0 & 0 & 0 \\
 0 & 0 & -1 & 0 & 0 & 0 & 0 & 0 & 0 & 0 & 0 & 0 & 0 & 0 & 0 & 0 & 0 \\
 0 & 0 & 0 & 1 & 0 & 0 & 0 & 0 & 0 & 0 & 0 & 0 & 0 & 0 & 0 & 0 & 0 \\
 0 & 0 & 0 & 0 & -1 & 0 & 0 & 0 & 0 & 0 & 0 & 0 & 0 & 0 & 0 & 0 & 0 \\
 0 & 0 & 0 & 0 & 0 & -1 & 0 & 0 & 0 & 0 & 0 & 0 & 0 & 0 & 0 & 0 & 0 \\
 0 & 0 & 0 & 0 & 0 & 0 & 1 & 0 & 0 & 0 & 0 & 0 & 0 & 0 & 0 & 0 & 0 \\
 0 & 0 & 0 & 0 & 0 & 0 & 0 & 1 & 0 & 0 & 0 & 0 & 0 & 0 & 0 & 0 & 0 \\
 0 & 0 & 0 & 0 & 0 & 0 & 0 & 0 & -1 & 0 & 0 & 0 & 0 & 0 & 0 & 0 & 0 \\
 0 & 0 & 0 & 0 & 0 & 0 & 0 & 0 & 0 & -1 & 0 & 0 & 0 & 0 & 0 & 0 & 0 \\
 0 & 0 & 0 & 0 & 0 & 0 & 0 & 0 & 0 & 0 & 1 & 0 & 0 & 0 & 0 & 0 & 0 \\
 0 & 0 & 0 & 0 & 0 & 0 & 0 & 0 & 0 & 0 & 0 & 1 & 0 & 0 & 0 & 0 & 0 \\
 0 & 0 & 0 & 0 & 0 & 0 & 0 & 0 & 0 & 0 & 0 & 0 & -1 & 0 & 0 & 0 & 0 \\
 0 & 0 & 0 & 0 & 0 & 0 & 0 & 0 & 0 & 0 & 0 & 0 & 0 & 1 & 0 & 0 & 0 \\
 0 & 0 & 0 & 0 & 0 & 0 & 0 & 0 & 0 & 0 & 0 & 0 & 0 & 0 & 1 & 0 & 0 \\
 0 & 0 & 0 & 0 & 0 & 0 & 0 & 0 & 0 & 0 & 0 & 0 & 0 & 0 & 0 & 1 & 0 \\
 0 & 0 & 0 & 0 & 0 & 0 & 0 & 0 & 0 & 0 & 0 & 0 & 0 & 0 & 0 & 0 & 1 \\
\end{array}
\right)$}
  \end{subfigure}

  \medskip

  \begin{subfigure}{0.42\textwidth}
    \centering
    \caption{$\Lambda = (2, 0, -2)$}\label{fig:tab3}
    \vspace{5pt}
    \adjustbox{max width=\linewidth}{$\left(
\begin{array}{ccccccccc}
 -1 & 0 & 0 & 0 & 0 & 0 & 0 & 0 & 0 \\
 0 & -1 & 0 & 0 & 0 & 0 & 0 & 0 & 0 \\
 0 & 0 & 1 & 0 & 0 & 0 & 0 & 0 & 0 \\
 0 & 0 & 0 & 1 & 0 & 0 & 0 & 0 & 0 \\
 0 & 0 & 0 & 0 & -1 & 0 & 0 & 0 & 0 \\
 0 & 0 & 0 & 0 & 0 & 1 & 0 & 0 & 0 \\
 0 & 0 & 0 & 0 & 0 & 0 & 1 & 0 & 0 \\
 0 & 0 & 0 & 0 & 0 & 0 & 0 & 1 & 0 \\
 0 & 0 & 0 & 0 & 0 & 0 & 0 & 0 & 1 \\
\end{array}
\right)$}
  \end{subfigure}
  \hfill
  \begin{subfigure}{0.27\textwidth}
    \centering
    \caption{$\Lambda = (1,1,-2)$}\label{fig:tab4}
    \vspace{5pt}
    \adjustbox{max width=\linewidth}{$
    \left(
\begin{array}{ccccccc}
 -1 & 0 & 0 & 0 & 0 & 0 & 0 \\
 0 & -1 & 0 & 0 & 0 & 0 & 0 \\
 0 & 0 & -1 & 0 & 0 & 0 & 0 \\
 0 & 0 & 0 & 1 & 0 & 0 & 0 \\
 0 & 0 & 0 & 0 & 1 & 0 & 0 \\
 0 & 0 & 0 & 0 & 0 & -1 & 0 \\
 0 & 0 & 0 & 0 & 0 & 0 & 1 \\
\end{array}
\right)$}
  \end{subfigure}
  \hfill
  \begin{subfigure}{0.24\textwidth}
    \centering
    \caption{$\Lambda = (2,-1,-1)$}\label{fig:tab5}
    \vspace{5pt}
    \adjustbox{max width=\linewidth}{$
    \left(
\begin{array}{ccccccc}
 -1 & 0 & 0 & 0 & 0 & 0 & 0 \\
 0 & -1 & 0 & 0 & 0 & 0 & 0 \\
 0 & 0 & 1 & 0 & 0 & 0 & 0 \\
 0 & 0 & 0 & 1 & 0 & 0 & 0 \\
 0 & 0 & 0 & 0 & 1 & 0 & 0 \\
 0 & 0 & 0 & 0 & 0 & 1 & 0 \\
 0 & 0 & 0 & 0 & 0 & 0 & 1 \\
\end{array}
\right)$}
  \end{subfigure}

  \medskip

  \begin{subfigure}{0.2\textwidth}
    \centering
    \caption{$\Lambda = (3, 0, -3)$}\label{fig:tab6}
    \vspace{5pt}
    \adjustbox{max width=\linewidth}{$\left(
    \begin{array}{c}
     1 \\
    \end{array}
    \right)$}
  \end{subfigure}
  \hfill
  \begin{subfigure}{0.2\textwidth}
    \centering
    \caption{$\Lambda =(2,1,-3)$}\label{fig:tab7}
    \vspace{5pt}
    \adjustbox{max width=\linewidth}{$
    \left(
    \begin{array}{cc}
     -1 & 0 \\
     0 & 1 \\
    \end{array}
    \right)$}
  \end{subfigure}
  \hfill
  \begin{subfigure}{0.2\textwidth}
    \centering
    \caption{$\Lambda =(3,-1,-2)$}\label{fig:tab8}
    \vspace{5pt}
    \adjustbox{max width=\linewidth}{$
    \left(
    \begin{array}{cc}
     1 & 0 \\
     0 & 1 \\
    \end{array}
    \right)$}
  \end{subfigure}

  \caption{All irrep matrices for the transposition $(12)$ generator of $\mathcal{A}^3_{3,3}$ in the symmetric group adapted basis.}
  \label{fig:A333_transposition_12_all_irreps}
\end{figure}

\begin{figure}[H]
  \centering
  \captionsetup{labelfont=bf}
  \small 

  \begin{subfigure}{0.31\textwidth}
    \centering
    \caption{$\Lambda = (0,0,0)$}\label{fig:tab1}
    \vspace{5pt}
    \adjustbox{max width=\linewidth}{$\left(
\begin{array}{cccccc}
 -1 & 0 & 0 & 0 & 0 & 0 \\
 0 & \frac{1}{2} & 0 & \frac{\sqrt{3}}{2} & 0 & 0 \\
 0 & 0 & \frac{1}{2} & 0 & \frac{\sqrt{3}}{2} & 0 \\
 0 & \frac{\sqrt{3}}{2} & 0 & -\frac{1}{2} & 0 & 0 \\
 0 & 0 & \frac{\sqrt{3}}{2} & 0 & -\frac{1}{2} & 0 \\
 0 & 0 & 0 & 0 & 0 & 1 \\
\end{array}
\right)$}
  \end{subfigure}

    \medskip
  \begin{subfigure}{0.95\textwidth}
    \centering
    \caption{$\Lambda = (1,0,-1)$}\label{fig:tab2}
    \vspace{5pt}
    \adjustbox{max width=\linewidth}{$\left(
\begin{array}{ccccccccccccccccc}
 -1 & 0 & 0 & 0 & 0 & 0 & 0 & 0 & 0 & 0 & 0 & 0 & 0 & 0 & 0 & 0 & 0 \\
 0 & -1 & 0 & 0 & 0 & 0 & 0 & 0 & 0 & 0 & 0 & 0 & 0 & 0 & 0 & 0 & 0 \\
 0 & 0 & \frac{1}{2} & \frac{\sqrt{3}}{2} & 0 & 0 & 0 & 0 & 0 & 0 & 0 & 0 & 0 & 0 & 0 & 0 & 0 \\
 0 & 0 & \frac{\sqrt{3}}{2} & -\frac{1}{2} & 0 & 0 & 0 & 0 & 0 & 0 & 0 & 0 & 0 & 0 & 0 & 0 & 0 \\
 0 & 0 & 0 & 0 & \frac{1}{2} & 0 & \frac{\sqrt{3}}{2} & 0 & 0 & 0 & 0 & 0 & 0 & 0 & 0 & 0 & 0 \\
 0 & 0 & 0 & 0 & 0 & \frac{1}{2} & 0 & \frac{\sqrt{3}}{2} & 0 & 0 & 0 & 0 & 0 & 0 & 0 & 0 & 0 \\
 0 & 0 & 0 & 0 & \frac{\sqrt{3}}{2} & 0 & -\frac{1}{2} & 0 & 0 & 0 & 0 & 0 & 0 & 0 & 0 & 0 & 0 \\
 0 & 0 & 0 & 0 & 0 & \frac{\sqrt{3}}{2} & 0 & -\frac{1}{2} & 0 & 0 & 0 & 0 & 0 & 0 & 0 & 0 & 0 \\
 0 & 0 & 0 & 0 & 0 & 0 & 0 & 0 & \frac{1}{2} & 0 & \frac{\sqrt{3}}{2} & 0 & 0 & 0 & 0 & 0 & 0 \\
 0 & 0 & 0 & 0 & 0 & 0 & 0 & 0 & 0 & \frac{1}{2} & 0 & \frac{\sqrt{3}}{2} & 0 & 0 & 0 & 0 & 0 \\
 0 & 0 & 0 & 0 & 0 & 0 & 0 & 0 & \frac{\sqrt{3}}{2} & 0 & -\frac{1}{2} & 0 & 0 & 0 & 0 & 0 & 0 \\
 0 & 0 & 0 & 0 & 0 & 0 & 0 & 0 & 0 & \frac{\sqrt{3}}{2} & 0 & -\frac{1}{2} & 0 & 0 & 0 & 0 & 0 \\
 0 & 0 & 0 & 0 & 0 & 0 & 0 & 0 & 0 & 0 & 0 & 0 & \frac{1}{2} & \frac{\sqrt{3}}{2} & 0 & 0 & 0 \\
 0 & 0 & 0 & 0 & 0 & 0 & 0 & 0 & 0 & 0 & 0 & 0 & \frac{\sqrt{3}}{2} & -\frac{1}{2} & 0 & 0 & 0 \\
 0 & 0 & 0 & 0 & 0 & 0 & 0 & 0 & 0 & 0 & 0 & 0 & 0 & 0 & 1 & 0 & 0 \\
 0 & 0 & 0 & 0 & 0 & 0 & 0 & 0 & 0 & 0 & 0 & 0 & 0 & 0 & 0 & 1 & 0 \\
 0 & 0 & 0 & 0 & 0 & 0 & 0 & 0 & 0 & 0 & 0 & 0 & 0 & 0 & 0 & 0 & 1 \\
\end{array}
\right)$}
  \end{subfigure}

  \medskip

  \begin{subfigure}{0.42\textwidth}
    \centering
    \caption{$\Lambda = (2, 0, -2)$}\label{fig:tab3}
    \vspace{5pt}
    \adjustbox{max width=\linewidth}{$\left(
\begin{array}{ccccccccc}
 \frac{1}{2} & 0 & \frac{\sqrt{3}}{2} & 0 & 0 & 0 & 0 & 0 & 0 \\
 0 & \frac{1}{2} & 0 & \frac{\sqrt{3}}{2} & 0 & 0 & 0 & 0 & 0 \\
 \frac{\sqrt{3}}{2} & 0 & -\frac{1}{2} & 0 & 0 & 0 & 0 & 0 & 0 \\
 0 & \frac{\sqrt{3}}{2} & 0 & -\frac{1}{2} & 0 & 0 & 0 & 0 & 0 \\
 0 & 0 & 0 & 0 & \frac{1}{2} & \frac{\sqrt{3}}{2} & 0 & 0 & 0 \\
 0 & 0 & 0 & 0 & \frac{\sqrt{3}}{2} & -\frac{1}{2} & 0 & 0 & 0 \\
 0 & 0 & 0 & 0 & 0 & 0 & 1 & 0 & 0 \\
 0 & 0 & 0 & 0 & 0 & 0 & 0 & 1 & 0 \\
 0 & 0 & 0 & 0 & 0 & 0 & 0 & 0 & 1 \\
\end{array}
\right)$}
  \end{subfigure}
  \hfill
  \begin{subfigure}{0.27\textwidth}
    \centering
    \caption{$\Lambda = (1,1,-2)$}\label{fig:tab4}
    \vspace{5pt}
    \adjustbox{max width=\linewidth}{$
    \left(
\begin{array}{ccccccc}
 -1 & 0 & 0 & 0 & 0 & 0 & 0 \\
 0 & \frac{1}{2} & 0 & \frac{\sqrt{3}}{2} & 0 & 0 & 0 \\
 0 & 0 & \frac{1}{2} & 0 & \frac{\sqrt{3}}{2} & 0 & 0 \\
 0 & \frac{\sqrt{3}}{2} & 0 & -\frac{1}{2} & 0 & 0 & 0 \\
 0 & 0 & \frac{\sqrt{3}}{2} & 0 & -\frac{1}{2} & 0 & 0 \\
 0 & 0 & 0 & 0 & 0 & \frac{1}{2} & \frac{\sqrt{3}}{2} \\
 0 & 0 & 0 & 0 & 0 & \frac{\sqrt{3}}{2} & -\frac{1}{2} \\
\end{array}
\right)$}
  \end{subfigure}
  \hfill
  \begin{subfigure}{0.24\textwidth}
    \centering
    \caption{$\Lambda = (2,-1,-1)$}\label{fig:tab5}
    \vspace{5pt}
    \adjustbox{max width=\linewidth}{$
    \left(
\begin{array}{ccccccc}
 \frac{1}{2} & 0 & \frac{\sqrt{3}}{2} & 0 & 0 & 0 & 0 \\
 0 & \frac{1}{2} & 0 & \frac{\sqrt{3}}{2} & 0 & 0 & 0 \\
 \frac{\sqrt{3}}{2} & 0 & -\frac{1}{2} & 0 & 0 & 0 & 0 \\
 0 & \frac{\sqrt{3}}{2} & 0 & -\frac{1}{2} & 0 & 0 & 0 \\
 0 & 0 & 0 & 0 & 1 & 0 & 0 \\
 0 & 0 & 0 & 0 & 0 & 1 & 0 \\
 0 & 0 & 0 & 0 & 0 & 0 & 1 \\
\end{array}
\right)$}
  \end{subfigure}

  \medskip

  \begin{subfigure}{0.2\textwidth}
    \centering
    \caption{$\Lambda = (3, 0, -3)$}\label{fig:tab6}
    \vspace{5pt}
    \adjustbox{max width=\linewidth}{$\left(
    \begin{array}{c}
     1 \\
    \end{array}
    \right)$}
  \end{subfigure}
  \hfill
  \begin{subfigure}{0.2\textwidth}
    \centering
    \caption{$\Lambda =(2,1,-3)$}\label{fig:tab7}
    \vspace{5pt}
    \adjustbox{max width=\linewidth}{$
    \left(
\begin{array}{cc}
 \frac{1}{2} & \frac{\sqrt{3}}{2} \\
 \frac{\sqrt{3}}{2} & -\frac{1}{2} \\
\end{array}
\right)$}
  \end{subfigure}
  \hfill
  \begin{subfigure}{0.2\textwidth}
    \centering
    \caption{$\Lambda =(3,-1,-2)$}\label{fig:tab8}
    \vspace{5pt}
    \adjustbox{max width=\linewidth}{$
    \left(
    \begin{array}{cc}
     1 & 0 \\
     0 & 1 \\
    \end{array}
    \right)$}
  \end{subfigure}

  \caption{All irrep matrices for the transposition $(23)$ generator of $\mathcal{A}^3_{3,3}$ in the symmetric group adapted basis.}
  \label{fig:A333_transposition_23_all_irreps}
\end{figure}

\begin{figure}[H]
  \centering
  \captionsetup{labelfont=bf}
  \small 

  \begin{subfigure}{0.31\textwidth}
    \centering
    \caption{$\Lambda = (0,0,0)$}\label{fig:tab1}
    \vspace{5pt}
    \adjustbox{max width=\linewidth}{$\left(
\begin{array}{cccccc}
 -1 & 0 & 0 & 0 & 0 & 0 \\
 0 & \frac{1}{2} & \frac{\sqrt{3}}{2} & 0 & 0 & 0 \\
 0 & \frac{\sqrt{3}}{2} & -\frac{1}{2} & 0 & 0 & 0 \\
 0 & 0 & 0 & \frac{1}{2} & \frac{\sqrt{3}}{2} & 0 \\
 0 & 0 & 0 & \frac{\sqrt{3}}{2} & -\frac{1}{2} & 0 \\
 0 & 0 & 0 & 0 & 0 & 1 \\
\end{array}
\right)$}
  \end{subfigure}

    \medskip
  \begin{subfigure}{0.95\textwidth}
    \centering
    \caption{$\Lambda = (1,0,-1)$}\label{fig:tab2}
    \vspace{5pt}
    \adjustbox{max width=\linewidth}{$\left(
\begin{array}{ccccccccccccccccc}
 \frac{1}{2} & \frac{\sqrt{3}}{2} & 0 & 0 & 0 & 0 & 0 & 0 & 0 & 0 & 0 & 0 & 0 & 0 & 0 & 0 & 0 \\
 \frac{\sqrt{3}}{2} & -\frac{1}{2} & 0 & 0 & 0 & 0 & 0 & 0 & 0 & 0 & 0 & 0 & 0 & 0 & 0 & 0 & 0 \\
 0 & 0 & -1 & 0 & 0 & 0 & 0 & 0 & 0 & 0 & 0 & 0 & 0 & 0 & 0 & 0 & 0 \\
 0 & 0 & 0 & -1 & 0 & 0 & 0 & 0 & 0 & 0 & 0 & 0 & 0 & 0 & 0 & 0 & 0 \\
 0 & 0 & 0 & 0 & \frac{1}{2} & \frac{\sqrt{3}}{2} & 0 & 0 & 0 & 0 & 0 & 0 & 0 & 0 & 0 & 0 & 0 \\
 0 & 0 & 0 & 0 & \frac{\sqrt{3}}{2} & -\frac{1}{2} & 0 & 0 & 0 & 0 & 0 & 0 & 0 & 0 & 0 & 0 & 0 \\
 0 & 0 & 0 & 0 & 0 & 0 & \frac{1}{2} & \frac{\sqrt{3}}{2} & 0 & 0 & 0 & 0 & 0 & 0 & 0 & 0 & 0 \\
 0 & 0 & 0 & 0 & 0 & 0 & \frac{\sqrt{3}}{2} & -\frac{1}{2} & 0 & 0 & 0 & 0 & 0 & 0 & 0 & 0 & 0 \\
 0 & 0 & 0 & 0 & 0 & 0 & 0 & 0 & \frac{1}{2} & \frac{\sqrt{3}}{2} & 0 & 0 & 0 & 0 & 0 & 0 & 0 \\
 0 & 0 & 0 & 0 & 0 & 0 & 0 & 0 & \frac{\sqrt{3}}{2} & -\frac{1}{2} & 0 & 0 & 0 & 0 & 0 & 0 & 0 \\
 0 & 0 & 0 & 0 & 0 & 0 & 0 & 0 & 0 & 0 & \frac{1}{2} & \frac{\sqrt{3}}{2} & 0 & 0 & 0 & 0 & 0 \\
 0 & 0 & 0 & 0 & 0 & 0 & 0 & 0 & 0 & 0 & \frac{\sqrt{3}}{2} & -\frac{1}{2} & 0 & 0 & 0 & 0 & 0 \\
 0 & 0 & 0 & 0 & 0 & 0 & 0 & 0 & 0 & 0 & 0 & 0 & 1 & 0 & 0 & 0 & 0 \\
 0 & 0 & 0 & 0 & 0 & 0 & 0 & 0 & 0 & 0 & 0 & 0 & 0 & 1 & 0 & 0 & 0 \\
 0 & 0 & 0 & 0 & 0 & 0 & 0 & 0 & 0 & 0 & 0 & 0 & 0 & 0 & \frac{1}{2} & \frac{\sqrt{3}}{2} & 0 \\
 0 & 0 & 0 & 0 & 0 & 0 & 0 & 0 & 0 & 0 & 0 & 0 & 0 & 0 & \frac{\sqrt{3}}{2} & -\frac{1}{2} & 0 \\
 0 & 0 & 0 & 0 & 0 & 0 & 0 & 0 & 0 & 0 & 0 & 0 & 0 & 0 & 0 & 0 & 1 \\
\end{array}
\right)$}
  \end{subfigure}

  \medskip

  \begin{subfigure}{0.42\textwidth}
    \centering
    \caption{$\Lambda = (2, 0, -2)$}\label{fig:tab3}
    \vspace{5pt}
    \adjustbox{max width=\linewidth}{$\left(
\begin{array}{ccccccccc}
 \frac{1}{2} & \frac{\sqrt{3}}{2} & 0 & 0 & 0 & 0 & 0 & 0 & 0 \\
 \frac{\sqrt{3}}{2} & -\frac{1}{2} & 0 & 0 & 0 & 0 & 0 & 0 & 0 \\
 0 & 0 & \frac{1}{2} & \frac{\sqrt{3}}{2} & 0 & 0 & 0 & 0 & 0 \\
 0 & 0 & \frac{\sqrt{3}}{2} & -\frac{1}{2} & 0 & 0 & 0 & 0 & 0 \\
 0 & 0 & 0 & 0 & 1 & 0 & 0 & 0 & 0 \\
 0 & 0 & 0 & 0 & 0 & 1 & 0 & 0 & 0 \\
 0 & 0 & 0 & 0 & 0 & 0 & \frac{1}{2} & \frac{\sqrt{3}}{2} & 0 \\
 0 & 0 & 0 & 0 & 0 & 0 & \frac{\sqrt{3}}{2} & -\frac{1}{2} & 0 \\
 0 & 0 & 0 & 0 & 0 & 0 & 0 & 0 & 1 \\
\end{array}
\right)$}
  \end{subfigure}
  \hfill
  \begin{subfigure}{0.27\textwidth}
    \centering
    \caption{$\Lambda = (1,1,-2)$}\label{fig:tab4}
    \vspace{5pt}
    \adjustbox{max width=\linewidth}{$
    \left(
\begin{array}{ccccccc}
 1 & 0 & 0 & 0 & 0 & 0 & 0 \\
 0 & \frac{1}{2} & \frac{\sqrt{3}}{2} & 0 & 0 & 0 & 0 \\
 0 & \frac{\sqrt{3}}{2} & -\frac{1}{2} & 0 & 0 & 0 & 0 \\
 0 & 0 & 0 & \frac{1}{2} & \frac{\sqrt{3}}{2} & 0 & 0 \\
 0 & 0 & 0 & \frac{\sqrt{3}}{2} & -\frac{1}{2} & 0 & 0 \\
 0 & 0 & 0 & 0 & 0 & 1 & 0 \\
 0 & 0 & 0 & 0 & 0 & 0 & 1 \\
\end{array}
\right)$}
  \end{subfigure}
  \hfill
  \begin{subfigure}{0.24\textwidth}
    \centering
    \caption{$\Lambda = (2,-1,-1)$}\label{fig:tab5}
    \vspace{5pt}
    \adjustbox{max width=\linewidth}{$
    \left(
\begin{array}{ccccccc}
 \frac{1}{2} & \frac{\sqrt{3}}{2} & 0 & 0 & 0 & 0 & 0 \\
 \frac{\sqrt{3}}{2} & -\frac{1}{2} & 0 & 0 & 0 & 0 & 0 \\
 0 & 0 & \frac{1}{2} & \frac{\sqrt{3}}{2} & 0 & 0 & 0 \\
 0 & 0 & \frac{\sqrt{3}}{2} & -\frac{1}{2} & 0 & 0 & 0 \\
 0 & 0 & 0 & 0 & -1 & 0 & 0 \\
 0 & 0 & 0 & 0 & 0 & \frac{1}{2} & \frac{\sqrt{3}}{2} \\
 0 & 0 & 0 & 0 & 0 & \frac{\sqrt{3}}{2} & -\frac{1}{2} \\
\end{array}
\right)$}
  \end{subfigure}

  \medskip

  \begin{subfigure}{0.2\textwidth}
    \centering
    \caption{$\Lambda = (3, 0, -3)$}\label{fig:tab6}
    \vspace{5pt}
    \adjustbox{max width=\linewidth}{$\left(
    \begin{array}{c}
     1 \\
    \end{array}
    \right)$}
  \end{subfigure}
  \hfill
  \begin{subfigure}{0.2\textwidth}
    \centering
    \caption{$\Lambda =(2,1,-3)$}\label{fig:tab7}
    \vspace{5pt}
    \adjustbox{max width=\linewidth}{$
    \left(
\begin{array}{cc}
 1 & 0 \\
 0 & 1 \\
\end{array}
\right)$}
  \end{subfigure}
  \hfill
  \begin{subfigure}{0.2\textwidth}
    \centering
    \caption{$\Lambda =(3,-1,-2)$}\label{fig:tab8}
    \vspace{5pt}
    \adjustbox{max width=\linewidth}{$
    \left(
\begin{array}{cc}
 \frac{1}{2} & \frac{\sqrt{3}}{2} \\
 \frac{\sqrt{3}}{2} & -\frac{1}{2} \\
\end{array}
\right)$}
  \end{subfigure}

  \caption{All irrep matrices for the transposition $(45)$ generator of $\mathcal{A}^3_{3,3}$ in the symmetric group adapted basis.}
  \label{fig:A333_transposition_45_all_irreps}
\end{figure}

\begin{figure}[H]
  \centering
  \captionsetup{labelfont=bf}
  \small 

  \begin{subfigure}{0.31\textwidth}
    \centering
    \caption{$\Lambda = (0,0,0)$}\label{fig:tab1}
    \vspace{5pt}
    \adjustbox{max width=\linewidth}{$\left(
\begin{array}{cccccc}
 -1 & 0 & 0 & 0 & 0 & 0 \\
 0 & -1 & 0 & 0 & 0 & 0 \\
 0 & 0 & 1 & 0 & 0 & 0 \\
 0 & 0 & 0 & -1 & 0 & 0 \\
 0 & 0 & 0 & 0 & 1 & 0 \\
 0 & 0 & 0 & 0 & 0 & 1 \\
\end{array}
\right)$}
  \end{subfigure}

    \medskip
  \begin{subfigure}{0.95\textwidth}
    \centering
    \caption{$\Lambda = (1,0,-1)$}\label{fig:tab2}
    \vspace{5pt}
    \adjustbox{max width=\linewidth}{$\left(
\begin{array}{ccccccccccccccccc}
 -1 & 0 & 0 & 0 & 0 & 0 & 0 & 0 & 0 & 0 & 0 & 0 & 0 & 0 & 0 & 0 & 0 \\
 0 & 1 & 0 & 0 & 0 & 0 & 0 & 0 & 0 & 0 & 0 & 0 & 0 & 0 & 0 & 0 & 0 \\
 0 & 0 & -1 & 0 & 0 & 0 & 0 & 0 & 0 & 0 & 0 & 0 & 0 & 0 & 0 & 0 & 0 \\
 0 & 0 & 0 & -1 & 0 & 0 & 0 & 0 & 0 & 0 & 0 & 0 & 0 & 0 & 0 & 0 & 0 \\
 0 & 0 & 0 & 0 & -1 & 0 & 0 & 0 & 0 & 0 & 0 & 0 & 0 & 0 & 0 & 0 & 0 \\
 0 & 0 & 0 & 0 & 0 & 1 & 0 & 0 & 0 & 0 & 0 & 0 & 0 & 0 & 0 & 0 & 0 \\
 0 & 0 & 0 & 0 & 0 & 0 & -1 & 0 & 0 & 0 & 0 & 0 & 0 & 0 & 0 & 0 & 0 \\
 0 & 0 & 0 & 0 & 0 & 0 & 0 & 1 & 0 & 0 & 0 & 0 & 0 & 0 & 0 & 0 & 0 \\
 0 & 0 & 0 & 0 & 0 & 0 & 0 & 0 & -1 & 0 & 0 & 0 & 0 & 0 & 0 & 0 & 0 \\
 0 & 0 & 0 & 0 & 0 & 0 & 0 & 0 & 0 & 1 & 0 & 0 & 0 & 0 & 0 & 0 & 0 \\
 0 & 0 & 0 & 0 & 0 & 0 & 0 & 0 & 0 & 0 & -1 & 0 & 0 & 0 & 0 & 0 & 0 \\
 0 & 0 & 0 & 0 & 0 & 0 & 0 & 0 & 0 & 0 & 0 & 1 & 0 & 0 & 0 & 0 & 0 \\
 0 & 0 & 0 & 0 & 0 & 0 & 0 & 0 & 0 & 0 & 0 & 0 & 1 & 0 & 0 & 0 & 0 \\
 0 & 0 & 0 & 0 & 0 & 0 & 0 & 0 & 0 & 0 & 0 & 0 & 0 & 1 & 0 & 0 & 0 \\
 0 & 0 & 0 & 0 & 0 & 0 & 0 & 0 & 0 & 0 & 0 & 0 & 0 & 0 & -1 & 0 & 0 \\
 0 & 0 & 0 & 0 & 0 & 0 & 0 & 0 & 0 & 0 & 0 & 0 & 0 & 0 & 0 & 1 & 0 \\
 0 & 0 & 0 & 0 & 0 & 0 & 0 & 0 & 0 & 0 & 0 & 0 & 0 & 0 & 0 & 0 & 1 \\
\end{array}
\right)$}
  \end{subfigure}

  \medskip

  \begin{subfigure}{0.42\textwidth}
    \centering
    \caption{$\Lambda = (2, 0, -2)$}\label{fig:tab3}
    \vspace{5pt}
    \adjustbox{max width=\linewidth}{$\left(
\begin{array}{ccccccccc}
 -1 & 0 & 0 & 0 & 0 & 0 & 0 & 0 & 0 \\
 0 & 1 & 0 & 0 & 0 & 0 & 0 & 0 & 0 \\
 0 & 0 & -1 & 0 & 0 & 0 & 0 & 0 & 0 \\
 0 & 0 & 0 & 1 & 0 & 0 & 0 & 0 & 0 \\
 0 & 0 & 0 & 0 & 1 & 0 & 0 & 0 & 0 \\
 0 & 0 & 0 & 0 & 0 & 1 & 0 & 0 & 0 \\
 0 & 0 & 0 & 0 & 0 & 0 & -1 & 0 & 0 \\
 0 & 0 & 0 & 0 & 0 & 0 & 0 & 1 & 0 \\
 0 & 0 & 0 & 0 & 0 & 0 & 0 & 0 & 1 \\
\end{array}
\right)$}
  \end{subfigure}
  \hfill
  \begin{subfigure}{0.27\textwidth}
    \centering
    \caption{$\Lambda = (1,1,-2)$}\label{fig:tab4}
    \vspace{5pt}
    \adjustbox{max width=\linewidth}{$
    \left(
\begin{array}{ccccccc}
 1 & 0 & 0 & 0 & 0 & 0 & 0 \\
 0 & -1 & 0 & 0 & 0 & 0 & 0 \\
 0 & 0 & 1 & 0 & 0 & 0 & 0 \\
 0 & 0 & 0 & -1 & 0 & 0 & 0 \\
 0 & 0 & 0 & 0 & 1 & 0 & 0 \\
 0 & 0 & 0 & 0 & 0 & 1 & 0 \\
 0 & 0 & 0 & 0 & 0 & 0 & 1 \\
\end{array}
\right)$}
  \end{subfigure}
  \hfill
  \begin{subfigure}{0.24\textwidth}
    \centering
    \caption{$\Lambda = (2,-1,-1)$}\label{fig:tab5}
    \vspace{5pt}
    \adjustbox{max width=\linewidth}{$
    \left(
\begin{array}{ccccccc}
 -1 & 0 & 0 & 0 & 0 & 0 & 0 \\
 0 & 1 & 0 & 0 & 0 & 0 & 0 \\
 0 & 0 & -1 & 0 & 0 & 0 & 0 \\
 0 & 0 & 0 & 1 & 0 & 0 & 0 \\
 0 & 0 & 0 & 0 & -1 & 0 & 0 \\
 0 & 0 & 0 & 0 & 0 & -1 & 0 \\
 0 & 0 & 0 & 0 & 0 & 0 & 1 \\
\end{array}
\right)$}
  \end{subfigure}

  \medskip

  \begin{subfigure}{0.2\textwidth}
    \centering
    \caption{$\Lambda = (3, 0, -3)$}\label{fig:tab6}
    \vspace{5pt}
    \adjustbox{max width=\linewidth}{$\left(
    \begin{array}{c}
     1 \\
    \end{array}
    \right)$}
  \end{subfigure}
  \hfill
  \begin{subfigure}{0.2\textwidth}
    \centering
    \caption{$\Lambda =(2,1,-3)$}\label{fig:tab7}
    \vspace{5pt}
    \adjustbox{max width=\linewidth}{$
    \left(
\begin{array}{cc}
 1 & 0 \\
 0 & 1 \\
\end{array}
\right)$}
  \end{subfigure}
  \hfill
  \begin{subfigure}{0.2\textwidth}
    \centering
    \caption{$\Lambda =(3,-1,-2)$}\label{fig:tab8}
    \vspace{5pt}
    \adjustbox{max width=\linewidth}{$
    \left(
\begin{array}{cc}
 -1 & 0 \\
 0 & 1 \\
\end{array}
\right)$}
  \end{subfigure}

  \caption{All irrep matrices for the transposition $(56)$ generator of $\mathcal{A}^3_{3,3}$ in the symmetric group adapted basis.}
  \label{fig:A333_transposition_56_all_irreps}
\end{figure}

\newpage
\section{Example: operators $\rho(k)$ for $k=1,2,3$ from the algebra $\mathcal{A}^3_{3,3}$ in the $\mathbb{C}[\s_3] \times \mathbb{C}[\s_3]$--adapted basis}\label{app:tensor_network_example_rhos}

Given all irreducible representation matrices of the contraction generator $V^{(1)}$ (see Figure \ref{fig:A333_contraction_all_irreps}), we can compute the twirl $\rho(1)$ quite easily using Schur's lemma. For a given irrep $\Lambda$ we need to take summations over diagonal elements, which correspond to the same irreps upon restriction to the subalgebra $\mathbb{C}[\s_3] \times \mathbb{C}[\s_3]$. This is trivial to implement, since our basis is already $\mathbb{C}[\s_3] \times \mathbb{C}[\s_3]$-adapted. The resulting number is the same up to a normalization as the corresponding matrix element of $\rho(1)$: due to $\mathbb{C}[\s_3] \times \mathbb{C}[\s_3]$ invariance the matrix $\rho(1)$ acts as a constant times identity on the irreps of $\mathbb{C}[\s_3] \times \mathbb{C}[\s_3]$.

\begin{figure}[H]
  \centering
  \captionsetup{labelfont=bf}
  \small 

  \begin{subfigure}{0.31\textwidth}
    \centering
    \caption{$\Lambda = (0,0,0)$}\label{fig:tab1}
    \vspace{5pt}
    \adjustbox{max width=\linewidth}{
    $\left(
\begin{array}{cccccc}
 \frac{1}{3} & 0 & 0 & 0 & 0 & 0 \\
 0 & 1 & 0 & 0 & 0 & 0 \\
 0 & 0 & 1 & 0 & 0 & 0 \\
 0 & 0 & 0 & 1 & 0 & 0 \\
 0 & 0 & 0 & 0 & 1 & 0 \\
 0 & 0 & 0 & 0 & 0 & \frac{5}{3} \\
\end{array}
\right)$}
  \end{subfigure}

    \medskip
  \begin{subfigure}{0.95\textwidth}
    \centering
    \caption{$\Lambda = (1,0,-1)$}\label{fig:tab2}
    \vspace{5pt}
    \adjustbox{max width=\linewidth}{$
    \left(
\begin{array}{ccccccccccccccccc}
 \frac{1}{3} & 0 & 0 & 0 & 0 & 0 & 0 & 0 & 0 & 0 & 0 & 0 & 0 & 0 & 0 & 0 & 0 \\
 0 & \frac{1}{3} & 0 & 0 & 0 & 0 & 0 & 0 & 0 & 0 & 0 & 0 & 0 & 0 & 0 & 0 & 0 \\
 0 & 0 & \frac{1}{3} & 0 & 0 & 0 & 0 & 0 & 0 & 0 & 0 & 0 & 0 & 0 & 0 & 0 & 0 \\
 0 & 0 & 0 & \frac{1}{3} & 0 & 0 & 0 & 0 & 0 & 0 & 0 & 0 & 0 & 0 & 0 & 0 & 0 \\
 0 & 0 & 0 & 0 & \frac{2}{3} & 0 & 0 & 0 & 0 & 0 & 0 & 0 & 0 & 0 & 0 & 0 & 0 \\
 0 & 0 & 0 & 0 & 0 & \frac{2}{3} & 0 & 0 & 0 & 0 & 0 & 0 & 0 & 0 & 0 & 0 & 0 \\
 0 & 0 & 0 & 0 & 0 & 0 & \frac{2}{3} & 0 & 0 & 0 & 0 & 0 & 0 & 0 & 0 & 0 & 0 \\
 0 & 0 & 0 & 0 & 0 & 0 & 0 & \frac{2}{3} & 0 & 0 & 0 & 0 & 0 & 0 & 0 & 0 & 0 \\
 0 & 0 & 0 & 0 & 0 & 0 & 0 & 0 & \frac{2}{3} & 0 & 0 & 0 & 0 & 0 & 0 & 0 & 0 \\
 0 & 0 & 0 & 0 & 0 & 0 & 0 & 0 & 0 & \frac{2}{3} & 0 & 0 & 0 & 0 & 0 & 0 & 0 \\
 0 & 0 & 0 & 0 & 0 & 0 & 0 & 0 & 0 & 0 & \frac{2}{3} & 0 & 0 & 0 & 0 & 0 & 0 \\
 0 & 0 & 0 & 0 & 0 & 0 & 0 & 0 & 0 & 0 & 0 & \frac{2}{3} & 0 & 0 & 0 & 0 & 0 \\
 0 & 0 & 0 & 0 & 0 & 0 & 0 & 0 & 0 & 0 & 0 & 0 & 1 & 0 & 0 & 0 & 0 \\
 0 & 0 & 0 & 0 & 0 & 0 & 0 & 0 & 0 & 0 & 0 & 0 & 0 & 1 & 0 & 0 & 0 \\
 0 & 0 & 0 & 0 & 0 & 0 & 0 & 0 & 0 & 0 & 0 & 0 & 0 & 0 & 1 & 0 & 0 \\
 0 & 0 & 0 & 0 & 0 & 0 & 0 & 0 & 0 & 0 & 0 & 0 & 0 & 0 & 0 & 1 & 0 \\
 0 & 0 & 0 & 0 & 0 & 0 & 0 & 0 & 0 & 0 & 0 & 0 & 0 & 0 & 0 & 0 & \frac{4}{3} \\
\end{array}
\right)$}
  \end{subfigure}

  \medskip

  \begin{subfigure}{0.42\textwidth}
    \centering
    \caption{$\Lambda = (2, 0, -2)$}\label{fig:tab3}
    \vspace{5pt}
    \adjustbox{max width=\linewidth}{$\left(
\begin{array}{ccccccccc}
 \frac{1}{9} & 0 & 0 & 0 & 0 & 0 & 0 & 0 & 0 \\
 0 & \frac{1}{9} & 0 & 0 & 0 & 0 & 0 & 0 & 0 \\
 0 & 0 & \frac{1}{9} & 0 & 0 & 0 & 0 & 0 & 0 \\
 0 & 0 & 0 & \frac{1}{9} & 0 & 0 & 0 & 0 & 0 \\
 0 & 0 & 0 & 0 & \frac{4}{9} & 0 & 0 & 0 & 0 \\
 0 & 0 & 0 & 0 & 0 & \frac{4}{9} & 0 & 0 & 0 \\
 0 & 0 & 0 & 0 & 0 & 0 & \frac{4}{9} & 0 & 0 \\
 0 & 0 & 0 & 0 & 0 & 0 & 0 & \frac{4}{9} & 0 \\
 0 & 0 & 0 & 0 & 0 & 0 & 0 & 0 & \frac{7}{9} \\
\end{array}
\right)$}
  \end{subfigure}
  \hfill
  \begin{subfigure}{0.27\textwidth}
    \centering
    \caption{$\Lambda = (1,1,-2)$}\label{fig:tab4}
    \vspace{5pt}
    \adjustbox{max width=\linewidth}{$
    \left(
\begin{array}{ccccccc}
 \frac{1}{3} & 0 & 0 & 0 & 0 & 0 & 0 \\
 0 & \frac{1}{3} & 0 & 0 & 0 & 0 & 0 \\
 0 & 0 & \frac{1}{3} & 0 & 0 & 0 & 0 \\
 0 & 0 & 0 & \frac{1}{3} & 0 & 0 & 0 \\
 0 & 0 & 0 & 0 & \frac{1}{3} & 0 & 0 \\
 0 & 0 & 0 & 0 & 0 & \frac{2}{3} & 0 \\
 0 & 0 & 0 & 0 & 0 & 0 & \frac{2}{3} \\
\end{array}
\right)$}
  \end{subfigure}
  \hfill
  \begin{subfigure}{0.24\textwidth}
    \centering
    \caption{$\Lambda = (2,-1,-1)$}\label{fig:tab5}
    \vspace{5pt}
    \adjustbox{max width=\linewidth}{$
    \left(
\begin{array}{ccccccc}
 \frac{1}{3} & 0 & 0 & 0 & 0 & 0 & 0 \\
 0 & \frac{1}{3} & 0 & 0 & 0 & 0 & 0 \\
 0 & 0 & \frac{1}{3} & 0 & 0 & 0 & 0 \\
 0 & 0 & 0 & \frac{1}{3} & 0 & 0 & 0 \\
 0 & 0 & 0 & 0 & \frac{1}{3} & 0 & 0 \\
 0 & 0 & 0 & 0 & 0 & \frac{2}{3} & 0 \\
 0 & 0 & 0 & 0 & 0 & 0 & \frac{2}{3} \\
\end{array}
\right)$}
  \end{subfigure}

  \medskip

  \begin{subfigure}{0.2\textwidth}
    \centering
    \caption{$\Lambda = (3, 0, -3)$}\label{fig:tab6}
    \vspace{5pt}
    \adjustbox{max width=\linewidth}{$\left(
    \begin{array}{c}
     0 \\
    \end{array}
    \right)$}
  \end{subfigure}
  \hfill
  \begin{subfigure}{0.2\textwidth}
    \centering
    \caption{$\Lambda =(2,1,-3)$}\label{fig:tab7}
    \vspace{5pt}
    \adjustbox{max width=\linewidth}{$
    \left(
\begin{array}{cc}
 0 & 0 \\
 0 & 0 \\
\end{array}
\right)$}
  \end{subfigure}
  \hfill
  \begin{subfigure}{0.2\textwidth}
    \centering
    \caption{$\Lambda =(3,-1,-2)$}\label{fig:tab8}
    \vspace{5pt}
    \adjustbox{max width=\linewidth}{$
    \left(
\begin{array}{cc}
 0 & 0 \\
 0 & 0 \\
\end{array}
\right)$}
  \end{subfigure}

  \caption{All irrep matrices for the operator $\rho(1)$ from the algebra $\mathcal{A}^3_{3,3}$ in the symmetric group adapted basis.}
  \label{fig:A333_rho_1_all_irreps}
\end{figure}

\newpage
Similarly to $\rho(1)$, we can compute irrep matrices of $\rho(2)$ and $\rho(3)$ using data from Appendix~\ref{app:tensor_network_example_gens}. Note that there can be non-trivial off-diagonal elements due to non-trivial multiplicities of $\mathbb{C}[\s_3] \times \mathbb{C}[\s_3]$ irreps inside the algebra $\mathcal{A}^3_{3,3}$.

\begin{figure}[H]
  \centering
  \captionsetup{labelfont=bf}
  \small 

  \begin{subfigure}{0.31\textwidth}
    \centering
    \caption{$\Lambda = (0,0,0)$}\label{fig:tab1}
    \vspace{5pt}
    \adjustbox{max width=\linewidth}{$\left(
\begin{array}{cccccc}
 \frac{1}{3} & 0 & 0 & 0 & 0 & 0 \\
 0 & \frac{4}{3} & 0 & 0 & 0 & 0 \\
 0 & 0 & \frac{4}{3} & 0 & 0 & 0 \\
 0 & 0 & 0 & \frac{4}{3} & 0 & 0 \\
 0 & 0 & 0 & 0 & \frac{4}{3} & 0 \\
 0 & 0 & 0 & 0 & 0 & \frac{10}{3} \\
\end{array}
\right)$}
  \end{subfigure}

    \medskip
  \begin{subfigure}{0.95\textwidth}
    \centering
    \caption{$\Lambda = (1,0,-1)$}\label{fig:tab2}
    \vspace{5pt}
    \adjustbox{max width=\linewidth}{$
    \left(
\begin{array}{ccccccccccccccccc}
 \frac{1}{6} & 0 & 0 & 0 & 0 & 0 & 0 & 0 & 0 & 0 & 0 & 0 & 0 & 0 & 0 & 0 & 0 \\
 0 & \frac{1}{6} & 0 & 0 & 0 & 0 & 0 & 0 & 0 & 0 & 0 & 0 & 0 & 0 & 0 & 0 & 0 \\
 0 & 0 & \frac{1}{6} & 0 & 0 & 0 & 0 & 0 & 0 & 0 & 0 & 0 & 0 & 0 & 0 & 0 & 0 \\
 0 & 0 & 0 & \frac{1}{6} & 0 & 0 & 0 & 0 & 0 & 0 & 0 & 0 & 0 & 0 & 0 & 0 & 0 \\
 0 & 0 & 0 & 0 & \frac{5}{21} & 0 & 0 & 0 & -\frac{\sqrt{5}}{42} & 0 & 0 & 0 & 0 & 0 & 0 & 0 & 0 \\
 0 & 0 & 0 & 0 & 0 & \frac{5}{21} & 0 & 0 & 0 & -\frac{\sqrt{5}}{42} & 0 & 0 & 0 & 0 & 0 & 0 & 0 \\
 0 & 0 & 0 & 0 & 0 & 0 & \frac{5}{21} & 0 & 0 & 0 & -\frac{\sqrt{5}}{42} & 0 & 0 & 0 & 0 & 0 & 0 \\
 0 & 0 & 0 & 0 & 0 & 0 & 0 & \frac{5}{21} & 0 & 0 & 0 & -\frac{\sqrt{5}}{42} & 0 & 0 & 0 & 0 & 0 \\
 0 & 0 & 0 & 0 & -\frac{\sqrt{5}}{42} & 0 & 0 & 0 & \frac{25}{42} & 0 & 0 & 0 & 0 & 0 & 0 & 0 & 0 \\
 0 & 0 & 0 & 0 & 0 & -\frac{\sqrt{5}}{42} & 0 & 0 & 0 & \frac{25}{42} & 0 & 0 & 0 & 0 & 0 & 0 & 0 \\
 0 & 0 & 0 & 0 & 0 & 0 & -\frac{\sqrt{5}}{42} & 0 & 0 & 0 & \frac{25}{42} & 0 & 0 & 0 & 0 & 0 & 0 \\
 0 & 0 & 0 & 0 & 0 & 0 & 0 & -\frac{\sqrt{5}}{42} & 0 & 0 & 0 & \frac{25}{42} & 0 & 0 & 0 & 0 & 0 \\
 0 & 0 & 0 & 0 & 0 & 0 & 0 & 0 & 0 & 0 & 0 & 0 & \frac{5}{6} & 0 & 0 & 0 & 0 \\
 0 & 0 & 0 & 0 & 0 & 0 & 0 & 0 & 0 & 0 & 0 & 0 & 0 & \frac{5}{6} & 0 & 0 & 0 \\
 0 & 0 & 0 & 0 & 0 & 0 & 0 & 0 & 0 & 0 & 0 & 0 & 0 & 0 & \frac{5}{6} & 0 & 0 \\
 0 & 0 & 0 & 0 & 0 & 0 & 0 & 0 & 0 & 0 & 0 & 0 & 0 & 0 & 0 & \frac{5}{6} & 0 \\
 0 & 0 & 0 & 0 & 0 & 0 & 0 & 0 & 0 & 0 & 0 & 0 & 0 & 0 & 0 & 0 & \frac{5}{3} \\
\end{array}
\right)$}
  \end{subfigure}

  \medskip

  \begin{subfigure}{0.42\textwidth}
    \centering
    \caption{$\Lambda = (2, 0, -2)$}\label{fig:tab3}
    \vspace{5pt}
    \adjustbox{max width=\linewidth}{$\left(
\begin{array}{ccccccccc}
 0 & 0 & 0 & 0 & 0 & 0 & 0 & 0 & 0 \\
 0 & 0 & 0 & 0 & 0 & 0 & 0 & 0 & 0 \\
 0 & 0 & 0 & 0 & 0 & 0 & 0 & 0 & 0 \\
 0 & 0 & 0 & 0 & 0 & 0 & 0 & 0 & 0 \\
 0 & 0 & 0 & 0 & 0 & 0 & 0 & 0 & 0 \\
 0 & 0 & 0 & 0 & 0 & 0 & 0 & 0 & 0 \\
 0 & 0 & 0 & 0 & 0 & 0 & 0 & 0 & 0 \\
 0 & 0 & 0 & 0 & 0 & 0 & 0 & 0 & 0 \\
 0 & 0 & 0 & 0 & 0 & 0 & 0 & 0 & 0 \\
\end{array}
\right)$}
  \end{subfigure}
  \hfill
  \begin{subfigure}{0.27\textwidth}
    \centering
    \caption{$\Lambda = (1,1,-2)$}\label{fig:tab4}
    \vspace{5pt}
    \adjustbox{max width=\linewidth}{$
    \left(
\begin{array}{ccccccc}
 0 & 0 & 0 & 0 & 0 & 0 & 0 \\
 0 & 0 & 0 & 0 & 0 & 0 & 0 \\
 0 & 0 & 0 & 0 & 0 & 0 & 0 \\
 0 & 0 & 0 & 0 & 0 & 0 & 0 \\
 0 & 0 & 0 & 0 & 0 & 0 & 0 \\
 0 & 0 & 0 & 0 & 0 & 0 & 0 \\
 0 & 0 & 0 & 0 & 0 & 0 & 0 \\
\end{array}
\right)$}
  \end{subfigure}
  \hfill
  \begin{subfigure}{0.24\textwidth}
    \centering
    \caption{$\Lambda = (2,-1,-1)$}\label{fig:tab5}
    \vspace{5pt}
    \adjustbox{max width=\linewidth}{$
    \left(
\begin{array}{ccccccc}
 0 & 0 & 0 & 0 & 0 & 0 & 0 \\
 0 & 0 & 0 & 0 & 0 & 0 & 0 \\
 0 & 0 & 0 & 0 & 0 & 0 & 0 \\
 0 & 0 & 0 & 0 & 0 & 0 & 0 \\
 0 & 0 & 0 & 0 & 0 & 0 & 0 \\
 0 & 0 & 0 & 0 & 0 & 0 & 0 \\
 0 & 0 & 0 & 0 & 0 & 0 & 0 \\
\end{array}
\right)$}
  \end{subfigure}

  \medskip

  \begin{subfigure}{0.2\textwidth}
    \centering
    \caption{$\Lambda = (3, 0, -3)$}\label{fig:tab6}
    \vspace{5pt}
    \adjustbox{max width=\linewidth}{$\left(
    \begin{array}{c}
     0 \\
    \end{array}
    \right)$}
  \end{subfigure}
  \hfill
  \begin{subfigure}{0.2\textwidth}
    \centering
    \caption{$\Lambda =(2,1,-3)$}\label{fig:tab7}
    \vspace{5pt}
    \adjustbox{max width=\linewidth}{$
    \left(
\begin{array}{cc}
 0 & 0 \\
 0 & 0 \\
\end{array}
\right)$}
  \end{subfigure}
  \hfill
  \begin{subfigure}{0.2\textwidth}
    \centering
    \caption{$\Lambda =(3,-1,-2)$}\label{fig:tab8}
    \vspace{5pt}
    \adjustbox{max width=\linewidth}{$
    \left(
\begin{array}{cc}
 0 & 0 \\
 0 & 0 \\
\end{array}
\right)$}
  \end{subfigure}

  \caption{All irrep matrices for the operator $\rho(2)$ from the algebra $\mathcal{A}^3_{3,3}$ in the symmetric group adapted basis. The eigenvalues of the two-by-two blocks inside the irrep $\Lambda = (1,0,-1)$ are $\frac{5+\sqrt 5}{12}$ and $ \frac{5-\sqrt 5}{12}$.}
  \label{fig:A333_rho_2_all_irreps}
\end{figure}

\begin{figure}[H]
  \centering
  \captionsetup{labelfont=bf}
  \small 

  \begin{subfigure}{0.31\textwidth}
    \centering
    \caption{$\Lambda = (0,0,0)$}\label{fig:tab1}
    \vspace{5pt}
    \adjustbox{max width=\linewidth}{$\left(
\begin{array}{cccccc}
 1 & 0 & 0 & 0 & 0 & 0 \\
 0 & 4 & 0 & 0 & 0 & 0 \\
 0 & 0 & 4 & 0 & 0 & 0 \\
 0 & 0 & 0 & 4 & 0 & 0 \\
 0 & 0 & 0 & 0 & 4 & 0 \\
 0 & 0 & 0 & 0 & 0 & 10 \\
\end{array}
\right)$}
  \end{subfigure}

    \medskip
  \begin{subfigure}{0.95\textwidth}
    \centering
    \caption{$\Lambda = (1,0,-1)$}\label{fig:tab2}
    \vspace{5pt}
    \adjustbox{max width=\linewidth}{$
    \left(
\begin{array}{ccccccccccccccccc}
 0 & 0 & 0 & 0 & 0 & 0 & 0 & 0 & 0 & 0 & 0 & 0 & 0 & 0 & 0 & 0 & 0 \\
 0 & 0 & 0 & 0 & 0 & 0 & 0 & 0 & 0 & 0 & 0 & 0 & 0 & 0 & 0 & 0 & 0 \\
 0 & 0 & 0 & 0 & 0 & 0 & 0 & 0 & 0 & 0 & 0 & 0 & 0 & 0 & 0 & 0 & 0 \\
 0 & 0 & 0 & 0 & 0 & 0 & 0 & 0 & 0 & 0 & 0 & 0 & 0 & 0 & 0 & 0 & 0 \\
 0 & 0 & 0 & 0 & 0 & 0 & 0 & 0 & 0 & 0 & 0 & 0 & 0 & 0 & 0 & 0 & 0 \\
 0 & 0 & 0 & 0 & 0 & 0 & 0 & 0 & 0 & 0 & 0 & 0 & 0 & 0 & 0 & 0 & 0 \\
 0 & 0 & 0 & 0 & 0 & 0 & 0 & 0 & 0 & 0 & 0 & 0 & 0 & 0 & 0 & 0 & 0 \\
 0 & 0 & 0 & 0 & 0 & 0 & 0 & 0 & 0 & 0 & 0 & 0 & 0 & 0 & 0 & 0 & 0 \\
 0 & 0 & 0 & 0 & 0 & 0 & 0 & 0 & 0 & 0 & 0 & 0 & 0 & 0 & 0 & 0 & 0 \\
 0 & 0 & 0 & 0 & 0 & 0 & 0 & 0 & 0 & 0 & 0 & 0 & 0 & 0 & 0 & 0 & 0 \\
 0 & 0 & 0 & 0 & 0 & 0 & 0 & 0 & 0 & 0 & 0 & 0 & 0 & 0 & 0 & 0 & 0 \\
 0 & 0 & 0 & 0 & 0 & 0 & 0 & 0 & 0 & 0 & 0 & 0 & 0 & 0 & 0 & 0 & 0 \\
 0 & 0 & 0 & 0 & 0 & 0 & 0 & 0 & 0 & 0 & 0 & 0 & 0 & 0 & 0 & 0 & 0 \\
 0 & 0 & 0 & 0 & 0 & 0 & 0 & 0 & 0 & 0 & 0 & 0 & 0 & 0 & 0 & 0 & 0 \\
 0 & 0 & 0 & 0 & 0 & 0 & 0 & 0 & 0 & 0 & 0 & 0 & 0 & 0 & 0 & 0 & 0 \\
 0 & 0 & 0 & 0 & 0 & 0 & 0 & 0 & 0 & 0 & 0 & 0 & 0 & 0 & 0 & 0 & 0 \\
 0 & 0 & 0 & 0 & 0 & 0 & 0 & 0 & 0 & 0 & 0 & 0 & 0 & 0 & 0 & 0 & 0 \\
\end{array}
\right)$}
  \end{subfigure}

  \medskip

  \begin{subfigure}{0.42\textwidth}
    \centering
    \caption{$\Lambda = (2, 0, -2)$}\label{fig:tab3}
    \vspace{5pt}
    \adjustbox{max width=\linewidth}{$\left(
\begin{array}{ccccccccc}
 0 & 0 & 0 & 0 & 0 & 0 & 0 & 0 & 0 \\
 0 & 0 & 0 & 0 & 0 & 0 & 0 & 0 & 0 \\
 0 & 0 & 0 & 0 & 0 & 0 & 0 & 0 & 0 \\
 0 & 0 & 0 & 0 & 0 & 0 & 0 & 0 & 0 \\
 0 & 0 & 0 & 0 & 0 & 0 & 0 & 0 & 0 \\
 0 & 0 & 0 & 0 & 0 & 0 & 0 & 0 & 0 \\
 0 & 0 & 0 & 0 & 0 & 0 & 0 & 0 & 0 \\
 0 & 0 & 0 & 0 & 0 & 0 & 0 & 0 & 0 \\
 0 & 0 & 0 & 0 & 0 & 0 & 0 & 0 & 0 \\
\end{array}
\right)$}
  \end{subfigure}
  \hfill
  \begin{subfigure}{0.27\textwidth}
    \centering
    \caption{$\Lambda = (1,1,-2)$}\label{fig:tab4}
    \vspace{5pt}
    \adjustbox{max width=\linewidth}{$
    \left(
\begin{array}{ccccccc}
 0 & 0 & 0 & 0 & 0 & 0 & 0 \\
 0 & 0 & 0 & 0 & 0 & 0 & 0 \\
 0 & 0 & 0 & 0 & 0 & 0 & 0 \\
 0 & 0 & 0 & 0 & 0 & 0 & 0 \\
 0 & 0 & 0 & 0 & 0 & 0 & 0 \\
 0 & 0 & 0 & 0 & 0 & 0 & 0 \\
 0 & 0 & 0 & 0 & 0 & 0 & 0 \\
\end{array}
\right)$}
  \end{subfigure}
  \hfill
  \begin{subfigure}{0.24\textwidth}
    \centering
    \caption{$\Lambda = (2,-1,-1)$}\label{fig:tab5}
    \vspace{5pt}
    \adjustbox{max width=\linewidth}{$
    \left(
\begin{array}{ccccccc}
 0 & 0 & 0 & 0 & 0 & 0 & 0 \\
 0 & 0 & 0 & 0 & 0 & 0 & 0 \\
 0 & 0 & 0 & 0 & 0 & 0 & 0 \\
 0 & 0 & 0 & 0 & 0 & 0 & 0 \\
 0 & 0 & 0 & 0 & 0 & 0 & 0 \\
 0 & 0 & 0 & 0 & 0 & 0 & 0 \\
 0 & 0 & 0 & 0 & 0 & 0 & 0 \\
\end{array}
\right)$}
  \end{subfigure}

  \medskip

  \begin{subfigure}{0.2\textwidth}
    \centering
    \caption{$\Lambda = (3, 0, -3)$}\label{fig:tab6}
    \vspace{5pt}
    \adjustbox{max width=\linewidth}{$\left(
    \begin{array}{c}
     0 \\
    \end{array}
    \right)$}
  \end{subfigure}
  \hfill
  \begin{subfigure}{0.2\textwidth}
    \centering
    \caption{$\Lambda =(2,1,-3)$}\label{fig:tab7}
    \vspace{5pt}
    \adjustbox{max width=\linewidth}{$
    \left(
\begin{array}{cc}
 0 & 0 \\
 0 & 0 \\
\end{array}
\right)$}
  \end{subfigure}
  \hfill
  \begin{subfigure}{0.2\textwidth}
    \centering
    \caption{$\Lambda =(3,-1,-2)$}\label{fig:tab8}
    \vspace{5pt}
    \adjustbox{max width=\linewidth}{$
    \left(
\begin{array}{cc}
 0 & 0 \\
 0 & 0 \\
\end{array}
\right)$}
  \end{subfigure}

  \caption{All irrep matrices for the operator $\rho(3)$ from the algebra $\mathcal{A}^3_{3,3}$ in the symmetric group adapted basis.}
  \label{fig:A333_rho_3_all_irreps}
\end{figure}

\newpage
\printbibliography
\end{document}